\documentclass[sigconf, nonacm]{acmart}

\newcommand\vldbdoi{XX.XX/XXX.XX}

\newcommand\vldbavailabilityurl{https://github.com/katekangKK/Auto-Relate.git}
\newcommand\vldbpagestyle{plain}

\usepackage{amsmath}
\usepackage{graphicx}
\usepackage{textcomp}
\usepackage{xcolor}
\usepackage{amsthm}
\usepackage{bbding}
\usepackage{xspace}
\usepackage{enumitem}
\usepackage{xcolor}
\usepackage{multirow}
\usepackage{tabularx}
\usepackage{hyperref}
\usepackage{dsfont}
\usepackage[ruled,vlined,linesnumbered]{algorithm2e}
\usepackage{diagbox}
\usepackage{subcaption}
\begin{document}

\title{Auto-Relate: A Unified Approach to Discovering Reliable Functional Relationships Leveraging Statistical Tests}

\newcommand{\kw}[1]{{\textsf{#1}}\xspace}
\newcommand{\eat}[1]{}
\newcommand{\tbf}{\kw{\textcolor{red}{XX}}}

\newcommand{\yeye}[1]{\textcolor{red}{#1}}

\newcommand{\indicator}{\ensuremath{\mathds{1}}\xspace}

\newcommand{\code}[1]{\texttt{\small #1}}
\newcommand{\codeq}[1]{{\tt {\small ``#1''}}}

\newcommand{\eop}{\hspace*{\fill}\mbox{$\Box$}}     

\newcounter{example}
\renewcommand{\theexample}{\arabic{example}}
\renewenvironment{example}{
        \vspace{1.2ex}
        \refstepcounter{example}
        {\noindent\bf Example \theexample:}}{\eop\vspace{1.2ex}}

\newcounter{definition}
\renewcommand{\thedefinition}{\arabic{definition}}
\renewenvironment{definition}[1][]{
    \vspace{1.5ex}
    \refstepcounter{definition}
    {\noindent\bf Definition {\bf \thedefinition}%
    \ifx&#1&\else\ (#1)\fi:}}

\renewenvironment{theorem}[1][]{\begin{em}
        \refstepcounter{theorem}
        {\vspace{1ex} \noindent\bf Theorem \thetheorem%
        \ifx&#1&\else\ (#1)\fi:}}{
        \end{em}\eop\vspace{1ex}}

\renewenvironment{lemma}[1][]{\begin{em}
        \refstepcounter{theorem}
        {\vspace{1.5ex}\noindent\bf Lemma \thelemma%
        \ifx&#1&\else\ (#1)\fi:}}{
        \end{em}\eop\vspace{1.5ex}}
        
\renewenvironment{proof}{
        \vspace{1ex}
        {\noindent\bf Proof:}}{\eop\vspace{1ex}}
        
\newcommand{\autorelate}{\kw{Auto}-\kw{Relate}}
\newcommand{\FR}{\kw{FR}}
\newcommand{\FRs}{\kw{FRs}}

\newcommand{\AR}{\kw{AR}}
\newcommand{\ARs}{\kw{ARs}}
\newcommand{\ST}{\kw{ST}}
\newcommand{\STs}{\kw{STs}}
\newcommand{\FD}{\kw{FD}}
\newcommand{\FDs}{\kw{FDs}}

\newcommand{\AFD}{\kw{AFD}}
\newcommand{\AFDs}{\kw{AFDs}}
\newcommand{\CFD}{\kw{CFD}}
\newcommand{\CFDs}{\kw{CFDs}}
\newcommand{\MD}{\kw{MD}}
\newcommand{\MDs}{\kw{MDs}}
\newcommand{\REE}{\kw{REE}}
\newcommand{\REEs}{\kw{REEs}}
\newcommand{\DC}{\kw{DC}}
\newcommand{\DCs}{\kw{DCs}}

\newcommand{\prauc}{\kw{PR}-\kw{AUC}}
\newcommand{\realAR}{\kw{Real}-\kw{AR}}
\newcommand{\realST}{\kw{Real}-\kw{ST}}
\newcommand{\realFD}{\kw{Real}-\kw{FD}}
\newcommand{\rwd}{\kw{RWD}}

\newcommand{\warn}[1]{{\color{red}{#1}}}
\newcommand{\revise}[1]{{\color{blue}{#1}}}

\newcommand{\szu}{$^{1}$}
\newcommand{\msr}{$^{2}$}
\newcommand{\zhbit}{$^3$}
\newcommand{\sics}{$^{4}$}

\author{Ziyan Han\szu, Yeye He\msr$^*$, Shuyuan Kang\zhbit, Min Xie\sics, Weiwei Cui\msr, Song Ge\msr, Haidong Zhang\msr, Dongmei Zhang\msr, Surajit Chaudhuri\msr, Rui Mao\szu$^*$, Jianbin Qin\szu$^*$}


\affiliation{%
  \institution{
    \szu Shenzhen University \hspace{3.8ex}
    \msr Microsoft Research \hspace{3.8ex}
    \zhbit Beijing Institute of Technology, Zhuhai \\
    \sics Shenzhen Institute of Computing Sciences
  }
}

\email{
{hanzy, mao, qinjianbin}@szu.edu.cn, 
{yeyehe, weiweicu, songge, haizhang, dongmeiz, surajitc}@microsoft.com,
}
\email{
shuyuankang@bitzh.edu.cn,
xiemin@sics.ac.cn
}


\eat{
\settopmatter{authorsperrow=4}

\author{Ziyan Han}
\affiliation{\institution{Shenzhen University}}
\email{hanzy@szu.edu.cn}

\author{Yeye He}
\affiliation{\institution{Microsoft Research}}
\email{yeyehe@microsoft.com}

\author{Shuyuan Kang}
\affiliation{\institution{Beijing Institute of Technology, Zhuhai}}
\email{shuyuankang@bitzh.edu.cn}

\author{Min Xie}
\affiliation{\institution{Shenzhen Institute of Computing Sciences}}
\email{xiemin@sics.ac.cn}

\author{Weiwei Cui}
\affiliation{\institution{Microsoft Research}}
\email{weiweicu@microsoft.com}

\author{Song Ge}
\affiliation{\institution{Microsoft Research}}
\email{songge@microsoft.com}

\author{Haidong Zhang}
\affiliation{\institution{Microsoft Research}}
\email{haizhang@microsoft.com}

\author{Dongmei Zhang}
\affiliation{\institution{Microsoft Research}}
\email{dongmeiz@microsoft.com}

\author{Surajit Chaudhuri}
\affiliation{\institution{Microsoft Research}}
\email{surajitc@microsoft.com}

\author{Rui Mao$^*$}
\affiliation{\institution{Shenzhen University}}
\email{mao@szu.edu.cn}

\author{Jianbin Qin$^*$}
\affiliation{\institution{Shenzhen University}}
\email{qinjianbin@szu.edu.cn}
}

\begin{abstract}
Tables in spreadsheets, computational notebooks, and databases often contain rich inter-column relationships.
Yet these relationships are typically implicit and are often lost when tables are exported to standard formats.
Recovering them can benefit downstream tasks, including table understanding, data quality improvement, and provenance analysis. 
However, simply mining relationships that hold on an observed table is insufficient, as many are spurious due to coincidence, redundancy, or limited data diversity.

In this paper, we introduce \emph{functional relationships} (\FRs) as a unified notion for inter-column relationships in tables, subsuming arithmetic relationships, string transformations, and functional dependencies. 
We characterize \FR reliability through four complementary criteria: accuracy, atomicity, stability, and integrity. 
Guided by these criteria, we propose \autorelate, a mine-then-verify framework that first generates accurate candidate \FRs and then verifies the remaining reliability criteria through a Minimality Test, a Perturbation Test, and an Independence Test, respectively.
To further improve efficiency, we develop three optimization strategies, including 
a group-by lower bound for early rejection, a closed-form speed-up for arithmetic \FRs, and a binomial bound for statistically guided early termination.
We construct a large-scale benchmark suite from 58,679 real-world spreadsheets and relational tables, containing 6,414 ground-truth \FRs spanning all three \FR types. 
Extensive experiments against 18 baselines show that \autorelate consistently achieves the best performance, with an average \prauc of 0.87, 
59\% higher than the best competing baseline across all settings.\looseness=-1
\end{abstract}

\maketitle

\begingroup
\renewcommand\thefootnote{\fnsymbol{footnote}}
\footnotetext{* Corresponding authors}
\endgroup

\pagestyle{\vldbpagestyle}

\vspace{-1em}
\ifdefempty{\vldbavailabilityurl}{}{
\vspace{.3cm}
\begingroup\small\noindent\raggedright\textbf{PVLDB Artifact Availability:}\\
The source code and data have been made available at \url{\vldbavailabilityurl}.
\endgroup
}

\section{Introduction}
\label{sec:intro}

Tables are a fundamental data representation across spreadsheets (e.g., Excel or Google Sheets), computational notebooks (e.g., Jupyter), and relational databases.
In these environments, users routinely create derived columns from existing ones by writing spreadsheet formulas, data-transformation programs, or SQL expressions~\cite{gulwani2011automating,harris2011spreadsheet,he2018transform,yang2021auto}.
As a result, tables often contain rich inter-column relationships.
For example, numeric columns are often related by arithmetic formulas, textual columns are commonly related by string transformations, and categorical columns may exhibit deterministic dependency patterns.
Although ubiquitous in real-world tables, such relationships are often implicit rather than explicitly represented, and are typically not preserved when tables are exported to standard formats such as CSV or Parquet.
\looseness=-1

\eat{
Figures~\ref{fig:ex-AR} and~\ref{fig:ex-ST} illustrate two typical examples.
In Figure~\ref{fig:ex-AR}, relationships such as $\text{\code{F}} \times \text{\code{G}} = \text{\code{H}}$ and $\text{\code{H}} + \text{\code{K}} + \text{\code{N}} = \text{\code{O}}$ reveal how multiple columns are composed to produce derived numeric values.
In Figure~\ref{fig:ex-ST}, relationships such as concatenating columns \text{\code{B}}, \text{\code{C}}, and \text{\code{D}} into \text{\code{E}} capture string-level transformations.
These examples show that many table columns are not isolated attributes, but are tied together by deterministic or near-deterministic mappings.
}

\begin{figure*}
\centering
\includegraphics[width=\linewidth]{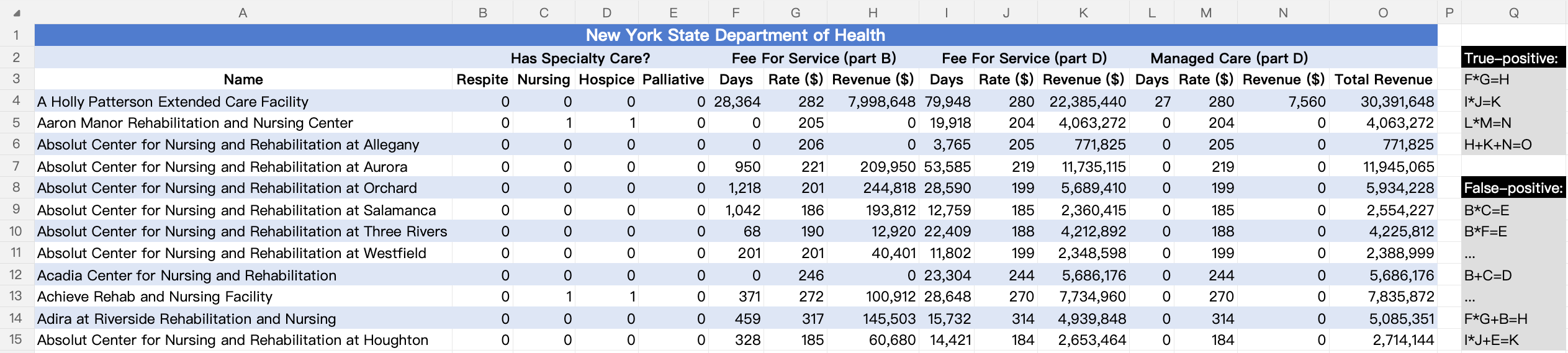}
\vspace{-0.5cm}
\caption{An example spreadsheet table with arithmetic relationships (\ARs).}
\label{fig:ex-AR}
\vspace{-0.1cm}
\end{figure*}

\begin{figure*}
\centering
\includegraphics[width=\linewidth]{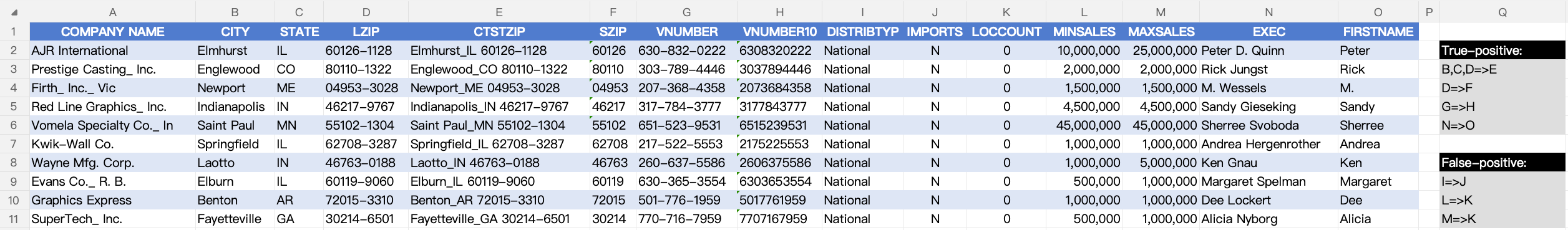}
\vspace{-0.5cm}
\caption{An example spreadsheet table with string transformations (\STs).}
\label{fig:ex-ST}
\vspace{-0.1cm}
\end{figure*}

Recovering such relationships from raw tables is practically valuable, benefiting a range of downstream applications, including table understanding~\cite{naumann2014data,limaye2010annotating}, data quality improvement~\cite{data-cleaning,data-cleaning-2}, data consistency maintenance~\cite{barowy2018excelint,panko1998we,koch2019metric,fan2008conditional}, and data provenance analysis~\cite{provenance-2, provenance-1}.
For table understanding, recovered relationships can reveal the latent structure of complex tables and help users interpret how columns are composed.
For data quality and consistency, they can serve as executable constraints for detecting, explaining, and repairing inconsistent values, and can support automatic propagation of updates to dependent cells or columns.
Moreover, they can aid provenance analysis by helping users trace how derived values are produced within or across tables.

However, discovering such relationships directly from table data is non-trivial. 
Simply finding relationships that hold on the observed table is insufficient.
Existing methods typically rely on observed-table signals such as exact satisfaction~\cite{huhtala1999tane,abedjan2014dfd}, low violation rate~\cite{kivinen1995approximate,giannella2004approximation,parciak2024measuring}, association strength~\cite{ilyas2004cords,piatetsky1993measuring}, or information-theoretic dependency scores~\cite{mandros2017discovering,pennerath2020discovering,parciak2024measuring}, but these signals alone cannot distinguish genuine relationships from spurious ones.
Because a table is only a finite sample of a larger underlying data space, many candidates may appear to hold on one table instance merely by coincidence, redundancy, or limited domain coverage, even though they are not meaningful or generalizable.
We refer to such coincidental, redundant, or incomplete candidates as \emph{spurious} \FRs.
This spuriousness problem becomes even more severe in dirty-data settings, where a candidate may satisfy an approximate criterion despite omitting essential attributes.

\begin{example}
Figures~\ref{fig:ex-AR}--\ref{fig:ex-FD} illustrate that genuine and spurious \FRs can coexist in the same table.
Although these tables contain genuine \FRs, a naive search may also produce several types of spurious candidates.
For example, Figure~\ref{fig:ex-AR} contains \emph{coincidental} candidates such as 
$\text{\code{B}} \times \text{\code{F}} = \text{\code{E}}$,
which hold on the observed rows but are unlikely to reflect genuine relationships,
as well as \emph{redundant} candidates such as 
$\text{\code{I}} \times \text{\code{J}} + \text{\code{E}} = \text{\code{K}}$,
where $\text{\code{E}}$ is extraneous and the real relationship is simply 
$\text{\code{I}} \times \text{\code{J}} = \text{\code{K}}$.
Similar spurious cases arise among string transformations in Figure~\ref{fig:ex-ST},
and Figure~\ref{fig:ex-FD} shows spurious \FRs such as $\text{\code{L}} \rightarrow \text{\code{M}}$, which hold only because of limited value diversity in the observed table.
These examples show that 
satisfying the observed table alone is insufficient for reliable \FR discovery.
\end{example}

\eat{
\begin{example}
Figures~\ref{fig:ex-AR}--\ref{fig:ex-FD} illustrate this challenge.
In Figure~\ref{fig:ex-AR}, relationships such as $\text{\code{F}} \times \text{\code{G}} = \text{\code{H}}$ and $\text{\code{H}} + \text{\code{K}} + \text{\code{N}} = \text{\code{O}}$ are genuine and meaningful, revealing how multiple columns are composed to produce derived numeric values.
At the same time, a naive search may also identify spurious \FRs, including coincidental ones such as $\text{\code{B}} \times \text{\code{F}} = \text{\code{E}}$, which holds on the observed table but is not semantically meaningful, and redundant ones such as $\text{\code{I}} \times \text{\code{J}} + \text{\code{E}} = \text{\code{K}}$, where \text{\code{E}} is extraneous and the real relationship is simply $\text{\code{I}} \times \text{\code{J}} = \text{\code{K}}$.
Similarly, in Figure~\ref{fig:ex-ST}, some string transformations are genuine, such as concatenating $\code{B}$, $\code{C}$, and $\code{D}$ into $\code{E}$, whereas others are spurious and fit only the currently observed values without generalizing beyond the table instance.
In Figure~\ref{fig:ex-FD}, genuine \FDs such as $\code{F} \rightarrow \code{G}$ appear alongside spurious ones such as $\code{L} \rightarrow \code{M}$, which hold only because of limited diversity in the observed values.
Naively using all such candidates can therefore mislead downstream tasks such as table understanding, data cleaning, and formula reconstruction.
\end{example}
}

Beyond these concrete examples, the key observation is that they all share the same functional nature: input values determine output values.
We therefore use the term \emph{functional relationships} (\FRs) to refer to such inter-column relationships in tables.
This unified notion highlights a common reliability challenge shared across different relationship types: how to distinguish meaningful relationships from spurious ones using only table data.

\eat{
Our key idea is that the reliability of \FRs cannot be captured by support alone.
Instead, a reliable \FR should satisfy four complementary properties:
(i) \emph{accuracy}, meaning that it is well supported by the observed table;
(ii) \emph{atomicity}, meaning that it does not contain redundant input columns;
(iii) \emph{stability}, meaning that it is non-trivial and does not survive arbitrary recombinations of observed values; and
(iv) \emph{integrity}, meaning that in dirty data its violations are not systematically explained by omitted non-participative columns.
These properties directly correspond to the major sources of spurious relationships described above.
}

Despite its practical importance, \FR discovery remains under-studied as a unified problem.
Accordingly, we formulate the problem as \emph{reliable functional relationship discovery}.
Since spurious \FRs can arise from different causes, reliability cannot be assessed from a single observed-table signal alone.
We therefore characterize \FR reliability through four complementary criteria: \emph{accuracy}, \emph{atomicity}, \emph{stability}, and \emph{integrity}.
Together, these criteria capture whether a candidate \FR is well supported by the observed table, free of redundant input columns, non-trivial under value recombination, and, in dirty-data settings, not systematically missing essential attributes.
\looseness=-1

Building on these reliability criteria, we propose \autorelate, a unified \emph{mine-then-verify} framework for discovering reliable \FRs by leveraging statistical tests.
\autorelate first generates candidate FRs that satisfy the accuracy requirement, and then verifies the remaining reliability criteria through dedicated tests.
A \emph{Minimality Test} removes non-atomic candidates with redundant input columns.
A \emph{Perturbation Test} evaluates stability by randomly recombining participative values across rows and measuring how likely the candidate is to survive; genuine, non-trivial relationships should usually be destroyed by such recombination, whereas spurious ones tend to survive.
In dirty-data settings, an additional \emph{Independence Test} examines whether the violation pattern of a candidate depends on columns outside its participative column set, thereby filtering incomplete candidates with omitted attributes.
We further develop efficient algorithms that exploit candidate-level pruning, deterministic lower bounds, and statistically guided early termination, enabling reliable predictions within seconds for interactive use cases.\looseness=-1
\eat{\autorelate first generates candidate \FRs that satisfy the accuracy requirement, then applies a \emph{Minimality Test} to remove non-atomic candidates and a \emph{Perturbation Test} to evaluate stability by randomly recombining participative values across rows. In dirty-data settings, it further applies an \emph{Independence Test} to 
filter incomplete candidates with omitted attributes. }

\smallskip
\noindent\textbf{Contributions \& Organizations.}
The main contributions and structure of this paper are summarized as follows.
\begin{itemize}[leftmargin=*]
    \item We introduce \emph{functional relationships} (\FRs), a unified notion for describing inter-column relationships in tables that subsumes arithmetic relationships (\ARs), string transformations (\STs), and functional dependencies (\FDs) 
    (Section~\ref{sec:definition}).

    \item We characterize the reliability of \FRs through four complementary properties, i.e., accuracy, atomicity, stability, and integrity, and use them to guide the design of \autorelate (Section~\ref{subsec:reliability}).

    \item We develop \autorelate, a unified \emph{mine-then-verify} framework for reliable \FR discovery (Section~\ref{subsec:overview}) that first generates accurate candidate \FRs and then verifies atomicity, stability, and integrity through a Minimality Test, a Perturbation Test, and an Independence Test, respectively (Section~\ref{subsec:minimality}--Section~\ref{subsec:independence}).

    \item We enhance \autorelate with three  optimization strategies, including 
    a group-by lower bound for early rejection, a closed-form speed-up strategy for AR-type \FRs, and a binomial bound for statistically guided early termination, substantially improving efficiency while providing theoretical guarantees (Section~\ref{sec:optimization}).\looseness=-1

    \item We built the first large-scale benchmark suite, collected from 58,679 real-world spreadsheets and relational tables, containing 6,414 real functional relationships that cover all three types of \FR. This allows us to systematically evaluate the \FR discovery problem, and can be a useful resource for future research. 
    (Section~\ref{subsec:exp_setting}).
    
    \item We conduct extensive experiments against 18 baselines across all \FR types and data settings. The results show that \autorelate consistently outperforms all baselines, achieving an average \prauc of 0.87 
    and improving over the strongest competing baseline by 59.0\% on average across all settings (Section~\ref{subsec:exp_result}).
\end{itemize}

In addition, we discuss the related work in Section~\ref{sec:related}, and conclude this
paper in Section~\ref{sec:conclusion}.

\eat{
For table understanding, access to functional relationships can help users interpret the internal structure of complex tables. 
For example, in Figure~\ref{fig:ex-AR}, the formulas $\text{\code{F}} \times \text{\code{G}} = \text{\code{H}}$, $\text{\code{I}} \times \text{\code{J}} = \text{\code{K}}$, $\text{\code{L}} \times \text{\code{M}} = \text{\code{N}}$, and $\text{\code{H}} + \text{\code{K}} + \text{\code{N}} = \text{\code{O}}$ immediately reveal three repeated local structures, which are then aggregated into the top-level column \code{O}.
Beyond table understanding, inferred \FRs can serve as rich data-quality constraints for detecting, explaining, and repairing inconsistent values. 
Once recovered as formulas or programs, they can also improve data consistency by automatically propagating updates to dependent cells or columns. 
For instance, in a spreadsheet such as Figure~\ref{fig:ex-AR}, editing a cell such as \codeq{F5} should trigger corresponding updates to dependent cells such as \codeq{H5} and \codeq{O5}.
Without such automatically recovered and enforced relationships, users often have to update dependent cells manually, which is error-prone.
Finally, recovered \FRs can support data provenance by helping users trace how derived columns are produced within or across tables, thereby facilitating applications such as impact analysis and data debugging~\cite{purview-lineage}.

A natural idea is therefore to mine \FRs directly from table data.
However, simply finding relationships that hold on the observed table is insufficient.
Existing methods typically guarantee only that a candidate satisfies some observed-table criterion, such as exact satisfaction or a low violation rate, but cannot distinguish genuine relationships from coincidental ones.
A table is only a finite sample of a larger underlying data space, and the values observed in its columns may cover only a limited portion of the true domains.
As a result, many candidates may appear to hold on one table instance even though they are not meaningful, not minimal, or not generalizable.
We refer to this as the \emph{spurious \FR phenomenon}: the candidate set may contain many \FRs that satisfy the mining criterion on the current table, yet do not correspond to genuine relationships.

This phenomenon is common in practice.
In Figure~\ref{fig:ex-AR}, a naive search may identify coincidental candidates such as $\text{\code{B}} \times \text{\code{F}} = \text{\code{E}}$, which holds on the observed table but is not semantically meaningful.
It may also identify redundant candidates such as $\text{\code{I}} \times \text{\code{J}} + \text{\code{E}} = \text{\code{K}}$, where column \text{\code{E}} is extraneous and the real relationship is simply $\text{\code{I}} \times \text{\code{J}} = \text{\code{K}}$.
Similarly, in Figure~\ref{fig:ex-ST}, one may obtain spurious string transformations that happen to fit the currently observed values but would fail on other instances.
The problem becomes even more severe in the dirty-data setting, where a candidate may still satisfy an approximate criterion despite omitting essential attributes.
For example, if violations are systematically associated with an unmodeled column, then the candidate may look plausible under a relaxed accuracy threshold while still being incomplete.
Naively using all such candidates can therefore mislead downstream tasks such as data analysis, data cleaning, and formula reconstruction.

These observations suggest that reliable \FR discovery requires more than ranking candidates by association strength or violation statistics; it requires explicitly modeling and verifying the reliability of candidate relationships.
Humans can often use semantic cues, domain knowledge, or column names to reject spurious candidates, but in many real scenarios such information may be missing, noisy, privacy-restricted, or simply unreliable.
Therefore, an automatic \FR discovery system should not aim merely to enumerate supported candidates; instead, it should identify \emph{reliable} \FRs that are useful, non-spurious, and likely to generalize beyond a single table instance.

To this end, we formulate the problem as \emph{reliable functional relationship discovery}.
Our key idea is that reliability cannot be captured by support alone.
Instead, a reliable \FR should satisfy four complementary properties:
(i) \emph{accuracy}, meaning that it is well supported by the observed table;
(ii) \emph{atomicity}, meaning that it does not contain redundant input columns;
(iii) \emph{stability}, meaning that it is non-trivial and does not survive arbitrary recombinations of observed values; and
(iv) \emph{integrity}, meaning that in dirty data its violations are not systematically explained by omitted non-participative columns.
These properties directly correspond to the major sources of spurious \FRs described above.

Based on this view, we propose \autorelate, a unified \emph{mine-then-verify} framework for discovering reliable \FRs.
\autorelate first generates \FR candidates that satisfy the accuracy requirement, and then subjects them to dedicated reliability checks.
A \emph{Minimality Test} removes non-atomic candidates with redundant input columns.
An \emph{Independence Test}, used in the dirty-data setting, examines whether the violation pattern of a candidate depends on columns outside its participative column set, thereby filtering incomplete candidates with omitted attributes.
A \emph{Perturbation Test} evaluates stability by randomly recombining participative values across rows and measuring how likely the candidate is to survive such perturbations: a genuine, non-trivial relationship should usually be destroyed by random recombination, whereas a spurious one tends to survive.
We further develop efficient algorithms that exploit 
deterministic lower bounds and statistically guided early termination, enabling reliable predictions within seconds for interactive use cases.
}

\section{Problem Definition}
\label{sec:definition}

\begin{figure*}
\centering
\includegraphics[width=\linewidth]{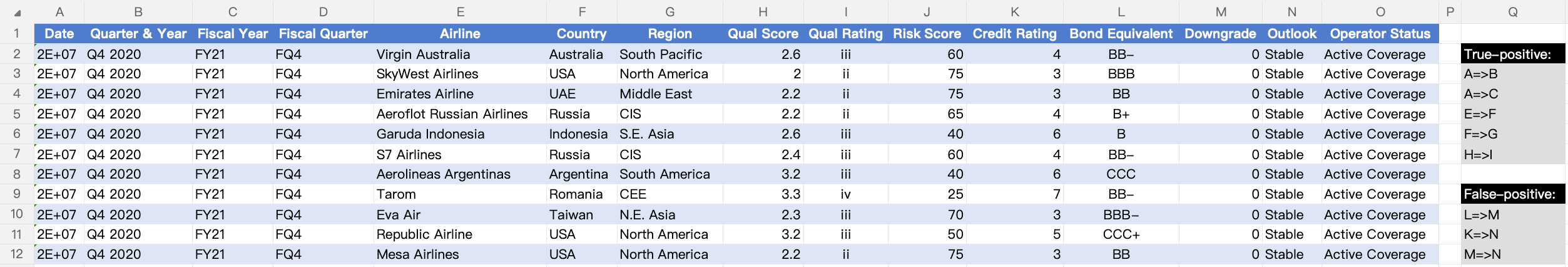}
\vspace{-0.5cm}
\caption{An example spreadsheet table with functional dependencies (\FDs).}
\label{fig:ex-FD}
\vspace{-0.1cm}
\end{figure*}

We use the \textit{Functional Relationship (\FR)} to refer to the mapping between different columns in tabular data, illustrating how input values (one or more columns in the data) are transformed into output values (other columns). Formally, we define \FRs as follows.

\begin{definition}[Functional Relationships]
    Given a table $T$ with a set of columns $C$, and column subsets $X, Y \subseteq C$. A \emph{Functional Relationship} $\Psi=(X,Y,f_r)$ is composed of input columns set $X$, output columns set $Y$ and a mapping function $f_r$.
\end{definition}

Let $V_T(X)$ and $V_T(Y)$ denote the domains of 
$X$ and $Y$ induced by table $T$, respectively. For all $x\in V_T(X) ~and~ y\in V_T(Y)$, it holds that 
\[
\Pr(Y = y \mid X = x) =
\begin{cases} 
1, & \text{if} ~y = f_r(x) \\
0, & \text{otherwise}
\end{cases}
\]
A \FR $\Psi=(X,Y,f_r)$ is \textit{non-trivial} if $Y \not\subseteq X$, and \textit{normalized} if $Y$ contains a single column. 
We denote $C_\Psi = X \cup Y$ as the \textit{participative column set} of \FR $\Psi$. 
For two \FRs $\Psi_1$ and $\Psi_2$, they are said to be \textit{isomerous} if $C_{\Psi_1} = C_{\Psi_2}$. 
In this paper, we consider only non-trivial, normalized \FRs, as they suffice to infer all other \FRs that hold on table $T$.
\looseness=-1

\smallskip
\noindent\textbf{Semantics.}
Given a tuple $t \in T$, we say that $t$ \emph{satisfies}
an \FR $\Psi=(X,Y,f_r)$, written as $t \models \Psi$,
if $f_r(t[X]) = t[Y]$; otherwise, $t$ is said to \emph{violate} $\Psi$, written as $t \not\models \Psi$.

For a table $T$, we say that $T$ \emph{satisfies} $\Psi$,
denoted as $T \models \Psi$, if $t \models \Psi$ for all $t \in T$.
The set of observed violating tuples of $\Psi$ on $T$ is
$\mathrm{vio}_T(\Psi)=\{t\in T \mid t \not\models \Psi\}$,
and the satisfaction ratio of $\Psi$ on $T$ is 
\[
\mathrm{sat}_T(\Psi) = \frac{|\{t\in T \mid t \models \Psi\}|}{|T|}.
\]
Equivalently, $T \models \Psi$ iff $\mathrm{sat}_T(\Psi)=1$,
or iff $\mathrm{vio}_T(\Psi)=\emptyset$.

\begin{definition}[Approximate \FRs]
\label{def:approx-fr}
Given a threshold $\tau \in (0, 1]$,
an \FR $\Psi = (X, Y, f_r)$ is said to be
\emph{$\tau$-approximate} on $T$ if
$\mathrm{sat}_T(\Psi) \geq \tau$,
i.e., $|\mathrm{vio}_T(\Psi)| \leq (1-\tau)|T|$.
When $\tau = 1$, a $\tau$-approximate \FR reduces to an exact \FR,
i.e., $T \models \Psi$.
\end{definition}

\vspace{0.5ex}
We next introduce three common types of \FRs as follows. 

\emph{\underline{Arithmetic Relationships (\ARs).}}
When input and output columns are numeric, we refer to the \FR as an 
\AR,
where the mapping function $f_r$ corresponds to an arithmetic formula involving operations such as addition, subtraction, multiplication, and division.

\emph{\underline{String Transformations (\STs).}}
When input and output columns are textual, we refer to the \FR as an \ST, where $f_r$ corresponds to a string operation such as concatenation, substring extraction, or splitting~\cite{bogatu2019towards,he2018transform,yang2021auto,harris2011spreadsheet,jin2017foofah,jin2018clx}.

\emph{\underline{Functional Dependencies (\FDs).}}
When both the input and output 
are single categorical columns,
we refer to the \FR as an \FD, a well-studied form of consistency constraint~\cite{caruccio2015relaxed,berti2018discovery}.
Although \FDs can express deterministic relationships over numeric or textual values, they capture value determinacy rather than the mapping function provided by \ARs and \STs.
We therefore reserve \FDs for the single-column categorical setting, 
where $f_r$ is implicitly induced by the dependency: identical input 
values imply identical output values.\looseness=-1
\eat{
Broadly speaking, deterministic relationships over numeric and textual values can also be expressed as \FDs. 
However, compared with \ARs and \STs, \FDs primarily emphasize the determining role of the input columns on the output columns and are less expressive in characterizing the underlying mapping function.
Therefore, to preserve the stronger expressiveness of \ARs and \STs where applicable, we use \FDs only in this single-column categorical setting. 
In this case, the mapping function $f_r$ is not represented as an explicit arithmetic formula or string transformation, but is implicitly induced by the dependency: identical input values imply identical output values.
}

\begin{example}
\label{exa:FRs}
Figures~\ref{fig:ex-AR}--\ref{fig:ex-FD} illustrate the three types of \FRs in real spreadsheet tables. 
In each figure, true-positive (TP) relationships are genuine, whereas false-positive (FP) relationships are spurious.

\begin{itemize}[leftmargin=*]
\item \textbf{\ARs} (Figure~\ref{fig:ex-AR}): 
$\Psi_1=(\{F,G\}, H, F \times G)$ captures that $H$ (\kw{Revenue}) equals 
$F$ (\kw{Days}) $\times$ $G$ (\kw{Rate}), and $\Psi_2=(\{H,K,N\},$ $O, H+K+N)$ captures that column $O$ (\kw{Total Revenue}) is obtained by summing the three partial revenues in $H$, $K$, and $N$. 

\item \textbf{\STs} (Figure~\ref{fig:ex-ST}): 
$\Psi_3=(\{B,C,D\}, E, \textit{concat}(B,C,D))$ concatenates 
$B$ (\kw{CITY}), $C$ (\kw{STATE}), and $D$ (\kw{LZIP}) into $E$ (\kw{CTSTZIP}), while 
$\Psi_4=(\{N\}, O, \textit{extract}(N,\kw{space}))$ extracts the first name in $O$ (\kw{FIRSTNAME}) from the full name in $N$ (\kw{EXEC}).

\item \textbf{\FDs} (Figure~\ref{fig:ex-FD}): 
$\Psi_5 = (\{A\}, B, A \to B)$ captures that $B$ (\kw{Quarter \& Year}) is 
determined by $A$ (\kw{Date}), and $\Psi_6 = (\{F\}, G, F \to G)$ 
captures that $G$ (\kw{Region}) is determined by $F$ (\kw{Country}). \looseness=-1 
\end{itemize}
\vspace{-1.0\baselineskip}
\end{example}

\eat{
\begin{figure*}
\centering
\includegraphics[width=\linewidth]{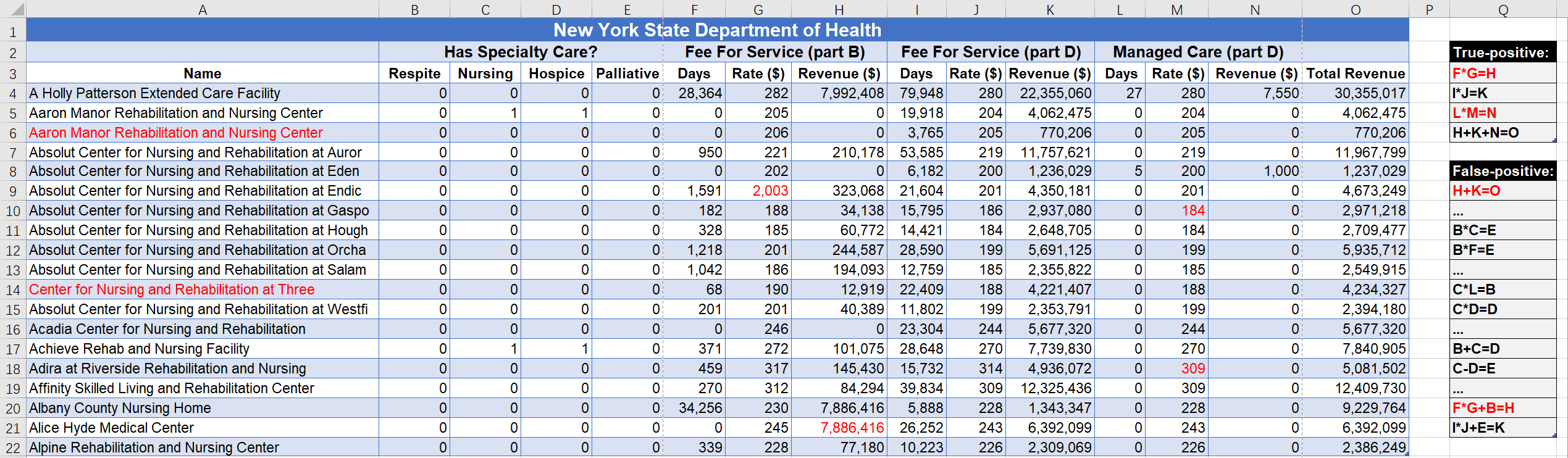}
\caption{A noisy example spreadsheet table with arithmetic relationship (ARs).}
\label{fig:dirty-AR}
\end{figure*}

\begin{example}
\label{eg:dirty-AR}
Figure~\ref{fig:dirty-AR} is the dirty data version of Figure~\ref{fig:ex-AR}. Noise cells and FRs with violation rows are highlighted in red. Considering $\alpha = 0.1$, the $\alpha$-FRs in Figure~\ref{fig:dirty-AR} have at most two violation rows. The FR $\Psi_H=(\{F,G\},F*G)$ is a reliable $\alpha$-FR, failing on the ninth and twenty-first rows. Additionally, due to the relaxed constraint, the unreliable FR $\Psi_O=(\{H,K\},H+K)$ is included in the $\alpha$-FRs set, despite failing on the fourth and eighth rows.
\end{example}
}

\eat{
\begin{figure}
\centering
\includegraphics[width=\columnwidth]{figures/FR discover and measure flow chart.pdf}
\vspace{-0.5cm}
\caption{Functional relationships discovery workflow.}
\vspace{-0.1cm}
\label{fig:FlowChart}
\end{figure}
}

\noindent
\textbf{Problem statement.}
Our task is to discover reliable functional relationships
(\FRs) from a table.

\begin{itemize}[leftmargin=*]
\item \textbf{Input:} A table $T$. 
\item \textbf{Output:} A set $R$ of reliable \FRs discovered from $T$.
\end{itemize}

To this end, we design a unified framework that first generates candidate \FRs from $T$ and then evaluates the reliability of each candidate \FR, thereby filtering out
spurious \FRs and returning a set $R$ of reliable \FRs.
We will formally define the reliability 
criteria for \FRs and
present the framework in Section~\ref{sec:auto-relate}.

\section{The \autorelate Framework}
\label{sec:auto-relate}

In this section, we present \autorelate, a unified \emph{mine-then-verify} framework for discovering reliable \FRs. 
We first introduce four reliability criteria for characterizing meaningful
\FRs 
(Section~\ref{subsec:reliability}).
We then provide an overview of the framework
(Section~\ref{subsec:overview}) and detail its three reliability tests:
the Minimality Test, the Perturbation Test, and the Independence Test
(Sections~\ref{subsec:minimality}--\ref{subsec:independence}).
\eat{We first formalize the reliability criteria of FRs,
including accuracy, atomicity, stability, and integrity, which together characterize the practical meaningfulness of discovered \FRs (Section~\ref{subsec:reliability}).
We then present an overview of the \autorelate framework, which consists of two components: a candidate-generation stage that enumerates FR candidates satisfying
the accuracy requirement, and a reliability-verification stage that further filters these candidates through minimality checking and statistical tests (Section~\ref{subsec:overview}). 
Finally, we present the two core statistical tests in \autorelate: the Perturbation Test
(Section~\ref{subsec:perturbation}) and the Independence Test (Section~\ref{subsec:independence}).
}

\subsection{Reliability Criteria for \FRs}
\label{subsec:reliability}
\eat{
The goal of \autorelate is not merely to enumerate candidate \FRs, but to identify \emph{reliable} \FRs that are useful,
non-spurious, and generalizable beyond a single table instance.
To make this goal explicit, we characterize the reliability of an \FR
through four desirable properties: \emph{accuracy}, \emph{atomicity},
\emph{stability}, and \emph{integrity}. Together, these properties
define the practical meaningfulness of an \FR and motivate the design of the statistical tests in our \autorelate framework.
}

\autorelate aims to identify reliable \FRs that are non-spurious and 
generalizable beyond a single table instance, rather than merely enumerating candidates.
We characterize \FR reliability through four properties: \emph{accuracy}, \emph{atomicity}, \emph{stability}, and \emph{integrity}, which capture complementary aspects of \FR reliability and motivate the statistical tests in \autorelate.

\begin{definition}[Accuracy]
\label{def:accuracy}
An \FR $\Psi=(X,Y,f_r)$ is said to be \emph{accurate} on a table $T$
if $T \models \Psi$, i.e., $\operatorname{sat}_T(\Psi)=1$.
In dirty-data settings, exact satisfaction may be too restrictive; we therefore relax this requirement and say that $\Psi$ is
\emph{$\tau$-accurate} on $T$ if it is a
\emph{$\tau$-approximate} \FR on $T$, 
i.e., $\operatorname{sat}_T(\Psi)\ge\tau$.
\end{definition}

\begin{definition}[Atomicity]
\label{def:atomicity}
An \FR $\Psi = (X, Y, f_r)$ is said to be \emph{atomic} on a table $T$ if there exists no other FR $\Psi' =(X', Y, f_r')$ such that
$X' \subseteq X$ and $\Psi$ is accurate on $T$.
That is, no proper subset of the input columns $X$ of $\Psi$ can determine the same output columns $Y$ on $T$.\looseness=-1
\end{definition}

\begin{definition}[Stability]
\label{def:stability}
An \FR $\Psi = (X, Y, f_r)$ is said to be \emph{stable} on a table $T$ if there exists at
least one violating combination in the induced domain, i.e.,
$\exists x \in V_T(X), \exists y \in V_T(Y) \text{ s.t. } 
y \neq f_r(x)$.
Equivalently, the valid mappings of $f_r$ are a strict subset of the Cartesian product
$V_T(X)\times V_T(Y)$.
\end{definition}

\eat{
\begin{definition}[Integrity]
\label{def:integrity}
Let $Z = C \setminus C_\Psi$ be the set of non-participative
columns of an \FR $\Psi = (X, Y, f_r)$ in a table $T$ with a set of columns $C$, and let $t$ be a  tuple sampled uniformly at random from $T$.
We define the violation indicator of $\Psi$ on $t$ as
$I_\Psi(t)=\indicator[ t \not\models \Psi ]$.
We say that $\Psi$ satisfies \emph{integrity} if
$I_\Psi(t) \perp\!\!\!\perp t[Z]$.
That is, whether $t$ satisfies or violates $\Psi$ is independent of the columns that do not participate in $\Psi$.
\end{definition}
}

\begin{definition}[Integrity]
\label{def:integrity}
Let $Z = C \setminus C_\Psi$ be the set of non-participative columns of an \FR $\Psi = (X, Y, f_r)$ in a table $T$ with a set of columns $C$. 
For a tuple $t$ sampled uniformly at random from $T$, let 
$I_\Psi(t)=\indicator[t \not\models \Psi]$ denote the violation indicator of $\Psi$.
We say that $\Psi$ satisfies \emph{integrity} if 
$I_\Psi(t) \perp\!\!\!\perp t[c]$ for every $c \in Z$.
That is, the violation pattern of $\Psi$ is independent of every non-participative column.\looseness=-1
\end{definition}

\begin{definition}[Reliability]
An \FR is said to be \emph{reliable} if it satisfies accuracy,
atomicity, stability, and integrity.
\end{definition}

\eat{
\begin{example}
\label{exa:reliability}
Continuing with \FRs in Example~\ref{exa:FRs}, we illustrate why reliability requires more than observed-table accuracy alone.
The candidate $\Psi_1=(\{F,G\}, H, F \times G)$ is accurate on the table in Figure~\ref{fig:ex-AR}, since each observed tuple satisfies $H=F\times G$.
Consider a spurious candidate 
$\Psi_1^{+}=(\{B,F,G\}, H, F \times G + B)$.
Although $\Psi_1^{+}$ is still accurate, it is not atomic, because the smaller candidate $\Psi_1$ already determines $H$ and column $B$ is redundant.
The FR $\Psi_1$ is also stable, since random recombination of the observed values in $F$, $G$, and $H$ will usually destroy the equality $H=F\times G$.
By contrast, a spurious candidate 
$\Psi_\mathrm{coin}=(\{B,C\}, E, B\times C)$  may hold only by coincidence:
because $B$, $C$, and $E$ are sparse binary columns, most tuples satisfy
$B \times C = E = 0$.
As a result, random perturbation often fails to destroy this candidate,
showing that its apparent validity is accidental rather than due to a genuine underlying relationship.
In a dirty-data variant of the table in Figure~\ref{fig:ex-AR}, consider an incomplete candidate 
$\Psi_2^{-}=(\{H,K\}, O, H+K)$ derived from the genuine FR $\Psi_2=(\{H,K,N\}, O, H+K+N)$.
Although $\Psi_2^{-}$ may still be $\tau$-accurate, its violations
are not random: they occur precisely on rows where $N \neq 0$, i.e., where
the omitted term makes $H+K \neq O$.
This dependence on $N$ indicates that $\Psi_2^{-}$ violates integrity.
\end{example}
}

\begin{example}
\label{exa:reliability}
Continuing with \FRs in Example~\ref{exa:FRs}, we illustrate why
observed-table accuracy alone is insufficient.
The genuine \FR $\Psi_1=(\{F,G\}, H, F \times G)$ is accurate on
the table in Figure~\ref{fig:ex-AR} and is stable because random 
recombination of $F$, $G$, and $H$ will usually destroy the equality 
$H=F\times G$.
By contrast, $\Psi_1^{+}=(\{B,F,G\}, H, F \times G + B)$ is also 
accurate but violates atomicity, since the smaller \FR $\Psi_1$ 
already determines $H$ and $B$ is redundant.
Another candidate, $\Psi_\mathrm{coin}=(\{B,C\}, E, B\times C)$,
is unstable: $B$, $C$, and $E$ are sparse binary columns, so most 
tuples trivially satisfy $B \times C = E = 0$ and $\Psi_\mathrm{coin}$
often survives perturbation.
Finally, in a dirty-data variant of Figure~\ref{fig:ex-AR}, the
incomplete candidate
$\Psi_2^{-}=(\{H,K\}, O, H+K)$ may still be $\tau$-accurate, but its
violations occur on rows where the omitted term $N$ is nonzero.
This dependence on $N$ indicates that $\Psi_2^{-}$ violates integrity.
\end{example}

Among these criteria, integrity is meaningful primarily in dirty-data settings, where an approximately accurate candidate may 
omit essential attributes.
In clean-data settings, accurate candidates have no observed violations from which such omissions can be diagnosed.
We next describe how \autorelate operationalizes these 
criteria.
\looseness=-1

\eat{
These four criteria directly motivate the design of \autorelate. 
Accuracy is ensured by the upstream candidate-generation stage, which retains 
only accurate or $\tau$-accurate candidates. The remaining criteria—atomicity, stability, and integrity—are then verified by the \kw{MinimalityTest} (Section~\ref{subsec:minimality}), \kw{PerturbationTest} (Section~\ref{subsec:perturbation}), and \kw{IndependenceTest} (Section~\ref{subsec:independence}), respectively.
}

\eat{
Intuitively, the four properties serve different purposes.
Accuracy ensures that an \FR is supported by the observed data, while atomicity removes candidates with redundant input columns.
Stability requires an \FR to be non-trivial over the induced domain, filtering out coincidental relationships. 
Finally, integrity is mainly meaningful in dirty-data settings, where an approximately accurate candidate may still omit essential attributes.
In clean-data settings, by contrast, accurate candidates have no observed violations, leaving no violation pattern for diagnosing 
such omissions.\looseness=-1

These criteria directly motivate the design of \autorelate.
Accuracy is ensured by the upstream \FR mining process, which outputs only candidates with low violation rates.
Atomicity is enforced by the \kw{MinimalityTest}, which identifies and prunes candidates with redundant input columns.
The remaining reliability properties are addressed by dedicated statistical tests in \autorelate:
\begin{itemize}[leftmargin=*]
\item \textbf{Minimality Test} (Section~~\ref{subsec:minimality})

\item \textbf{Perturbation Test} 
(Section~\ref{subsec:perturbation}) verifies stability by 
perturbing table values and testing whether the candidate FR remains valid under random perturbations.

\item \textbf{Independence Test} 
(Section~\ref{subsec:independence}) verifies integrity in dirty-data settings by examining whether the violation pattern of candidate FR depends on columns outside its participative column set $C_\Psi$. \looseness=-1
\end{itemize}
}

\subsection{Overview of \autorelate}
\label{subsec:overview}
Given an input table $T$, \autorelate aims to discover a set of reliable \FRs from $T$.
It mainly consists of two stages: 
\emph{candidate generation} and \emph{reliability verification}.
The candidate-generation stage enumerates candidate \FRs that satisfy the accuracy requirement, 
while the reliability-verification stage further examines these candidates through three tests: the \kw{MinimalityTest} for atomicity (Section~\ref{subsec:minimality}), the \kw{PerturbationTest} for stability (Section~\ref{subsec:perturbation}), and the \kw{IndependenceTest} for integrity in the dirty-data setting (Section~\ref{subsec:independence}).

\eat{
The overall procedure of \autorelate is summarized in
Algorithm~\ref{alg:Auto-Relate}.
The algorithm operates in a \emph{levelwise search} manner, interleaving
\emph{candidate generation} and \emph{reliability verification} at each level,
so that verification results from earlier levels can guide and
prune candidate enumeration at later levels.
}

\eat{
The overall procedure of \autorelate is summarized in Algorithm~\ref{alg:Auto-Relate}. 
Given an input table $T$, \autorelate aims to discover a set of reliable FRs from $T$ in a level-wise manner, interleaving the following two stages at each level.

\smallskip
\noindent
\textbf{Candidate generation.}
\autorelate first invokes a type-specific candidate generation procedure
to obtain the candidate FR set $R_C$ from $T$.
This stage is responsible for ensuring \emph{accuracy} and
\emph{atomicity}: it retains only candidates whose satisfaction ratio
meets the accuracy threshold, and enumerates from smaller to larger
input column sets so that redundant candidates are pruned early.

\smallskip
\noindent\textbf{Reliability verification.}
For each candidate $\Psi \in R_C$, \autorelate performs reliability
verification to distinguish genuine FRs from spurious ones.
Let $C_{\Psi}$ denote the participative column set of $\Psi$.
The core of this stage is the Perturbation Test, which evaluates
\emph{stability} by testing whether $\Psi$ is non-trivial over the
induced domain.
Specifically, it measures how likely $\Psi$ is to remain valid after
random cross-row perturbations on $C_{\Psi}$.
Intuitively, a genuine FR represents a strict, non-trivial mapping that
depends on the exact alignment of row values.
As a result, it holds only for a proper subset of the possible
recombinations of observed participative values, and is therefore
expected to be violated by random perturbations with high probability,
yielding a low survival rate.
By contrast, a spurious FR that appears valid only because the observed data cover a limited portion of the domain (e.g., a coincidental mapping
to a constant) lacks this strict dependence and is insensitive to data
shuffling; thus, it is more likely to survive perturbation by chance and thus receive a high score.

In the dirty-data setting ($\delta = \text{dirty}$), \autorelate
additionally performs an Independence Test \emph{prior to} perturbation scoring to evaluate \emph{integrity}.
Specifically, it checks whether the violation pattern of $\Psi$ depends on any non-participative columns.
If such dependence is detected, the candidate is discarded as incomplete.
Together, the two tests address the remaining reliability criteria: the Perturbation Test targets \emph{stability}, while the Independence Test targets \emph{integrity}.

A candidate is admitted into the reliable FR set $R$ only if its
perturbation score satisfies $p \leq \eta$, where $\eta$ is a
user-specified perturbation threshold; in the dirty-data setting, the
candidate must additionally pass the Independence Test.
The details of the two tests are presented below.
}

\begin{algorithm}[!ht]
\SetKwInOut{Input}{Input}
\SetKwInOut{Output}{Output}

\Input{A table $T$, a data setting $\delta \in \{\text{clean}, \text{dirty}\}$, 
an accuracy threshold $\tau$, a significance threshold $\alpha$, a perturbation threshold $\eta$, 
a maximum number of input columns (i.e., the search level)
$\ell_{\max}$.}
\Output{A set $R$ of reliable \FRs discovered from $T$.}

$R \gets \emptyset$\;

$R_\kw{prune} \gets \emptyset$\; 

\For{$i \gets 1$ \KwTo $\ell_{\max}$}{
    $R_\kw{cand} \gets \kw{CandidateGeneration}(T, \delta, \tau, i)$\;
    \lIf{$R_\emph{\kw{cand}} = \emptyset$}{\textbf{continue}}
    
    \ForEach{$\Psi \in R_\emph{\kw{cand}}$}{
        \If{\textbf{not} \emph{\kw{MinimalityTest}}$(\Psi, R_\emph{\kw{prune}})$}{\textbf{continue}}
        
        \If{$\delta = \text{dirty}$}{
            \lIf{\textbf{not} \emph{\kw{IndependenceTest}}$(T, \Psi, \alpha)$}{\textbf{continue}}
        }

        $R_\kw{prune} \gets R_\kw{prune} \cup \{\Psi\} $\;
        
        $p_\kw{pert} \gets \kw{PerturbationTest}(T, \Psi, \eta)$\;
        
        \If{$p_\emph{\kw{pert}} \le \eta$}{
            $R \gets R \cup \{\Psi\}$\;
        }
    }
}

\Return{$R$}\;

\caption{Auto-Relate}
\label{alg:Auto-Relate}
\end{algorithm}

\smallskip
\noindent\textbf{Algorithm.}
As shown in Algorithm~\ref{alg:Auto-Relate}, \autorelate operates in a \emph{level-wise} manner.
It first initializes two sets: (i) $R$, the set of reliable \FRs to be returned, and (ii) $R_{\text{prune}}$, the set of previously retained candidates maintained across levels for atomicity-based pruning of subsequent candidates (lines~1--2).
It then proceeds level by level from $i = 1$ to $\ell_{\max}$, interleaving \emph{candidate generation} and \emph{reliability verification} at each level (line~3--14).
At each level $i$, \autorelate invokes \kw{CandidateGeneration} to obtain the candidate \FR set $R_{\kw{cand}}$ (line~4); if no candidates are produced, the level is skipped (line~5).
For each candidate $\Psi \in R_{\kw{cand}}$ (line~6), it first applies \kw{MinimalityTest} to check whether $R_{\kw{prune}}$ already contains a candidate whose input column set is a proper subset of that of $\Psi$ (line~7); if so, $\Psi$ is non-atomic and discarded.
In the dirty-data setting, the surviving candidate is further checked via \kw{IndependenceTest} (lines~9--10) for integrity; a candidate that fails this test is rejected as incomplete.
Each candidate that passes the above filters is added to $R_{\kw{prune}}$ (line~11), so that it can prune supersets at subsequent levels.
Finally, \kw{PerturbationTest} is applied to compute the perturbation score $p_\kw{pert}$ (line~12); if $p_\kw{pert} \le \eta$, the candidate is admitted into the reliable FR set $R$ (lines~13--14).
After all levels have been processed, \autorelate returns $R$ as the set of reliable FRs discovered from $T$ (line~15).

\smallskip
\noindent\textbf{Complexity.}
Let $R_C$ be the set of all candidate FRs generated across all levels.
For a candidate \FR $\Psi$, let $\mathcal{T}_{\mathrm{mt}}(\Psi)$,
$\mathcal{T}_{\mathrm{it}}(\Psi)$, and $\mathcal{T}_{\mathrm{pt}}(\Psi)$
denote the time of \kw{MinimalityTest}, \kw{IndependenceTest}, and
\kw{PerturbationTest}, respectively.
Then the total running time of \autorelate is the candidate-generation cost plus at most
$O(\sum_{\Psi \in R_C}
(\mathcal{T}_{\mathrm{mt}}(\Psi)
+
\indicator[\delta=\text{dirty}] \cdot \mathcal{T}_{\mathrm{it}}(\Psi)
+
\mathcal{T}_{\mathrm{pt}}(\Psi)))$.
Note that this is the worst-case complexity;
in practice, our algorithm is much faster by effective pruning and optimization strategies introduced in Section~\ref{sec:optimization}.

\eat{
\smallskip
\noindent\textbf{Remark.}
In practice, users need not know in advance whether the input table is clean or dirty.
The parameter $\delta$ determines which verification branch is activated: the clean-data setting follows a simplified path with only the Perturbation Test, while the dirty-data setting invokes a more comprehensive procedure with an additional Independence Test.
When such prior knowledge is unavailable, one may use the dirty-data setting as a practical default, since it provides a more conservative verification path.
}

\eat{
\begin{example}
$\Psi_H=(\{F,G,B\},F*G+B)$ in Figure~\ref{fig:ex-AR} can be decomposed into $\Psi_H=(\{F,G\},F*G)$, with column $B$ consisting entirely of zeros. The \SubformulaTest~ sets the  score \ReliabilityScore\ to 1 to filter out the redundant FR $\Psi_H=(\{F,G,B\},F*G+B)$.
For FRs that passed \SubformulaTest~, \autorelate uses \HTOne~ to calculate their scores.
\end{example}
}

\subsection{Minimality Test}
\label{subsec:minimality}
The Minimality Test is the first verification step in Algorithm~\ref{alg:Auto-Relate}.
It enforces the \emph{atomicity} criterion by pruning redundant candidates. 
Recall from Definition~\ref{def:atomicity} that a candidate $\Psi=(X,Y,f_r)$ is non-atomic 
if there exists another FR $\Psi'=(X',Y,f'_r)$ with $X' \subset X$ and $\mathrm{sat}_T(\Psi') \ge \tau$.
In this case, $Y$ can already be determined with sufficient accuracy by a strict subset of $X$, 
and thus $\Psi$ can be 
discarded.
\looseness=-1

To operationalize this idea, \autorelate performs candidate enumeration in a \emph{level-wise} manner, processing all candidates with $|X|=i$ before those with $|X|=i+1$.
For each candidate $\Psi=(X,Y,f_r)$ at level $i$, the \kw{MinimalityTest} checks whether the pruning set $R_{\kw{prune}}$ already contains a previously retained candidate $\Psi'=(X',Y,f'_r)$ with the same output column $Y$ and a strict subset of input columns, i.e., $X' \subset X$.
If so, $\Psi$ is immediately pruned and does not enter the remaining verification procedures.
Since pruning decisions at earlier levels propagate to later levels, this test can eliminate many redundant supersets, especially on wide tables.
\eat{Because the enumeration proceeds level by level, pruning decisions made at level $i$ naturally propagate to later levels, thereby eliminating potentially large numbers of supersets at level $i+1$ and beyond.
This strategy is especially effective for wide tables, where the number of candidate column subsets grows exponentially.

A further practical advantage of the level-wise structure is that candidates within the same level can be processed independently, since the outcome for one candidate does not affect any other candidate at that level.
We exploit this property in our implementation using multiprocessing.
}

\eat{
\begin{example}
\label{exa:minimality}
Consider the spurious candidate
$\Psi_1^{+}=(\{B,F,G\}, H,$ $F\times G+B)$
from Example~\ref{exa:reliability}.
The Minimality Test checks whether some proper subset of its input columns already yields an accurate FR.
Since $B=0$ on every tuple of Figure~\ref{fig:ex-AR}, the smaller candidate
$\Psi_1=(\{F,G\}, H, F\times G)$
achieves the same accuracy as $\Psi_1^{+}$.
Hence, $B$ is redundant, so $\Psi_1^{+}$ violates atomicity and is rejected by the Minimality Test.\looseness=-1
\end{example}
}

\subsection{Perturbation Test}
\label{subsec:perturbation}
The perturbation test is designed to evaluate the \emph{stability} of a candidate \FR $\Psi$ by measuring how likely it is to remain valid under random cross-row perturbations of its participative values.
If $\Psi$ captures a genuine, non-trivial relationship, then randomly recombining observed values across tuples should often destroy it. 
By contrast, spurious candidates caused by low domain diversity or coincidental patterns tend to survive such perturbations more easily.
\looseness=-1

Let $\Psi=(X,Y,f_r)$ be a candidate \FR on a table $T$, and let $C_{\Psi}=X\cup Y$ denote its participative column set.
A single \emph{perturbation trial} uniformly samples two distinct rows $r_a, r_b$ from $T$ without replacement, uniformly samples a non-empty subset $C' \subseteq C_{\Psi}$, swaps 
the values of $C'$ between $r_a$ and $r_b$ to produce perturbed rows $r_a'$ and $r_b'$, 
and checks whether both $r_a'$ and $r_b'$ still satisfy $\Psi$.
If so, we say that $\Psi$ \emph{survives} the trial. 
\eat{
A single perturbation trial proceeds as follows:
\begin{enumerate}[leftmargin=*]
    \item Uniformly sample two distinct rows $r_a$ and $r_b$ from $T$ without replacement;
    
    \item Uniformly sample a non-empty subset $C' \subseteq C_{\Psi}$ from all non-empty subsets of $C_{\Psi}$;

    \item Swap the values of columns in $C'$ between $r_a$ and $r_b$, producing two perturbed rows $r_a'$ and $r_b'$;
    
    \item Check whether both perturbed rows $r_a'$ and $r_b'$ satisfy $\Psi$.
    If so, we say that $\Psi$ \emph{survives} the trial;

\end{enumerate}
}

We define the \emph{perturbation score} $p_{\kw{pert}}$ of $\Psi$, denoted by $p_{\kw{pert}}$ as the probability that $\Psi$ survives a random perturbation trial:
\[
p_{\kw{pert}} = \Pr[\Psi \text{ survives a random perturbation trial}].
\]
In practice, $p_{\kw{pert}}$ is estimated by 
the fraction of survivals over repeated trials.
A low score indicates that $\Psi$ is sensitive to random recombination and is therefore more likely to reflect a genuine 
relationship,
whereas
a high score suggests 
$\Psi$ is spurious.
We accept $\Psi$ as stable only if $p_{\kw{pert}} \le \eta$, where $\eta$ is a user-specified 
threshold.
\looseness=-1

\eat{
\begin{figure}
\centering
\includegraphics[width=\columnwidth]{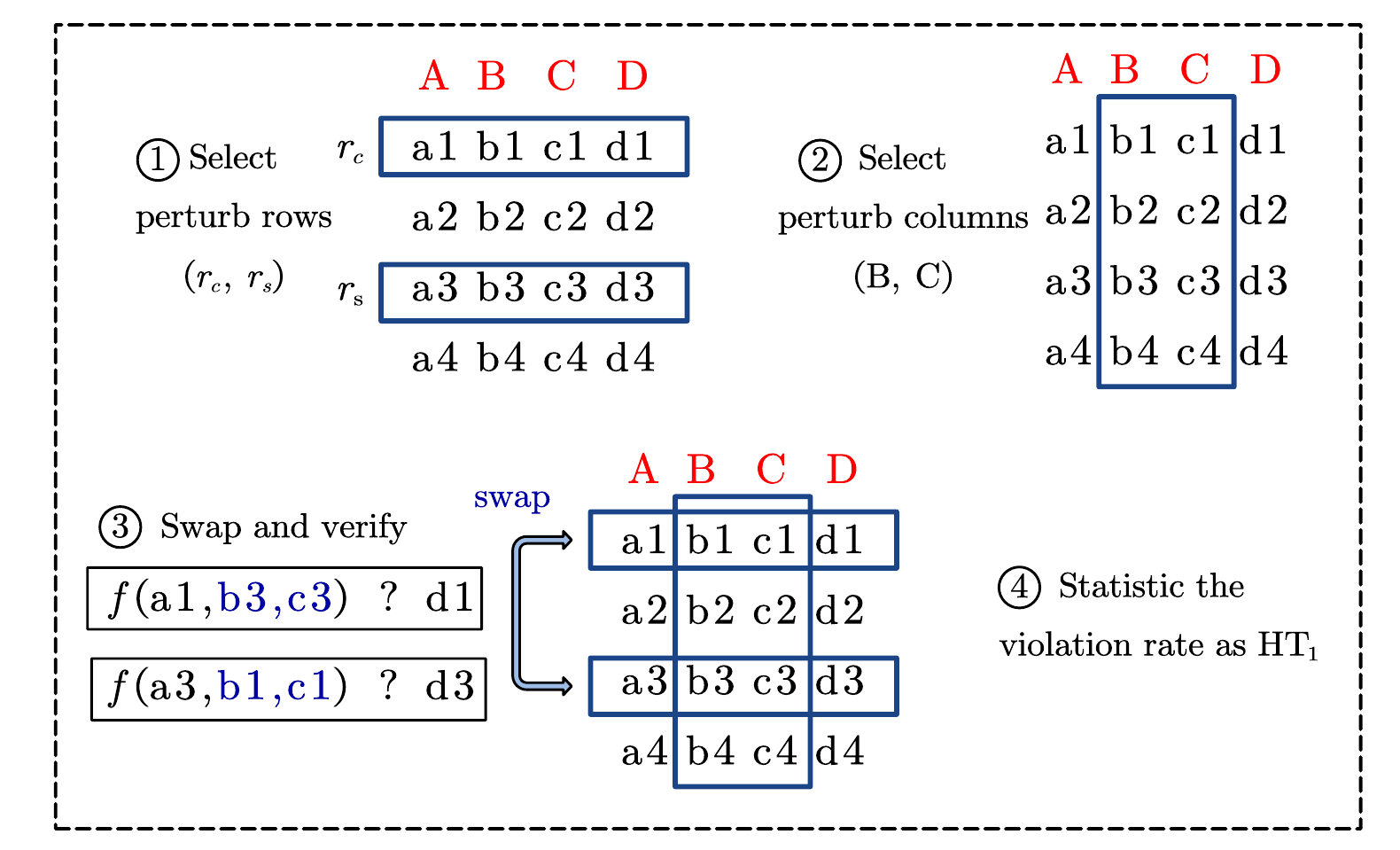}
\caption{Overview of the Perturbation Test procedure.}
\label{fig:HT1Flow}
\end{figure}
}

\smallskip
\noindent\textbf{Algorithm.}
An exact computation of $p_{\kw{pert}}$ requires enumerating all pairs of rows in $T$ together with all $2^{|C_{\Psi}|}-1$ non-empty column subsets of $C_{\Psi}$, which is computationally prohibitive.
We therefore adopt a sampling-based approximation to estimate $p_{\kw{pert}}$, as shown in Algorithm~\ref{alg:PerturbationTest}.
To further improve efficiency, \autorelate incorporates two optimization strategies: 
a \emph{group-by lower bound} that enables exact early rejection before sampling 
(lines~2--4), and a \emph{binomial bound} based on the Wilson confidence interval that 
enables statistically guided early termination during sampling (lines~14--17).
Their details and theoretical guarantees are presented in Section~\ref{sec:optimization}.
\eat{
To further improve efficiency, \autorelate incorporates two optimization strategies.
First, it computes a group-by lower bound $B_{\kw{group}}$ before sampling (line~2).
If this lower bound already exceeds the threshold $\eta$,  the candidate \FR can be rejected immediately without any sampling (lines~3--4).
Second, during sampling, the algorithm periodically  computes a binomial confidence interval for $p_{\kw{pert}}$ using the Wilson method, and terminates early once the interval lies entirely on one side of $\eta$ (lines~14--17). 
Together, these optimizations substantially improve efficiency.
The group-by bound enables exact early rejection before sampling, 
while the binomial bound provides a practical statistically motivated early-termination strategy during sampling.
Their details and theoretical guarantees are presented in Section~\ref{sec:optimization}.
}

\begin{algorithm}[!t]
\SetKwInOut{Input}{Input}
\SetKwInOut{Output}{Output}
\Input{A table $T$, a candidate \FR $\Psi$, a perturbation threshold $\eta$, maximum iterations $n_{\max}$, a confidence $z$-score, and checking interval $\beta$.}
\Output{An estimated perturbation score $p_\kw{pert}$.}

$C_{\Psi} \gets$ the participative column set of $\Psi$\;
$B_\kw{group} \gets$ the group-by lower bound computed from table $T$ and FR $\Psi$\;

\If{$B_\emph{\kw{group}} > \eta$}{
    \Return{$B_\emph{\kw{group}}$}
}

$n_s \gets 0$\;
$n \gets 0$\;

\For{$i \gets 1$ \KwTo $n_{\max}$}{
    $(r_a,\; r_b) \gets$ randomly sample two distinct
    rows from $T$\;
    $C' \gets$ a random non-empty subset of $C_{\Psi}$\;
    $(r_a',\; r_b') \gets \kw{perturb}(r_a,\; r_b,\; C')$\;

    \If{$r_a'$ satisfies $\Psi$ \textbf{and} $r_b'$ satisfies $\Psi$}{
        $n_s \gets n_s + 1$\;
    }

    $n \gets n + 1$\;

    \If{$\mathrm{mod}(n,\beta)=0$}{
        $B_{\kw{lower}}, B_{\kw{upper}} \gets \mathrm{wilson\_interval}(n, n_s, z)$\;
        \If{$B_\emph{\kw{lower}} > \eta$ \textbf{or} $B_\emph{\kw{upper}} < \eta$}{
            \textbf{break}\;
        }
    }
}

$p_\kw{pert} \gets n_s / n$\;
\Return{$p_\emph{\kw{pert}}$}\;
\caption{Perturbation Test}
\label{alg:PerturbationTest}
\end{algorithm}

\smallskip
\noindent\textbf{Complexity.}
An exact computation of the perturbation score $p_{\kw{pert}}$ requires enumerating all
$\binom{|T|}{2}$ row pairs and all $2^{|C_\Psi|}-1$ non-empty column subsets,
yielding $O(|T|^2 \cdot 2^{|C_\Psi|} \cdot |C_\Psi|)$ per candidate,
which is prohibitive even for moderately sized tables.
The sampling-based Algorithm~2 performs at most $n_{\max}$ iterations in the worst case,
each costing $O(|C_\Psi|)$ for the perturbation and verification.
Computing the group-by lower bound incurs an additional
pre-computation cost of $O(|T| \cdot 2^{|C_\Psi|})$ per candidate, but when $B_{\mathrm{group}} > \eta$, the sampling loop is bypassed entirely.
For the remaining candidates, Wilson-interval early termination often stops sampling after only $n^* \ll n_{\max}$ iterations in practice, yielding an effective per-candidate cost of
$O(|T| \cdot 2^{|C_\Psi|} + n^* \cdot |C_\Psi|)$.

\begin{example}
\label{exa:perturbation}
Consider the genuine \FR $\Psi_1=(\{F,G\}, H, F\times G)$ and the spurious candidate 
$\Psi_{\mathrm{coin}}=(\{B,C\}, E, B\times C)$ from Example~\ref{exa:reliability}.
For $\Psi_1$, swapping the value of column $G$ between two sampled rows in 
Figure~\ref{fig:ex-AR} typically violates the equality $H=F\times G$, yielding a low 
perturbation score.
By contrast, $\Psi_{\mathrm{coin}}$ often survives perturbation because $B$, $C$, and $E$ are sparse binary columns and many tuples coincidentally satisfy $B\times C=E=0$.
\end{example}

\eat{
For example, recall Figure \ref{fig:ex-AR}-\ref{fig:ex-FD} and observe how \HTOne~ works on three functional relationships.

\textbf{Arithmetic Relationship.} We start with the Arithmetic Relationship examples from Figure \ref{fig:ex-AR}. The first AR $\Psi_H=(\{F,G\},F*G)$ consists of three columns $F,G, and H$, each of which has multiple elements in the domain.
When applying \HTone~to this three-columns relationship, perturbing one column and perturbing two columns have the same effect. AR does not violate only if \HTOne~ selects the same value for swap or if both F and H columns are equal to 0. According to the \HTOne, the \ReliabilityScore of $\Psi_H=(\{F,G\},F*G)$ is set to 0.02. In the experiments section, we set the default value of \SignificanceLevel to 0.5, so $\Psi_H=(\{F,G\},F*G)$ is correctly considered reliable. As mentioned in Section 2.2, $\Psi_E=(\{B,C\},B*C)$ is not a AR. Since columns B and E are composed of zeros, no matter which group of perturbation rows and perturbation columns is selected, there is no violation in AR. Fake AR $\Psi_E=(\{B,C\},B*C)$ will be filtered out because \ReliabilityScore = 1.

\textbf{String Transformation.} Recall the String Transformation examples in Figure \ref{fig:ex-ST}. The first functional relationship in TP is $\Psi_E=(\{B,C,D\}, concatenate(B,C,D))$. For each column in $C_{\Psi}$ , the relationship only holds if the perturbation results in the same value. $\Psi_E=(\{B,C,D\}, concatenate(B,C,D))$'s \ReliabilityScore is $0.06$, suggesting that this transformation is likely reliable.
The second FR in FP is $\Psi_K=(\{L\}, L\to K$), which means that $K$ is equal to the last character of $L$. $L$(MINSALES) is a column of defined large integers ending in zero, while $K$ consists of a column of zeros.  From a human perspective, there is no FR between $K$ and $L$;  it is merely a coincidence that they both end in zero. This type of coincidental string transformation is prevalent in tables and can be identified using \HTone. Whether column A or B is perturbed, the FR always holds. The fake FR $\Psi_B=(\{A\}, A\to B$) will be filtered out because \ReliabilityScore = 1.

\textbf{Functional Dependency.} The FR $\Psi_B=(\{A\}, A\to B)$ was introduced in Section 2. $\Psi_B=(\{A\}, A\to B)$ is considered reliable because its \ReliabilityScore score is 0.02.
The first functional dependency in FP is $\Psi_M=(\{L\}, L\to M)$, which coincidentally appears because the domain of column $M$ contains only one element. Regardless of how any two elements within $L$ are swapped, the FR always holds. Therefore, $\Psi_M=(\{L\}, L\to M)$ will be filtered out by \HTone.
$\Psi_H=(\{E\}, E\to H)$ holds across the entire table because each value in $E$ is unique. However, there is no necessary connection between Unit price and Category. Suppose new data rows are added with items from two different categories having the same price, the new rows might be altered due to not fitting the fake FR $\Psi_H=(\{E\}, E\to H)$. When \HTone~is applied to $\Psi_H=(\{E\}, E\to H)$, since each Unit price appears only once, perturbing any two rows still maintains the relationship. $\Psi_H=(\{E\}, E\to H)$ will also be filtered out due to an  score $p$ of 1.
}

\subsection{Independence Test}
\label{subsec:independence}
In the dirty-data setting, a candidate \FR $\Psi$ may hold on most rows 
while violating a few 
due to noise.
A natural approach is to discard the violation rows and apply the Perturbation Test on the remaining clean rows.
However, this 
removes precisely the rows that carry diagnostic information; violations may be systematically caused by an omitted column, and removing them eliminates the evidence needed to detect such incompleteness.
For example, removing the violation rows of
$\Psi_2^{-}$ 
in Example~\ref{exa:reliability}
would eliminate the rows where $N\neq 0$,
causing it to incorrectly survive the subsequent Perturbation Test.
\looseness=-1

To address this issue, we introduce the \emph{Independence Test}, which verifies the \emph{integrity} of a candidate \FR before the Perturbation Test.
Recall from Definition~\ref{def:integrity} that integrity requires the violation pattern of $\Psi$ to be
independent of all non-participative columns $Z = C \setminus C_\Psi$.
Intuitively, if the violations of $\Psi$ are systematically  correlated with any column
$c' \in Z$, then $\Psi$ violates integrity; the output column
is not solely determined by the declared input columns $X$, but is also
influenced by $c'$, which should have been included in $X$.
Such candidates are therefore regarded as unreliable and are discarded.

Let $\Psi=(X,Y,f_r)$ be a candidate \FR on table $T$, with $C_\Psi=X\cup Y$. 
For each tuple $t \in T$, we define the \emph{violation indicator}
\[
I_{\Psi}(t)=\indicator[t \not\models \Psi],
\]
which equals $1$ if $t$ violates $\Psi$ and $0$ otherwise. 
Let $\mathbf{I}_\Psi = (I_\Psi(t))_{t \in T}$ denote the 
indicator vector over all tuples in $T$.
For every non-participative column $c' \in C \setminus C_{\Psi}$,
the Independence Test examines whether the binary variable
$I_{\Psi}(t)$ is statistically independent of the observed value $t[c']$.
Formally, for each such column $c'$, we test
\[
H_0^{c'}: I_{\Psi}(t) \perp\!\!\!\perp t[c']
\qquad \text{against} \qquad
H_1^{c'}: I_{\Psi}(t) \not\!\perp\!\!\!\perp t[c'].
\]
Empirically, this is implemented by constructing the contingency table
between $\mathbf{I}_{\Psi}$ and $c'$ and computing the corresponding $p$-value
via a chi-square test of independence~\cite{pearson1900}.
If any $p$-value is no greater than the significance threshold $\alpha$,
we reject $\Psi$ as violating integrity; otherwise, $\Psi$ passes the Independence Test.

\eat{
\begin{figure}
\centering
\includegraphics[width=\columnwidth]{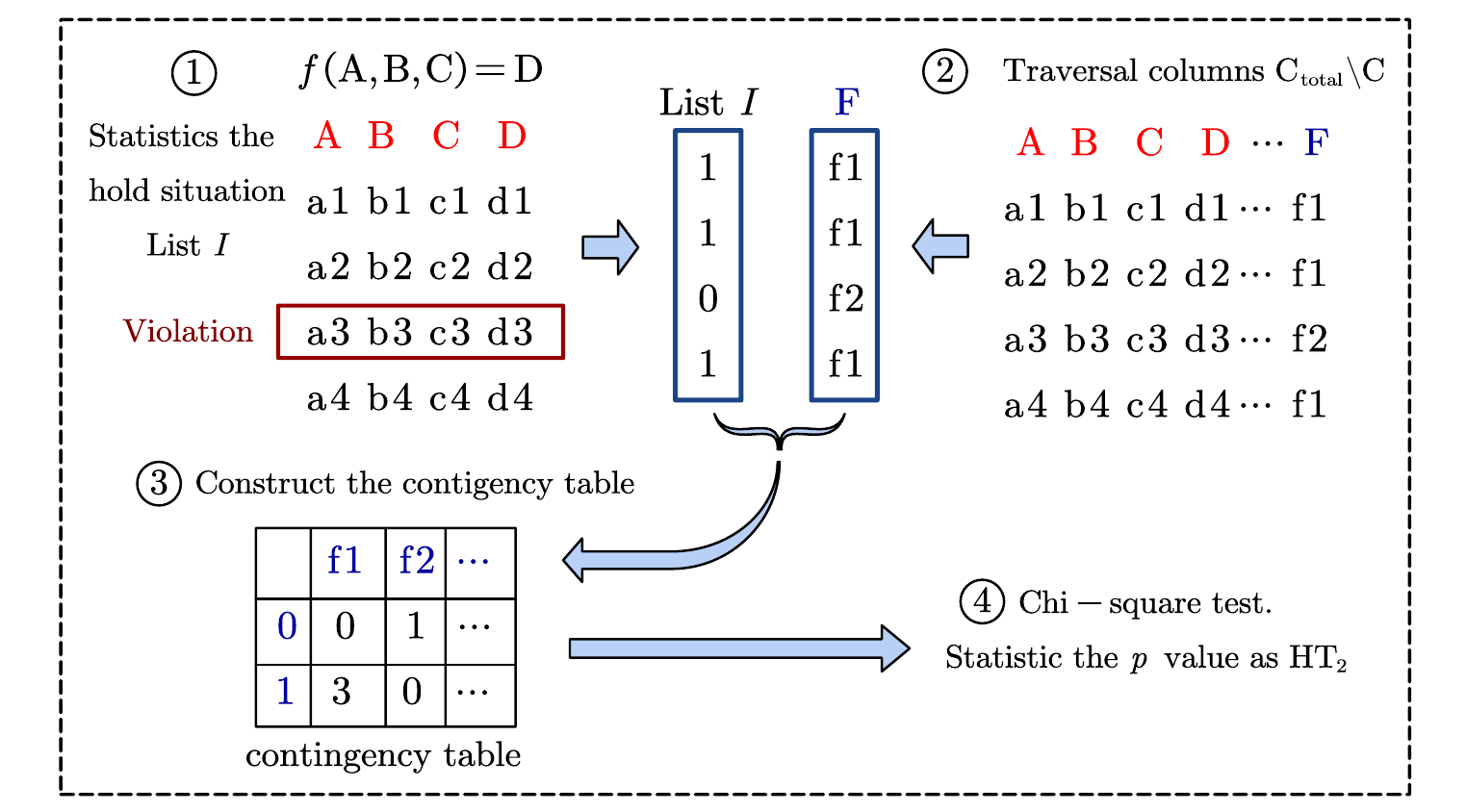}
\caption{Overview of the Independence Test procedure.}
\label{fig:IndependenceTest}
\end{figure}
}

\smallskip
\noindent\textbf{Algorithm.}
Algorithm~\ref{alg:IndependenceTest} summarizes the Independence Test procedure. 
It first computes the violation indicator vector $\mathbf{I}_\Psi$,
where $I_\Psi(t) = 1$ if tuple $t$ violates $\Psi$ and $0$ otherwise (line~1).
It then iterates over each non-participative column $c' \in C_{\mathrm{irr}}$ (line~4):
for each such column, it constructs the contingency table between $c'$ and
$\mathbf{I}_\Psi$ (line~5) and computes the corresponding $p$-value via a
chi-square test of independence~\cite{pearson1900} (line~6).
If any $p$-value falls below the significance threshold $\alpha$, the candidate \FR
$\Psi$ is immediately rejected as violating integrity and the algorithm returns
\kw{False} (lines~7--8); if no such column is found, $\Psi$ passes the
Independence Test and the algorithm returns \kw{True} (line~9).

\eat{
the Independence Test proceeds in four steps:
\begin{enumerate}[leftmargin=*]
\item Compute the violation indicator vector $\mathbf{I}_\Psi$, 
where $I_\Psi(t) = 1$ if tuple $t$ violates $\Psi$ 
and $I_\Psi(t) = 0$ otherwise.

\item For each non-participative column $c' \in C \setminus C_\Psi$,
construct the contingency table between $c'$ and $I_\Psi$.

\item Apply the chi-square test of independence~\cite{pearson1900} to each contingency
table, testing the null hypothesis $H_0^{c'}: I_{\Psi}(t) \perp\!\!\!\perp t[c']$, 
and record the corresponding $p$-value $p_c'$ measuring the association between
$c'$ and $I$.

\item If any $p$-value $p_c'$ falls below the significance threshold $\alpha$, reject the candidate FR $\Psi$ as violating integrity; otherwise, $\Psi$ passes the Independence Test.
\end{enumerate}
}

\eat{
\begin{table}
    \caption{An example of a contingency table}
    \centering
    \label{tab:ContingencyTable}
    \begin{tabular}{c|ccc|c}
    \hline
    \diagbox{I}{c$_N$} & 7550 & 1000 & 0  & Total \\ \hline
    1                  & 1    & 1    & 0  & 2     \\
    0                  & 0    & 0    & 18 & 18    \\ \hline
    Total              & 1    & 1    & 18 & 20    \\ \hline
    \end{tabular}
\end{table}
}

\eat{
\begin{example}
\label{exa:independence}
Continuing with the incomplete candidate
$\Psi_2^{-}=(\{H,K\}, O, H+K)$
from Example~\ref{exa:reliability}, suppose we are given a dirty-data variant of
Figure~\ref{fig:ex-AR}.
Although $\Psi_2^{-}$ may still be $\tau$-accurate, its violations are not random;
they occur on rows where the omitted column $N$ is nonzero.
The Independence Test therefore examines whether the violation indicator of
$\Psi_2^{-}$ is independent of each non-participative column.
In this case, the violation pattern is significantly associated with $N$,
indicating that $\Psi_2^{-}$ omits an essential attribute and thus violates
integrity.
\end{example}
}

\eat{
\begin{example}
Consider the FRs in Figure~\ref{fig:dirty-AR}.
For $\Psi_6 = (\{F,G\}, H, F \times G)$, the indicator vector $\mathbf{I}_{\Psi_6}$ shows no significant correlation with any non-participative column, so $\Psi_6$ passes the Independence Test.
For $\Psi_7 = (\{H,K\}, O, H+K)$, the contingency table between $\mathbf{I}_{\Psi_7}$ and column $N$ (shown in Table~\ref{tab:ContingencyTable}) yields a chi-square $p$-value of $4.54 \times 10^{-5}$, well below
$\alpha = 0.05$, indicating a significant dependence.
This reveals that the violations of $\Psi_7$ are not random noise but are systematically associated with the values of column $N$, indicating that $\Psi_7$ violates integrity and is therefore unreliable.
\autorelate correctly rejects it.
\end{example}
}

\noindent\textbf{Complexity.}
Let $C_{\mathrm{irr}} = C \setminus C_\Psi$ be the set of non-participative columns.
Computing the violation indicator vector $\mathbf{I}_\Psi$ costs $O(|T| \cdot |C_\Psi|)$.
For each $c' \in C_{\mathrm{irr}}$, constructing the contingency table requires a single
pass over $T$ in $O(|T|)$, and the chi-square test costs $O(d_{c'})$ where $d_{c'}$
is the number of distinct values in $c'$.
The total cost is therefore
$O\!\left(|T| \cdot |C_\Psi| + |C_{\mathrm{irr}}| \cdot (|T| + \max_{c'} d_{c'})\right)
= O(|T| \cdot |C|)$, since $|C_\Psi|, |C_{\mathrm{irr}}| \leq |C|$ and $d_{c'} \leq |T|$.
Note that the Independence Test is invoked only in the dirty-data setting;
in the clean-data setting, $\mathbf{I}_\Psi$ is uniformly zero and the test is skipped.

\eat{
\smallskip
\noindent\textbf{Remark.}
The Independence Test is meaningful only in the dirty-data setting.
In the clean-data setting, any candidate that satisfies accuracy has no observed violations, so $\mathbf{I}_\Psi$ is uniformly zero and carries no information for independence analysis.
}

\begin{algorithm}[t]
\caption{Independence Test}
\label{alg:IndependenceTest}
\SetKwInOut{Input}{Input}
\SetKwInOut{Output}{Output}
\Input{A table $T$, a candidate FR $\Psi$, and a significance threshold $\alpha$.}
\Output{A Boolean decision indicating whether $\Psi$ passes the Independence Test.}

$\mathbf{I}_{\Psi} \gets$ the violation indicator vector of $\Psi$ on $T$\;
$C_{\Psi} \gets$ the participative column set of $\Psi$\;
$C_{\mathrm{irr}} \gets C \setminus C_{\Psi}$, where $C$ is the column set of $T$\;

\ForEach{$c' \in C_{\mathrm{irr}}$}{
    construct the contingency table between $\mathbf{I}_{\Psi}$ and $c'$\;
    $p \gets$ the $p$-value of the chi-square test of independence on the contingency table\;
    \If{$p \le \alpha$}{
        \Return{\emph{\kw{False}}}
    }
}
\Return{\emph{\kw{True}}}
\end{algorithm}

\section{Optimization Strategies}
\label{sec:optimization}
The reliability verification stage of \autorelate can be computationally expensive.
As analyzed in Sections~\ref{subsec:perturbation} and~\ref{subsec:independence}, 
the Independence Test runs in $O(|T| \cdot |C|)$ time per candidate, which is linear in the table size and inexpensive in practice.
The Perturbation Test, however, requires enumerating all $\binom{|T|}{2}$ row pairs and all $2^{|C_\Psi|}-1$ non-empty column subsets in its exact form,
yielding a per-candidate complexity of $O(|T|^2 \cdot 2^{|C_\Psi|})$, which is prohibitive even for moderately sized tables.

To address these bottlenecks, we develop three complementary optimization strategies for the \kw{PerturbationTest}.
We first introduce a \emph{group-by bound} that 
rejects candidate FRs before sampling begins (Section~\ref{subsec:groupby}). 
We then propose a \emph{closed-form speed-up} strategy for AR-type FRs that reduces the per-iteration verification cost within the sampling loop by exploiting their expression-tree structure (Section~\ref{subsec:closedform}).
Finally, we develop a \emph{binomial bound} that enables statistically guided early termination during sampling once the perturbation score can be determined with statistical confidence (Section~\ref{subsec:binomial}).
Together, these strategies substantially reduce the practical cost of perturbation-based verification while preserving the correctness of deterministic pruning and providing approximate statistical error control for early termination.

\eat{
\subsection{Candidate-Level Pruning}
\label{subsec:pruning}
We first propose a pruning strategy that exploits the \emph{atomicity} criterion to reduce the number of candidates that enter the reliability verification stage.
If a smaller input column set $X' \subset X$ already determines the output column $Y$ with sufficient accuracy, i.e., $\mathrm{sat}_T(X', Y, f'_r) \geq \tau$
for some mapping $f'_r$, then any candidate $\Psi = (X, Y, f_r)$ with $X' \subset X$ is redundant and can be safely discarded before entering the subsequent verification steps.

To filter out non-atomic FRs, we enumerate candidate FRs
in a \emph{layer-wise} manner, processing all candidates with $|X| = i$ before those with $|X| = i+1$.
At each level $i$, \autorelate invokes the \kw{MinimalityTest} (Line~7 of Algorithm~\ref{alg:Auto-Relate}):
for each candidate $\Psi = (X, Y, f_r)$, it checks whether the pruning set $R_{\kw{prune}}$ already contains a previously retained candidate $\Psi' = (X', Y, f'_r)$ with $X' \subset X$.
If so, $\Psi$ is non-atomic and is pruned immediately.
Because the enumeration proceeds level by level, pruning decisions made at level $i$ propagate to all later levels, thereby eliminating potentially large numbers of supersets at level $i+1$ and beyond.
This strategy is especially effective for tables with many columns, where the number of candidate column subsets grows exponentially.

A further advantage of the layer-wise structure is that candidates within the same level are mutually independent, since the verification outcome of
one candidate does not affect any other at the same level.
This enables all candidates within a level to be verified in parallel, and we exploit this property in our implementation using a multi-process architecture.
}

\subsection{Group-by Bound}
\label{subsec:groupby}
We first introduce a deterministic group-by lower bound on the perturbation score of a candidate \FR $\Psi=(X,Y,f_r)$.
This bound enables exact rejection before the sampling loop begins.

The key observation is straightforward.
In a perturbation trial, two rows $r_a$ and $r_b$ are sampled uniformly at random from $T$ without replacement, and the values of a randomly selected non-empty subset $C' \subseteq C_\Psi$ are swapped between them.
If the swapped values are identical across $r_a$ and $r_b$ for every column in $C'$, i.e., $r_a[c] = r_b[c]$ for all $c \in C'$, then the perturbed rows $r'_a$ and $r'_b$ remain unchanged, and the FR trivially survives the perturbation.
Therefore, the perturbation score $p_\kw{pert}$ is at least as large as the probability of such an \emph{identity-swap} event.

\eat{ 
We formalize this observation as follows.
For a single column $c \in C_\Psi$, let $v_1, v_2, \ldots, v_m$ be the distinct values of $c$ in $T$, and let $n_i = |\{t \in T : t[c] = v_i\}|$ be the frequency of value $v_i$.
Then, when two rows are sampled uniformly at random without replacement, the probability that they share the same value in column $c$ is
\begin{equation}
    \label{eq:single_col_bound}
    q_c = \sum_{i=1}^{m} \frac{n_i}{|T|} \cdot \frac{n_i - 1}{|T| - 1}.
\end{equation}

More generally, for a non-empty subset $C' \subseteq C_\Psi$ consisting of multiple columns, the identity-swap event requires all columns in $C'$ to match simultaneously.
Let $\mathbf{v}_1, \mathbf{v}_2, \ldots, \mathbf{v}_m$ be the distinct \emph{joint values} observed on $C'$, and let $n_i$ be the number of rows sharing joint value $\mathbf{v}_i$.
Then the probability of an identity swap on $C'$ is
\begin{equation}
    \label{eq:joint_col_bound}
    q_{C'} = \sum_{i=1}^{m} \frac{n_i}{|T|} \cdot \frac{n_i - 1}{|T| - 1}.
\end{equation}
}

For a non-empty subset $C' \subseteq C_\Psi$, let $\mathbf{v}_1, \ldots, \mathbf{v}_m$ be the distinct joint values observed on $C'$, and let $n_i$ be the number of rows sharing joint value $\mathbf{v}_i$.
When two rows are sampled uniformly at random without replacement, the probability of an identity swap on $C'$ is
\begin{equation}
    \label{eq:joint_col_bound}
    q_{C'} = \sum_{i=1}^{m} \frac{n_i}{|T|} \cdot \frac{n_i - 1}{|T| - 1}.
\end{equation}

Recall that in Algorithm~\ref{alg:PerturbationTest}, 
each column in $C_\Psi$ is independently included in $C'$ with probability $1/2$, conditioned on $C'$ being non-empty.
This is equivalent to sampling $C'$ uniformly at random from the $2^{|C_\Psi|} - 1$ non-empty subsets of $C_\Psi$. 
Accordingly, we define the group-by lower bound $B_{\kw{group}}$ as the expected identity-swap probability under this distribution:
\begin{equation}
    \label{eq:group_bound}
    B_{\kw{group}} = \frac{1}{2^{|C_\Psi|} - 1} \sum_{\emptyset \neq C' \subseteq C_\Psi} q_{C'}.
\end{equation}

\begin{lemma}[Group-by Lower Bound]
\label{lem:groupby}
In the clean-data setting, let $p_\emph{\kw{pert}}$ denote the perturbation score of a candidate \FR $\Psi$ on a table $T$, computed by Algorithm~\ref{alg:PerturbationTest}. Then $p_\emph{\kw{pert}} \geq B_{\kw{group}}$.
\end{lemma}

\begin{proof}
Fix any non-empty subset $C' \subseteq C_\Psi$.
Let $I$ denote the identity-swap event that $r_a[c] = r_b[c]$ for every $c \in C'$, and let $S$ denote the event that $\Psi$ survives the perturbation trial.
Under event $I$, swapping the values on $C'$ leaves both rows unchanged, i.e., $r'_a = r_a$ and $r'_b = r_b$.
In the clean-data setting, both $r_a$ and $r_b$ satisfy $\Psi$ before perturbation, and since they remain unchanged, $\Psi$ is also satisfied after perturbation.
Hence, $I \subseteq S$, and therefore
$\Pr[S | C'] \geq \Pr[I | C'] = q_{C'}$.
Taking expectation over the uniform choice of $C'$, we have \looseness=-1
\begin{align}
p_\kw{pert} &= \Pr[S]
   = \frac{1}{2^{|C_\Psi|}-1}
     \sum_{\emptyset \neq C' \subseteq C_\Psi} \Pr[S \mid C'] \notag \\
  &\geq \frac{1}{2^{|C_\Psi|}-1}
     \sum_{\emptyset \neq C' \subseteq C_\Psi} q_{C'}
   = B_{\kw{group}}. \tag*{\mbox{$\Box$}}
\end{align}
\renewcommand{\eop}{}
\end{proof}

Therefore, if $B_{\kw{group}} > \eta$, the candidate \FR can be safely rejected without entering the sampling loop (Lines~2--4 of Algorithm~\ref{alg:PerturbationTest}).

\smallskip
\noindent\textbf{Dirty-data extension.}
In the dirty-data setting, a candidate \FR is allowed to have a small number of violations after passing the $\tau$-accuracy filter.
To apply the group-by bound in this setting, we restrict
attention to the subset of rows on which $\Psi$ holds, i.e., $T^+_\Psi = \{t \in T : t \models \Psi\}$.
Since every identity swap on $T^+_\Psi$ leaves the sampled rows unchanged and those rows satisfy $\Psi$ before perturbation, the same argument as in Lemma~\ref{lem:groupby} yields a valid lower bound for the
perturbation score over initially satisfied rows.
Accordingly, in the dirty-data setting, we use
$B_{\kw{group}}^{+} = 
B_{\kw{group}}(T_\Psi^{+})$.

\eat{
In practice, computing $B_{\kw{group}}$ or $B_{\kw{group}}^{+}$ requires group-by aggregations over the $2^{|C_\Psi|} - 1$ non-empty subsets of $C_\Psi$, incurring a cost of $O(|T| \cdot 2^{|C_\Psi|})$ per candidate.
Although exponential in $|C_\Psi|$, this is substantially cheaper than the full sampling process, and in practice $|C_\Psi|$ is small for most candidates.
As we show in Section~\ref{sec:experiment}, this bound alone allows 85.4\% of candidate FRs on the clean AR dataset to skip the sampling process entirely.

\smallskip
\noindent\textbf{Remark.}
For AR-type FRs represented as expression trees, an additional per-iteration speed-up is available. 
When the perturbed columns correspond to the left and right children of a binary operation node, the survival check reduces to testing whether the two rows share the same value of the corresponding sub-expression (the
\emph{perturbation core}), thereby avoiding re-evaluation of the full expression tree for each trial.
}

\subsection{Closed-form Speed-up}
\label{subsec:closedform}
For AR-type FRs that enter the sampling loop, we additionally reduce the per-trial verification cost by exploiting the expression-tree representation of the arithmetic formula.

Specifically, let $\Psi=(X,Y,f_r)$ be an AR candidate whose formula $f_r$ is represented as a binary expression tree.
When the sampled perturbation set $C'$ coincides with the set of leaf columns in the subtree rooted at 
some internal binary operator node $u$, we refer to the value of that sub-expression as the \emph{perturbation core}.
In this case, swapping all columns in $C'$ between two sampled rows only affects the perturbation core while leaving the rest of the expression unchanged.
Therefore, for such perturbations, the trial outcome can be determined by checking whether the perturbation core remains unchanged, rather than re-evaluating the entire 
expression tree on the perturbed rows.
This optimization is specific to AR-type FRs and is applied only when the sampled perturbation set matches the leaf set of some subtree in the expression tree.

\eat{
\begin{example}
Consider the FR $A \times B - C = D$ in Figure~\ref{fig:ClosedForm}.
If the perturbed columns are $A$ and $B$, then the sub-expression $A \times B$ forms the perturbation core. 
For two sampled rows that satisfy the FR before perturbation, swapping $A$ and $B$ preserves the FR whenever the two rows have the same value of $A \times B$, so the trial outcome can be determined without recomputing the full formula.
\end{example}

\begin{figure}[t]
\centering
\includegraphics[width=0.8\columnwidth]{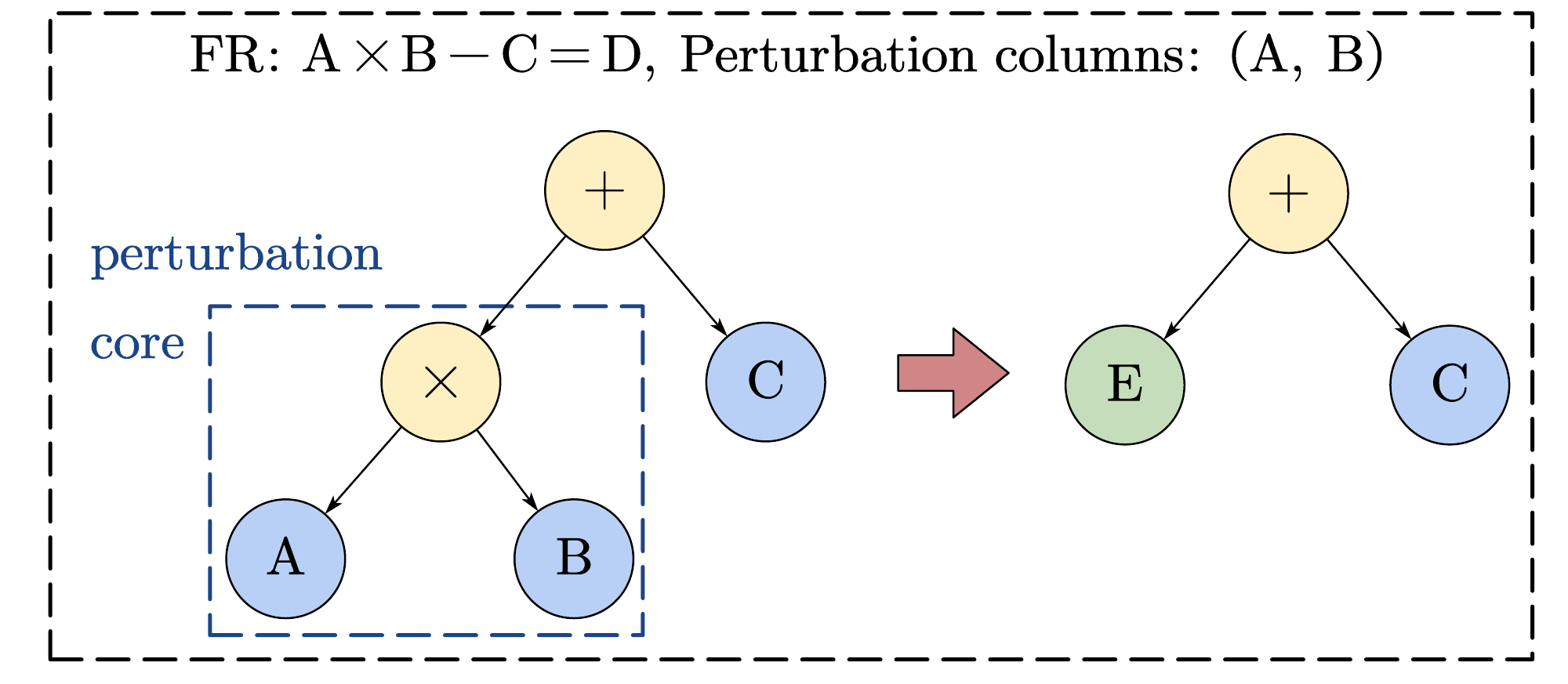}
\caption{Illustration of the closed-form speed-up for ARs.}
\label{fig:ClosedForm}
\end{figure}
}

\begin{example}
\label{exa:closedform}
Consider the genuine AR $\Psi_2=(\{H,K,N\}, O, H+K+N)$ in Example~\ref{exa:FRs} and a perturbation trial in which the sampled column subset is $C'=\{H,K\}$.
In the expression tree of $\Psi_2$, the columns $H$ and $K$ form the sub-expression $H+K$, which serves as the perturbation core.
Since the remaining part of the expression, namely $N$, is unchanged in this trial, the Closed-form Speed-up does not need to recompute the full expression $H+K+N$ on the perturbed tuples.
Instead, provided both sampled tuples satisfy $\Psi_2$ before perturbation,
it only checks whether they share the same core value, i.e., 
$H_{r_a}+K_{r_a}=H_{r_b}+K_{r_b}$, which is both necessary and sufficient for
$\Psi_2$ to remain satisfied after swapping $H$ and $K$ between $r_a$ and $r_b$.
\end{example}

\subsection{Binomial Bound}
\label{subsec:binomial}
When the group-by lower bound does not suffice to reject a candidate FR, Algorithm~\ref{alg:PerturbationTest} enters the sampling loop.
Each perturbation trial yields a Bernoulli outcome.
If the number of trials were fixed in advance at $n$, then the cumulative number of survivals $n_s$ would follow a $\mathrm{Binomial}(n, p_\kw{pert})$ distribution, where $p_\kw{pert}$ denotes the true perturbation score.
This motivates the use of the Wilson score interval as a practical confidence interval for $p_\kw{pert}$.
In Algorithm~\ref{alg:PerturbationTest}, at each checkpoint with $n$ completed trials and $n_s$ observed survivals, we compute the Wilson interval for $p_\kw{pert}$ and terminate sampling early once the interval lies entirely on one side of the threshold $\eta$.
We emphasize, however, that because the interval is repeatedly monitored during sampling and the stopping time is data-dependent, the resulting error control should be interpreted as approximate rather than as an exact finite-sample sequential guarantee.
\looseness=-1

Specifically, we adopt the Wilson score interval~\cite{wilson1927probable}
rather than the Wald interval, because it is known to exhibit substantially better finite-sample performance, especially for small $n$ or extreme values of $p_\kw{pert}$.
At any fixed checkpoint after $n$ trials with $n_s$ successes, the nominal $(1-\gamma)$ confidence interval $[B_{\mathrm{lower}}, B_{\mathrm{upper}}]$ for $p_\kw{pert}$ is
\begin{equation}
    \label{eq:wilson_lower}
    B_{\mathrm{lower}} = \frac{n_s + z^2/2}{n + z^2} -
    \frac{z\sqrt{n_s(n - n_s)/n + z^2/4}}{n + z^2},
\end{equation}
\begin{equation}
    \label{eq:wilson_upper}
    B_{\mathrm{upper}} = \frac{n_s + z^2/2}{n + z^2} +
    \frac{z\sqrt{n_s(n - n_s)/n + z^2/4}}{n + z^2},
\end{equation}
where $z=\Phi^{-1}(1-\gamma/2)$ is the standard normal critical value for a nominal two-sided confidence level of $1-\gamma$.

Based on this interval, the sampling loop in Algorithm~\ref{alg:PerturbationTest} terminates early if $B_{\kw{lower}} > \eta$ or $B_{\kw{upper}} < \eta$ (lines~14--17).
That is, early termination occurs once the entire Wilson interval lies on one side of the threshold~$\eta$.
To reduce the overhead of interval evaluation, the Wilson interval is computed once every $\beta$ iterations rather than after every perturbation trial.

\begin{lemma}[Binomial Bound]
\label{lem:wilson}
Let $p_\emph{\kw{pert}}$ denote the true perturbation score of a candidate FR $\Psi$.
If Algorithm~\ref{alg:PerturbationTest} terminates early because $B_{\mathrm{lower}} > \eta$ or $B_{\mathrm{upper}} < \eta$,
then the resulting threshold-based decision can disagree with the exact decision based on $p_\emph{\kw{pert}}$ only if $p_\emph{\kw{pert}} \notin [B_{\mathrm{lower}}, B_{\mathrm{upper}}]$.
\end{lemma}

\begin{proof}
The early termination is triggered only when the Wilson interval at the stopping checkpoint
$[B_{\mathrm{lower}}, B_{\mathrm{upper}}]$ lies strictly on one side of the threshold $\eta$, i.e., when
$B_{\mathrm{lower}} > \eta$ or $B_{\mathrm{upper}} < \eta$.
If $B_{\mathrm{upper}} < \eta$ and $p_\kw{pert} \in [B_{\mathrm{lower}}, B_{\mathrm{upper}}]$, then $p_\kw{pert} < \eta$,
so the resulting threshold-based decision agrees with the exact decision based on $p_\kw{pert}$.
Similarly, if $B_{\mathrm{lower}} > \eta$ and $p_\kw{pert} \in [B_{\mathrm{lower}}, B_{\mathrm{upper}}]$, then $p_\kw{pert} > \eta$, so the two decisions also agree.
Therefore, a disagreement can occur only if $p_\kw{pert} \notin [B_{\mathrm{lower}}, B_{\mathrm{upper}}]$, that is, 
\[
\{\text{decision error}\} \subseteq
\{p_\kw{pert} \notin [B_{\mathrm{lower}}, B_{\mathrm{upper}}]\}.
\]
Taking probabilities on both sides gives
\[
\Pr(\text{decision error}) \le
\Pr(p_\kw{pert} \notin [B_{\mathrm{lower}}, B_{\mathrm{upper}}]). \tag*{\mbox{$\Box$}}
\]
\vspace{-1em}
\renewcommand{\eop}{}
\end{proof}

Consequently, any early-termination decision can be erroneous only if the Wilson interval at the termination checkpoint fails to cover $p_\kw{pert}$.
Thus, the Binomial Bound reduces the practical cost of the Perturbation Test by terminating sampling once the decision relative to $\eta$ becomes clear from the monitored Wilson interval, while this early-termination criterion should be interpreted as a practical approximation rather than an exact sequential coverage guarantee under repeated monitoring.

\eat{
Consequently, any early-termination decision can be erroneous only if the Wilson interval at the termination checkpoint fails to cover $p_\kw{pert}$.
Thus, the reliability of the early decision depends on the coverage behavior of the monitored Wilson interval. Under repeated monitoring, this criterion should be interpreted as a practical approximation for early termination rather than as an exact sequential coverage guarantee.

\eat{
Together with the Group-by Bound, the Binomial Bound allows the sampling process in the Perturbation Test to be either skipped entirely or terminated early once the decision relative to $\eta$ can be made from the Wilson interval.
Accordingly, the effective per-candidate cost is
$O(|T| \cdot 2^{|C_\Psi|} + n^* \cdot |C_\Psi|)$, where the first term accounts for the Group-by Bound computation and the second for the sampling loop.
In practice, early termination often occurs at $n^* \ll n_{\max}$ iterations.
}

Together with the Group-by Bound, the Binomial Bound further reduces the practical cost of the Perturbation Test by allowing the sampling process to terminate early once the decision relative to $\eta$ becomes clear from the monitored Wilson interval.
}
\section{Experiments}
\label{sec:experiment}

Using real-life data, we experimentally evaluated the (1)
effectiveness and (2) efficiency of the proposed \autorelate of discovering reliable \FRs for (a) different types of \FRs and (b) clean and dirty data. \looseness=-1

\subsection{Experimental Settings}
\label{subsec:exp_setting}

\smallskip
\noindent\textbf{Datasets.}
Since benchmark datasets for \AR and \ST discovery are scarce, and existing \FD benchmarks are limited in both scale and diversity, we collaborated with Microsoft Research
to construct a new benchmark suite, \kw{Real}, 
based on a large collection of real-world Excel spreadsheets
and tabular data from \kw{Auto-Pipeline}~\cite{yang2021auto}. 

The suite supports the evaluation of all three \FR types considered in this paper, namely \ARs, \STs, and \FDs.  
As shown in Table~\ref{tab:dataset}, we summarized the statistics of the datasets, including the number of tables, the row and column statistics per table, and the number of positive and negative \FRs.

\begin{itemize}[leftmargin=*]
    \item \textbf{\realAR} contains 2{,}539 tables, with 28--16{,}961 rows and 3--400 columns per table, yielding 966 positive and 8{,}102 negative \AR candidates.
    Positive samples correspond to ground-truth formulas recorded in original spreadsheets, while negative ones are spurious candidates 
    with observed violations on the source tables. \looseness=-1

    \item \textbf{\realST} contains 55{,}766 tables, with 5--23{,}005 rows and 32--500 columns per table, yielding 1{,}938 positive and 294{,}253 negative \ST candidates.
    Positive samples are verified string-transformation formulas, 
    while negative ones are spurious candidates whose output column is a string column inserted from an unrelated table.\looseness=-1

    \item \textbf{\realFD} contains 374 tables, with 107--717{,}321 rows and 4--333 columns per table, yielding 3{,}510 positive and 5{,}579 negative \FD candidates.
    Positive samples include both reliable and approximate \FDs extracted from the data, while negative samples 
    comprise spurious candidates with substantial observed violations and manually added unreliable dependencies.
\end{itemize}

For evaluation, we derived both a \emph{clean-data} setting and a \emph{dirty-data} setting from the underlying source tables.
In the clean-data setting, violating rows of negative candidate \FRs were removed so that both positive and negative candidates hold exactly on the resulting tables, making the discrimination task more challenging.
In the dirty-data setting, we injected 10\% random noise into the source tables to simulate realistic data errors.
\looseness=-1

In addition, we included \rwd~\cite{parciak2024measuring}, which contains 10 tables and 126 manually labeled \FD relationships, 
for further validating effectiveness.
Unlike commonly used \FD discovery benchmarks that are primarily designed
for evaluating enumeration efficiency, \rwd provides ground-truth design \FDs and is therefore suitable for evaluating candidate-level reliability discrimination.
\looseness=-1

\smallskip 
\noindent\textbf{Baselines.}
We compared \autorelate against a comprehensive set of 18 baselines and ranking measures from 12 groups.
(1) \textbf{\kw{Explain-Da-V}}~\cite{shraga2023explaining}, 
which explains dataset-version changes using data transformations; we adapt its conciseness and concentration criteria to score candidate \FRs.
(2) 
\textbf{GPT-5}~\cite{openai2025gpt5}, a large language model prompted with column names and sampled rows in a few-shot in-context learning setting to verify \FRs. 
(3) \textbf{\kw{Chi-Squared (CS)}}~\cite{ilyas2004cords}, 
a classical test of independence. 
(4) \textbf{\kw{Mutual Information (MI)}}~\cite{cheng1997learning}, 
an information-theoretic dependence measure. 
(5) \textbf{\kw{Fraction of Information (FI)}}~\cite{cavallo1987theory,giannella2004approximation}, an information-theoretic measure that quantifies the relative reduction in uncertainty about the RHS given the LHS. 
(6) \textbf{Reliable FI (\kw{RFI$^+$}, \kw{RFI$^{\prime+}$})}~\cite{mandros2017discovering,parciak2024measuring}, two bias-corrected FI-based variants. 
(7) \textbf{\kw{Smoothed MI (SMI)}}~\cite{pennerath2020discovering}, an MI 
variant with Laplace smoothing. 
(8) \textbf{\kw{Co-occurrence ($\rho$)}}~\cite{ilyas2004cords}, 
the ratio of distinct LHS values to distinct $(\text{LHS}, \text{RHS})$ value pairs.
(9) \textbf{\kw{$g_1$-based measures} ($g_1$, 
$g_1^s$, $g_1^{s+}$)}~\cite{kivinen1995approximate,giannella2004approximation,parciak2024measuring}, 
based on the fraction of violating tuple pairs and closely related variants.
(10) \textbf{\kw{Violating row ratio ($g_2$)}}~\cite{kivinen1995approximate,parciak2024measuring}, 
the fraction of tuples involved in 
violations.
(11) \textbf{\kw{Row-removal ratio ($g_3$, $g_3^{\prime}$)}}~\cite{kivinen1995approximate,giannella2004approximation,parciak2024measuring}, 
the minimum fraction of rows whose removal eliminates all violations.
(12) \textbf{\kw{Probabilistic dependency} ($p_{\mathit{dep}}$}) and its normalized variants ($\mu^+$, $\tau$)~\cite{piatetsky1993measuring,goodman1979measures}, 
based on the conditional probability that two randomly sampled tuples agree on 
the LHS also agree on the RHS.

\eat{ 
For dependency measures originally defined for ranking \kw{FD}/\kw{AFD} candidates, we extend them to multi-column FRs by encoding the input columns as a composite attribute using an underscore delimiter, following standard practice~\cite{parciak2024measuring}.
}

For quality comparison, all methods were evaluated on the same candidate \FR pool in each benchmark. 
Each benchmark provided a fixed set of positive and negative FR candidates, and each method scored, ranked, or filtered these candidates independently.

\begin{table}[t]
\centering
\caption{Statistics of the \kw{Real} benchmark suite}
\vspace{-2ex}
\label{tab:dataset}
\resizebox{\linewidth}{!}{
\begin{tabular}{|c|c|c|c|c|c|c|}
\hline 
Dataset & \#Tables & \#Rows (Min, Avg, Max) & \#Cols (Min, Avg, Max) & \#Pos. FRs & \#Neg. FRs \\ \hline
\kw{Real-AR}  & 2539   & (28, 578, 16961)   & (3, 12, 400)         & 966       & 8102     \\ \hline
\kw{Real-ST} & 55766  & (5, 98, 23005)   & (32, 113, 500)      & 1938  & 294253  \\ \hline
\kw{Real-FD}  & 374    & (107, 25688, 717321)   &  (4, 22, 333)    & 3510   & 5579     \\ \hline
\end{tabular}
}
\end{table}

\smallskip
\noindent\textbf{Metrics.}
We evaluated the quality of discovered \FRs using \emph{precision}, \emph{recall}, and \emph{F1-score}. 
Here, precision is the fraction of predicted positives that are correct, and recall is the fraction of true positives that are successfully identified. Formally,
\[
P = \frac{\mathrm{TP}}{\mathrm{TP} + \mathrm{FP}}, \quad
R = \frac{\mathrm{TP}}{\mathrm{TP} + \mathrm{FN}}, \quad
F_1 = \frac{2PR}{P+R}.
\]
To summarize performance over all score thresholds, we further reported \prauc (Precision-Recall Area Under the Curve)~\cite{eval}, which measures the area under the precision--recall curve, where a larger value indicates better overall performance.

\begin{table*}[!ht]
\centering
\caption{Quality comparison. Each cell reports \prauc (precision, recall, F1-score); the best \prauc per setting is in bold}
\vspace{-2ex}
\resizebox{\linewidth}{!}{
    \begin{tabular}{|c||c|c||c|c||c|c||c|c||}
    \hline
     & \multicolumn{2}{c||}{Arithmetic Relationships (\kw{Real-AR})}  & \multicolumn{2}{c||}{String Transformations (\kw{Real-ST})}  & \multicolumn{2}{c||}{Functional Dependencies (\kw{Real-FD})} \\ \hline
    Methods    &  Clean          &  Dirty         &   Clean          & Dirty          &  Clean          &  Dirty \\ \hline
    
    \autorelate  & \textbf{0.82} (0.89, 0.85, {0.87}) & \textbf{0.92} (0.91, 0.87, {0.89}) & \textbf{0.84} (0.92, 0.90, {0.91})   & \textbf{0.81} (0.92, 0.86, {0.89})
    & \textbf{0.91} (0.73, 1.0, 0.85) 
    & \textbf{0.92} (0.74, 1.0, 0.85)    \\ \hline
    
    \kw{Explain-Da-V} & {0.53 (0.62, 0.49, 0.55)}        & 0.52 (0.62, 0.49, 0.55)          & {0.01 (0.01, 1.0, 0.01)}            & 0.0 (0.0, 1.0, 0.01)          & {0.35 (0.39, 1.0, 0.56)}             & 0.36 (0.38, 1.0, 0.55)              \\ \hline
    
    
    \kw{GPT-5} & 0.54 (0.36, 0.62, 0.46) & 0.59 (0.61, 0.92, 0.73) & 0.81 (0.84, 0.77, 0.80) & 0.25 (0.39, 0.92, 0.55) & 0.90 (0.82, 0.96, 0.88) &  0.75 (0.72, 1.0, 0.84)   \\ \hline 
    
    
    \kw{CS}           & {0.53 (0.57, 0.90, 0.69)}        & 0.15 (0.18, 0.92, 0.30)          & {0.66 (0.64, 0.89, 0.75)}          & 0.01 (0.02, 0.21, 0.03)          & {0.45 (0.46, 0.75, 0.57)}          & 0.44 (0.43, 0.91, 0.59)              \\ \hline
    \kw{MI}           & {0.43 (0.48, 0.97, 0.64)}        & 0.57 (0.57, 0.89, 0.70)          & {0.01 (0.01, 1.0, 0.02)}            & 0.03 (0.05, 0.86, 0.10)          & {0.38 (0.38, 0.82, 0.52)}          & 0.45 (0.47, 0.83, 0.60)              \\ \hline
    \kw{FI}           & {0.52 (0.53, 0.93, 0.68)}        & 0.18 (0.39, 0.90, 0.55)          & {0.58 (0.59, 0.99, 0.74)}           & 0.02 (0.03, 0.92, 0.07)          & {0.42 (0.43, 0.90, 0.59)}          & 0.42 (0.47, 0.95, 0.63)         \\ \hline
    \kw{RFI$^+$}      & {0.34 (0.51, 0.45, 0.48)}        & 0.08 (0.11, 1.0, 0.19)           & {0.21 (0.34, 0.28, 0.31)}          & 0.03 (0.25, 0.04, 0.06)          & {0.49 (0.46, 0.76, 0.57)}          & 0.53 (0.56, 0.76, 0.64)          \\ \hline
    RFI$^{\prime+}$   & {0.50 (0.55, 0.78, 0.65)}   & 0.49 (0.56, 0.57, 0.56)          & {0.39 (0.49, 0.60, 0.54)}          & 0.01 (0.02, 0.39, 0.03)          & {0.44 (0.45, 0.88, 0.59)}          & 0.55 (0.58, 0.81, 0.68)         \\ \hline
    \kw{SMI}          & {0.26 (0.39, 0.23, 0.29)}        & 0.06 (0.11, 1.0, 0.19)           & {0.23 (0.39, 0.45, 0.42)}          & 0.02 (0.29, 0.04, 0.07)          & {0.49 (0.39, 0.73, 0.51)}          & 0.41 (0.38, 1.0, 0.55)          \\ \hline
    $p_{dep}$       & {0.10 (0.11, 1.0, 0.19)}         & 0.47 (0.60, 0.53, 0.56)          & {0.01 (0.01, 1.0, 0.01)}            & 0.0 (0.0, 0.78, 0.01)            & {0.39 (0.39, 1.0, 0.56)}           & 0.34 (0.43, 1.0, 0.61)          \\ \hline
    $\mu^+$      & {0.43 (0.52, 0.58, 0.55)}        & 0.32 (0.27, 0.48, 0.35)          & {0.26 (0.39, 0.36, 0.37)}          & 0.01 (0.02, 0.33, 0.03)          & {0.49 (0.52, 0.80, 0.63)}          & 0.56 (0.60, 0.81, {0.69})          \\ \hline
    $\tau$       & {0.53 (0.53, 0.93, 0.68)}        & 0.61 (0.39, 0.87, 0.54)          & {0.58 (0.59, 0.99, 0.74)}          & 0.02 (0.04, 0.92, 0.07)          & {0.43 (0.44, 0.90, 0.59)}          & 0.45 (0.52, 0.90, 0.66)           \\ \hline
    $\rho$       & {0.11 (0.11, 1.0, 0.19)}         & 0.52 (0.55, 0.69, 0.61)          & {0.01 (0.01, 1.0, 0.01)}            & 0.0 (0.01, 0.89, 0.01)           & {0.40 (0.40, 1.0, 0.57)}           & 0.37 (0.39, 0.99, 0.55)          \\ \hline
    $g_1$        & {0.11 (0.11, 1.0, 0.19)}         & 0.57 (0.56, 0.82, 0.67)          & {0.01 (0.01, 1.0, 0.01)}            & 0.0 (0.01, 0.88, 0.01)           & {0.38 (0.38, 1.0, 0.55)}           & 0.41 (0.44, 1.0, 0.61)           \\ \hline
    $g_1^s$      & {0.11 (0.11, 1.0, 0.19)}         & 0.51 (0.61, 0.57, 0.59)          & {0.01 (0.01, 1.0, 0.01)}            & 0.0 (0.01, 0.86, 0.01)           & {0.40 (0.40, 1.0, 0.57)}           & 0.40 (0.41, 0.86, 0.55)           \\ \hline
    $g_1^{s+}$   & {0.11 (0.11, 1.0, 0.19)}         & 0.51 (0.61, 0.57, 0.59)          & {0.01 (0.01, 1.0, 0.01)}            & 0.0 (0.01, 0.86, 0.01)           & {0.40 (0.40, 1.0, 0.57)}           & 0.40 (0.38, 1.0, 0.55)              \\ \hline
    $g_2$        & {0.11 (0.11, 1.0, 0.19)}         & 0.56 (0.56, 0.85, 0.68)          & {0.01 (0.01, 1.0, 0.01)}            & 0.0 (0.01, 0.87, 0.01)           & {0.40 (0.40, 1.0, 0.57)}           & 0.40 (0.38, 0.98, 0.55)           \\ \hline
    $g_3$        & {0.11 (0.11, 1.0, 0.19)}         & 0.44 (0.52, 0.53, 0.52)          & {0.01 (0.01, 1.0, 0.01)}            & 0.0 (0.0, 1.0, 0.01)             & {0.40 (0.40, 1.0, 0.57)}           & 0.41 (0.44, 1.0, 0.61)           \\ \hline
    $g_3^{\prime+}$    & {0.08 (0.11, 1.0, 0.19)}   & 0.07 (0.11, 1.0, 0.19)           & {0.01 (0.01, 1.0, 0.01)}            & 0.0 (0.0, 1.0, 0.01)             & {0.43 (0.45, 0.88, 0.59)}          & 0.37 (0.48, 0.92, 0.63)        \\ \hline 
    \end{tabular}
}
\label{tab:Auc}
\end{table*}

\eat{
\begin{table}[!ht]
\centering
\caption{Quality comparison of all methods}
\vspace{-2ex}
\resizebox{\linewidth}{!}{
    \begin{tabular}{|c|c|c|c|c|c|c|}
    \hline
     & \multicolumn{2}{c|}{\kw{Real-AR}}  
     & \multicolumn{2}{c|}{\kw{Real-ST}}  
     & \multicolumn{2}{c|}{\kw{Real-FD}} \\ \hline
    Methods    &  Clean          &  Dirty         &   Clean          & Dirty          &  Clean          &  Dirty \\ \hline
    \autorelate  
    & \textbf{0.82} 
    & \textbf{0.92} 
    & \revise{\textbf{0.84}}   
    & \revise{\textbf{0.70}} 
    & \revise{\textbf{0.57}}   
    & \revise{\textbf{0.71}}    \\ \hline

    \kw{Explain-Da-V} 
    & 0.53        
    & 0.52          
    & 0.01            
    & 0.00          
    & 0.35             
    & 0.36              \\ \hline

    \revise{\kw{GPT-5.4-nano}}        
    & \revise{0.51}         
    & \revise{0.64}          
    & \revise{0.69}          
    & \revise{0.21}          
    & \revise{0.55}          
    & \revise{0.39}         \\ \hline

    \kw{CS}           
    & 0.53        
    & 0.15          
    & 0.66          
    & 0.01          
    & 0.45          
    & 0.44              \\ \hline

    \kw{MI}           
    & 0.43        
    & 0.57          
    & 0.01            
    & 0.03          
    & 0.38          
    & 0.45              \\ \hline

    \kw{FI}           
    & 0.52        
    & 0.18          
    & 0.58           
    & 0.02          
    & 0.42          
    & 0.42         \\ \hline

    \kw{RFI$^+$}      
    & 0.34        
    & 0.08           
    & 0.21          
    & 0.03          
    & 0.49          
    & 0.53          \\ \hline

    RFI$^{\prime+}$   
    & 0.50   
    & 0.49          
    & 0.39          
    & 0.01          
    & 0.44          
    & 0.55         \\ \hline

    \kw{SMI}          
    & 0.26        
    & 0.06           
    & 0.23          
    & 0.02          
    & 0.49          
    & 0.41          \\ \hline

    $p_{dep}$       
    & 0.10         
    & 0.47          
    & 0.01            
    & 0.00            
    & 0.39           
    & 0.34          \\ \hline

    $\mu^+$      
    & 0.43        
    & 0.32          
    & 0.26          
    & 0.01          
    & 0.49          
    & 0.56          \\ \hline

    $\tau$       
    & 0.53        
    & 0.61          
    & 0.58          
    & 0.02          
    & 0.43          
    & 0.45           \\ \hline

    $\rho$       
    & 0.11         
    & 0.52          
    & 0.01            
    & 0.00           
    & 0.40           
    & 0.37          \\ \hline

    $g_1$        
    & 0.11         
    & 0.57          
    & 0.01            
    & 0.00           
    & 0.38           
    & 0.41           \\ \hline

    $g_1^s$      
    & 0.11         
    & 0.51          
    & 0.01            
    & 0.00           
    & 0.40           
    & 0.40           \\ \hline

    $g_1^{s+}$   
    & 0.11         
    & 0.51          
    & 0.01            
    & 0.00           
    & 0.40           
    & 0.40              \\ \hline

    $g_2$        
    & 0.11         
    & 0.56          
    & 0.01            
    & 0.00           
    & 0.40           
    & 0.40           \\ \hline

    $g_3$        
    & 0.11         
    & 0.44          
    & 0.01            
    & 0.00             
    & 0.40           
    & 0.41           \\ \hline

    $g_3^{\prime+}$    
    & 0.08   
    & 0.07           
    & 0.01            
    & 0.00             
    & 0.43          
    & 0.37        \\ \hline 
    \end{tabular}
}
\label{tab:Auc}
\end{table}
}

\begin{figure*}[t]
    \centering
    \includegraphics[scale=0.3]{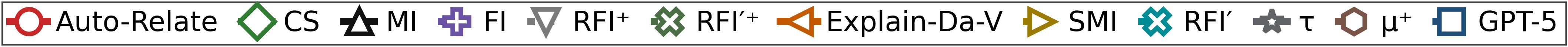}
    \vspace{0.2em}
    
    \begin{subfigure}[b]{0.23\linewidth}        
        \centering
        \includegraphics[width=\linewidth]{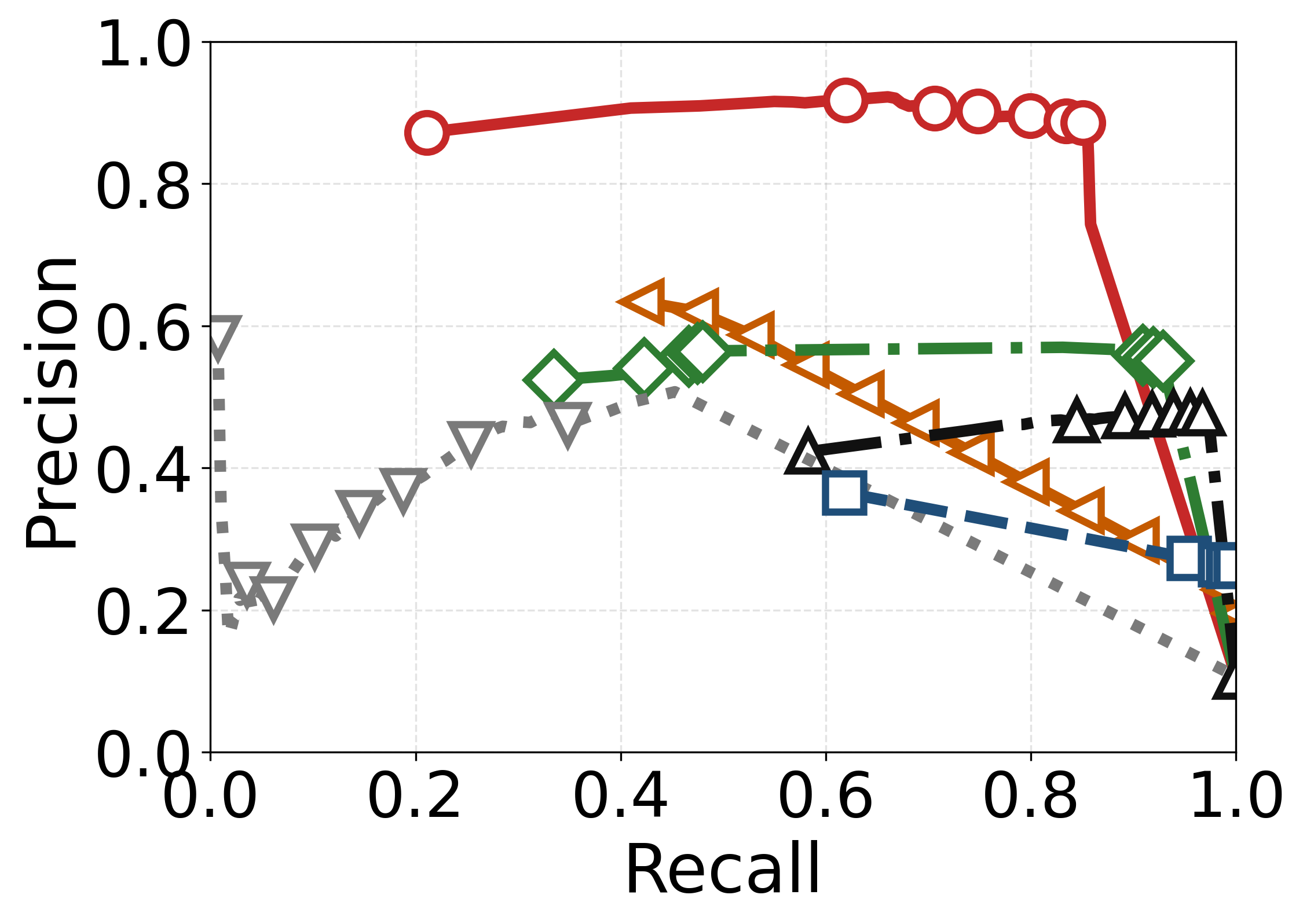}
        \caption{Real-AR (Clean)}
        \label{fig:auc-real-ar-clean}
    \end{subfigure}
    \hfill
    \begin{subfigure}[b]{0.23\linewidth}
        \centering
        \includegraphics[width=\linewidth]{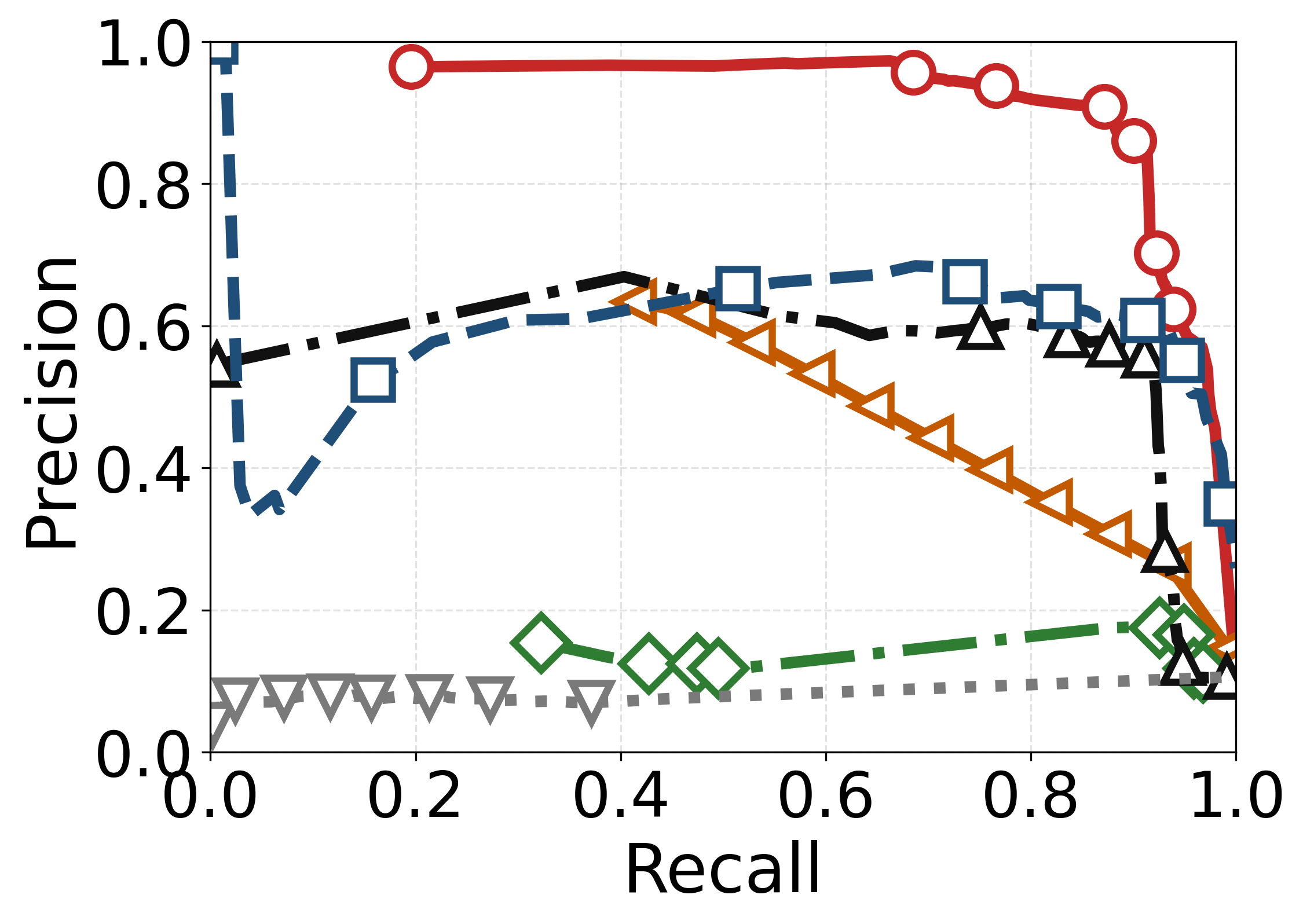}
        \caption{Real-AR (Dirty)}
        \label{fig:auc-real-ar-dirty}
    \end{subfigure}
    \hfill
    \begin{subfigure}[b]{0.23\linewidth}
        \centering
        \includegraphics[width=\linewidth]{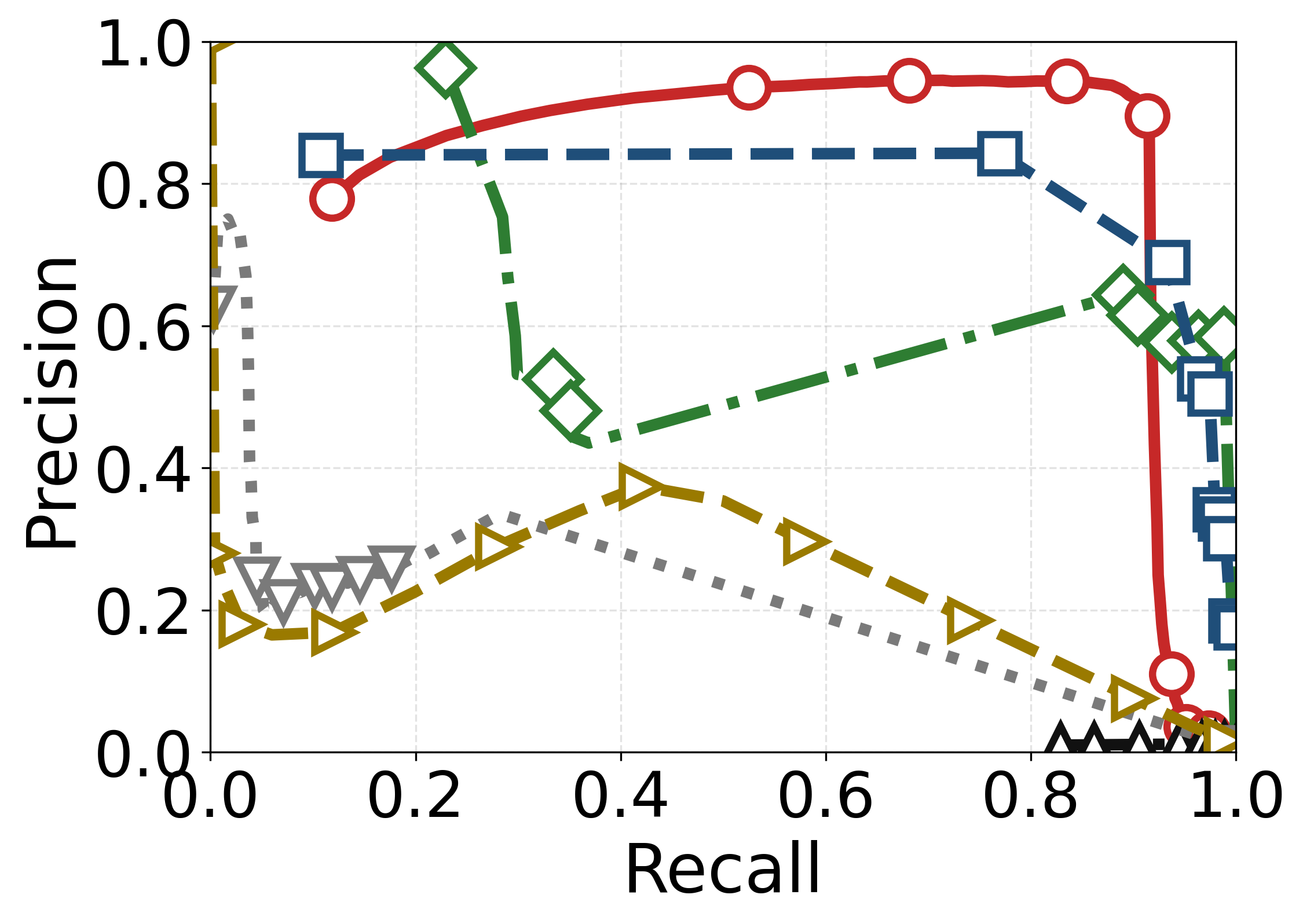}
        \caption{Real-ST (Clean)}
        \label{fig:auc-real-st-clean}
    \end{subfigure}
    \hfill
    \begin{subfigure}[b]{0.23\linewidth}
        \centering
        \includegraphics[width=\linewidth]{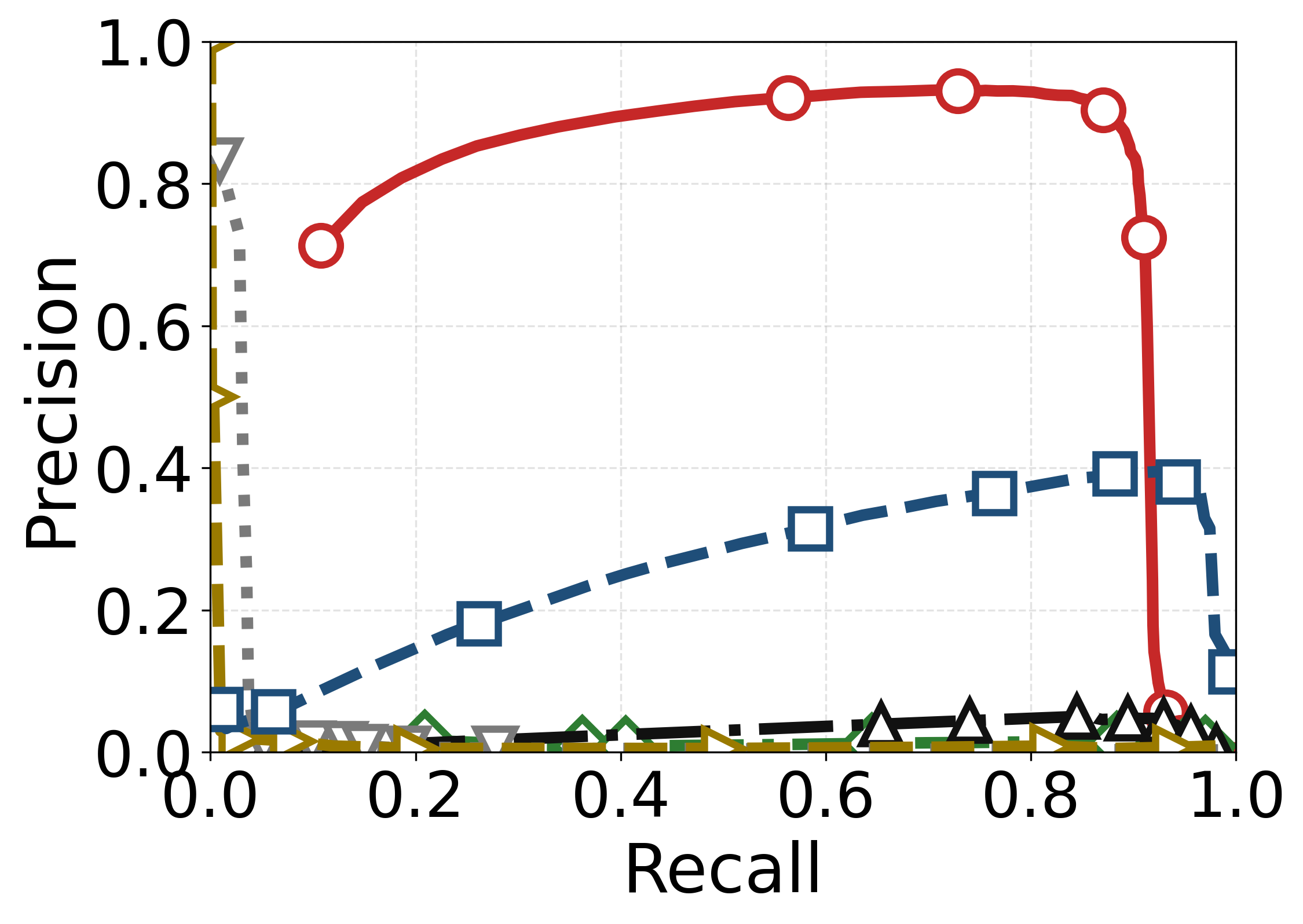}
        \caption{Real-ST (Dirty)}
        \label{fig:auc-real-st-dirty}
    \end{subfigure}


    \begin{subfigure}[b]{0.23\linewidth}
        \centering
        \includegraphics[width=\linewidth]{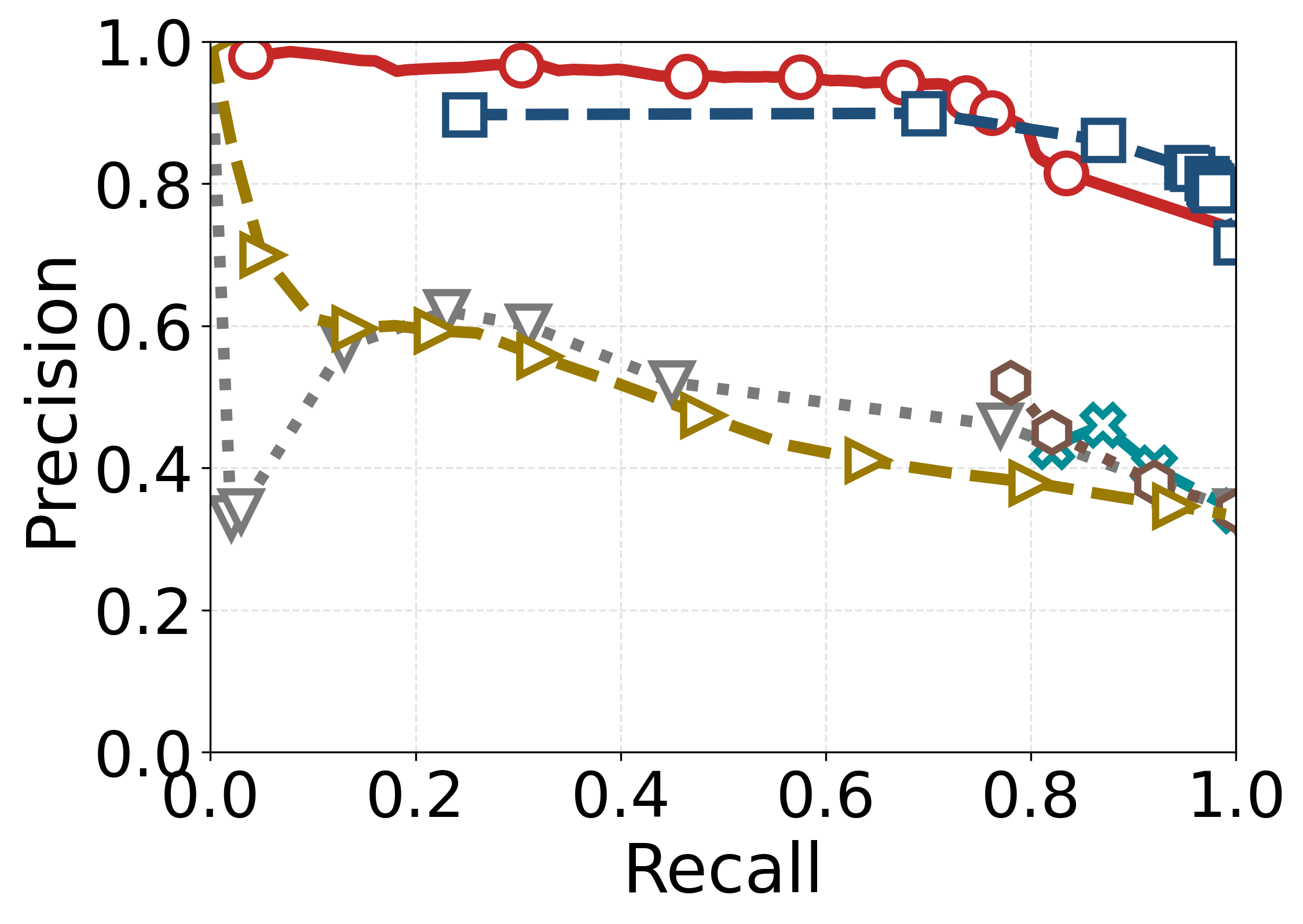}
        \caption{Real-FD (Clean)}
        \label{fig:auc-real-fd-clean}
    \end{subfigure}
    \hfill
    \begin{subfigure}[b]{0.23\linewidth}
        \centering
        \includegraphics[width=\linewidth]{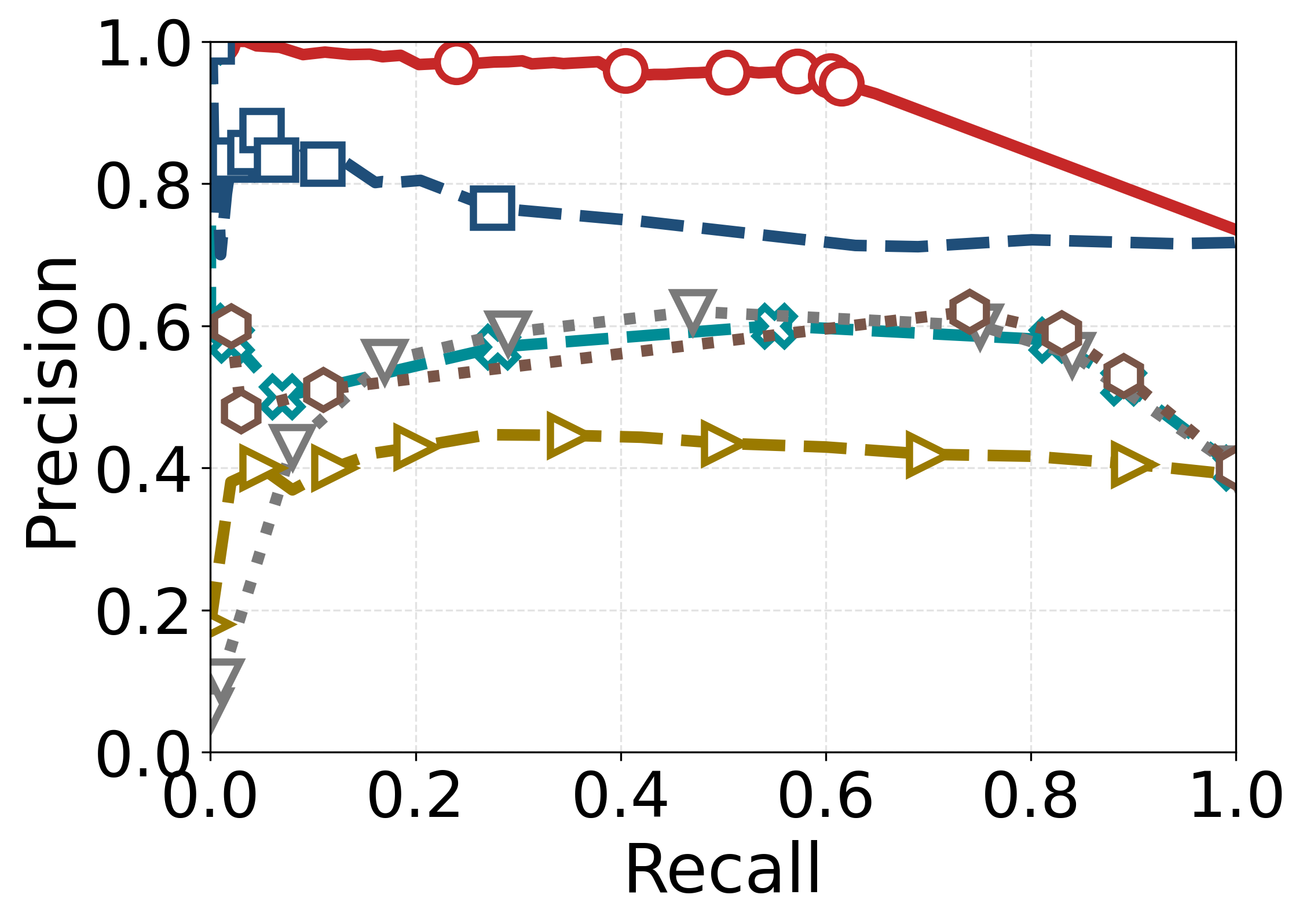}
        \caption{Real-FD (Dirty)}
        \label{fig:auc-real-fd-dirty}
    \end{subfigure}
    \hfill
    \begin{subfigure}[b]{0.23\linewidth}
        \centering
        \includegraphics[width=\linewidth]{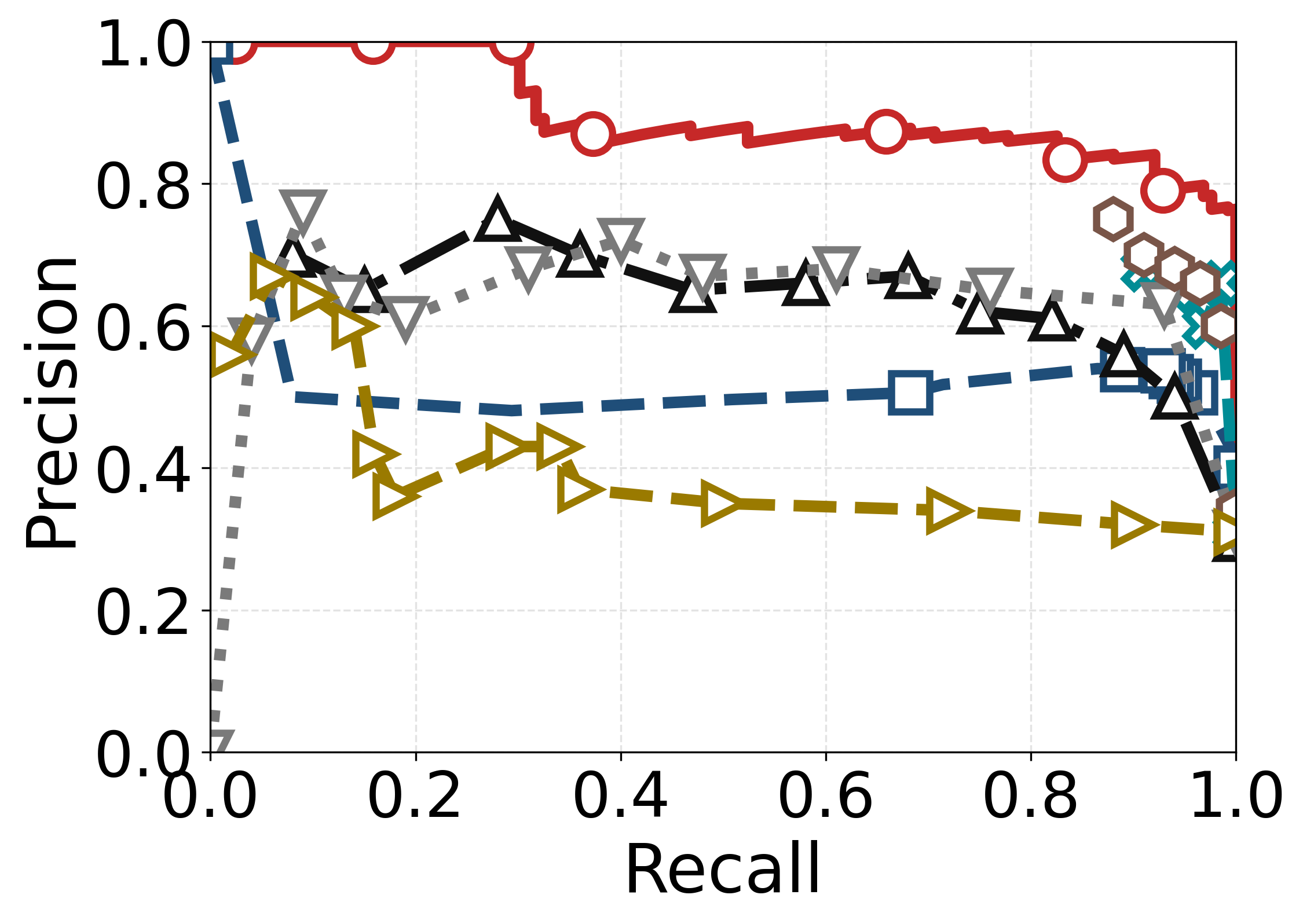}
        \caption{RWD (Clean)}
        \label{fig:auc-rwd-clean}
    \end{subfigure}
    \hfill
    \begin{subfigure}[b]{0.23\linewidth}
        \centering
        \includegraphics[width=\linewidth]{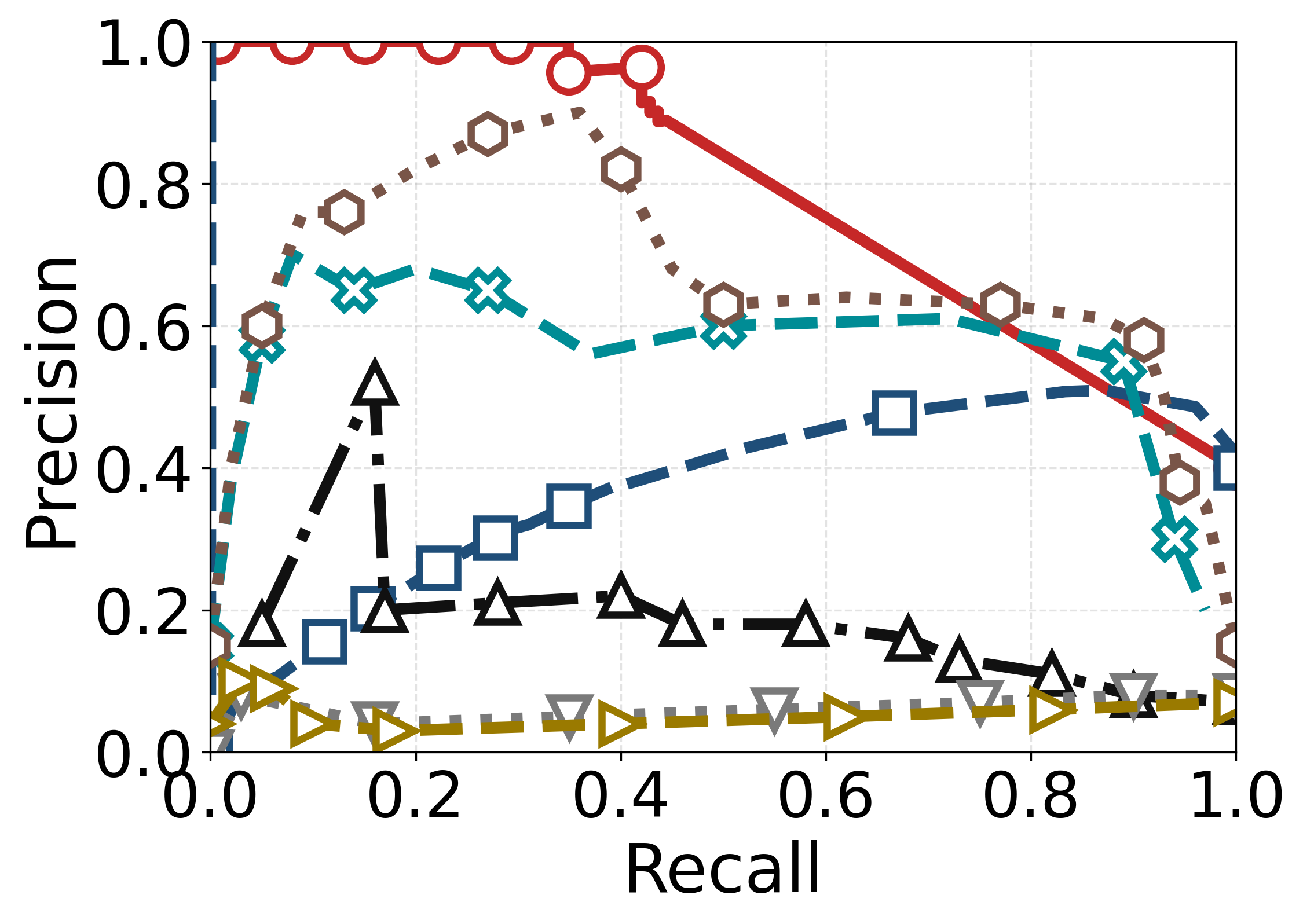}
        \caption{RWD (Dirty)}
        \label{fig:auc-rwd-dirty}
    \end{subfigure}

    \vspace{-2ex}
    \caption{Quality comparisons on the \kw{Real} and \rwd benchmarks, between six best performing methods.
    }
    \label{fig:all-auc-comparisons}
\end{figure*}

\eat{
\begin{table}[!t]
\centering
\caption{\prauc comparison on \rwd}
\vspace{-2ex}
\label{tab:clean-dirty-results}
\begin{tabular}{|c|c|c|}
\hline
\textbf{Method} & \textbf{Clean} & \textbf{Dirty} \\
\hline
\autorelate & \textbf{0.900198} & \textbf{0.797278} \\
\kw{GPT-5}       & 0.521988          & 0.368822          \\
\kw{RFI}'        & 0.055425          & 0.562200          \\
\kw{MI}          & 0.588700          & 0.174200          \\
\kw{RFI$^+$}        & 0.635800          & 0.060800          \\
$\mu^+$     & 0.078850          & 0.668125          \\
\kw{SMI}         & 0.383750                & 0.050650          \\
\hline
\end{tabular}
\end{table}
}

\smallskip
\noindent\textbf{Configuration.}
All experiments were conducted on a server running Ubuntu 20.04 with an Intel(R) Xeon(R) Platinum 8370C CPU (2.80\,GHz), 96\,GB RAM, and a 1.8\,TB SSD.
For \autorelate, we set the perturbation threshold to $\eta = 0.5$, the significance threshold in the Independence Test to $\alpha = 0.05$, the maximum discovery level $\ell_\kw{max} = 4$, and the accuracy threshold to $\tau = 1.0$ in the clean-data setting and $\tau = 0.8$ in the dirty-data setting. 
To exploit the within-level independence of candidates, we set the degree of parallelism to the number of CPU cores.
All baselines were run with the default configurations reported in their original papers.

\subsection{Experimental results}
\label{subsec:exp_result}

\smallskip
\noindent\textbf{Exp-1: Effectiveness of \autorelate.}
We first evaluated the effectiveness of \autorelate and all baselines under the default setting.

\noindent
\emph{\underline{\kw{Real} benchmarks.}}
Table~\ref{tab:Auc} reports the overall quality comparison on the \kw{Real} benchmarks, and Figure~\ref{fig:all-auc-comparisons} shows the corresponding precision--recall curves of the top-six methods.
As shown in Table~\ref{tab:Auc}, \autorelate consistently 
achieves the best \prauc across all settings.
It achieves an average \prauc of 0.87, which is on average 59\% higher than the strongest competing baseline across all settings, with the largest absolute gain of 0.56 on \realST in the dirty setting.
Figure~\ref{fig:all-auc-comparisons} further shows that \autorelate maintains a clear precision--recall advantage in most \kw{Real} benchmark settings.
\looseness=-1
\eat{
As shown in Table~\ref{tab:Auc}, \autorelate achieves the best \kw{PR-AUC} in all settings, e.g., 0.82 and 0.92 on \kw{Real-AR}, \revise{0.84} and \revise{0.7} on \kw{Real-ST}, and \revise{0.57} and \revise{0.71} on \kw{Real-FD} under the clean and dirty settings, respectively. 
This verifies the effectiveness of \autorelate for discovering reliable FRs across different FR types and data conditions.
}

\autorelate delivers consistent gains across all three \FR types. 
On \realAR, it achieves \prauc scores of 0.82 and 0.92 under the clean and dirty settings, respectively, which are 51.9\% and 50.8\% higher than the strongest competing baselines.
On \realST, \autorelate obtains \prauc scores of 0.84 and 0.81, 
ranking first in both settings, 
with a particularly large margin
under dirty data.
%
On \realFD, it 
ranks first in both settings, achieving \prauc scores of 0.91 (clean) and 0.92 (dirty).
While the margin over the strongest baseline is narrow on clean data (0.91 vs. 0.90), 
\autorelate maintains a 22.7\% improvement under dirty data, demonstrating stronger robustness to noise.
These indicate that \autorelate generalizes effectively across \ARs, \STs, and classical \FDs, providing strong reliability discrimination across both clean and dirty data.
\eat{
On \kw{Real-FD}, \autorelate also performs the best in both settings, with PR-AUC \revise{0.57} on clean data and \revise{0.71} on dirty data, compared with the best baseline results of \revise{0.55} (GPT-\revise{5.4-nano}) on clean and 0.56 (\kw{$\mu^+$}) on dirty, respectively.
Although the gap is smaller than that on AR and ST, \autorelate still consistently performs best, showing that the proposed framework generalizes effectively not only to arithmetic and string-based FRs, but also to classical FD discovery.
This is because \autorelate explicitly verifies the reliability of candidate FRs rather than relying solely on association strength or violation statistics.
In particular, the Perturbation Test helps reject exact but coincidental relationships, while the Independence Test in the dirty-data setting further filters out incomplete approximate FRs whose violations are systematically associated with omitted columns.
\looseness=-1
}

\noindent
\emph{\underline{\rwd benchmark.}}
We further compared \autorelate against the top-six best-performing baselines on the \rwd benchmark.
As shown in 
Figures~\ref{fig:auc-rwd-clean} and \ref{fig:auc-rwd-dirty},
\autorelate consistently achieves the best performance under both clean and dirty settings, with \prauc scores of 0.90 and 0.80, respectively.
It outperforms the strongest competing baselines by 41.6\% and 19.3\%,
and maintains higher precision across a wide recall range.
These results confirm its effectiveness on this manually labeled \FD benchmark.

\eat{
\begin{figure}[t]
    \centering
    \includegraphics[width=\columnwidth]{figures/exp/AR_ST.png}
    \caption{Quality comparisons on the \kw{Real-AR} and \kw{Real-ST} benchmarks, between 6 best performing methods.}
    \label{fig:AR-ST}
\end{figure}

\begin{figure}[t]
    \centering
    \includegraphics[width=\columnwidth]{figures/exp/FD.png}
    \caption{Quality comparisons on the \kw{Real-FD} and \kw{RWD}  benchmarks~\cite{parciak2024measuring}, between 6 best performing methods.}
    \label{fig:FD}
\end{figure}
}


\begin{figure*}[!t]
    \centering
    \begin{minipage}{0.85\linewidth}
        \centering
        \begin{subfigure}[b]{0.34\linewidth}
            \centering
            \includegraphics[width=\linewidth]{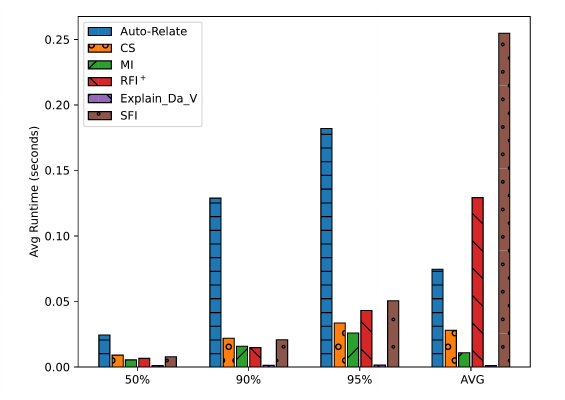}
            \caption{Real-AR dataset}
            \label{fig:running_time_AR}
        \end{subfigure}
        \hfill
        \begin{subfigure}[b]{0.32\linewidth}
            \centering
            \includegraphics[width=\linewidth]{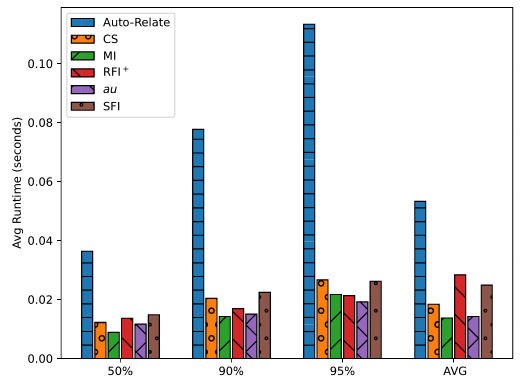}
            \caption{Real-ST dataset}
            \label{fig:running_time_ST}
        \end{subfigure}
        \hfill
        \begin{subfigure}[b]{0.32\linewidth}
            \centering
            \includegraphics[width=\linewidth]{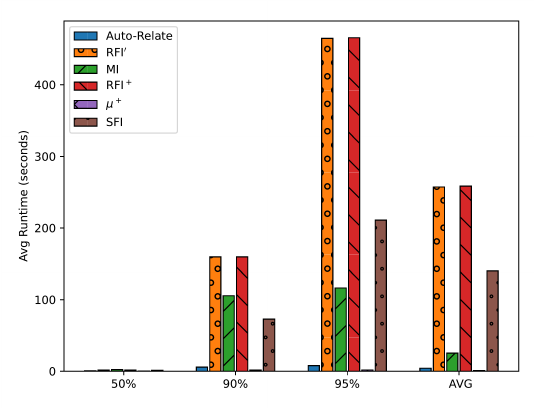}
            \caption{Real-FD dataset}
            \label{fig:running_time_FD}
        \end{subfigure}
    \end{minipage}
    \vspace{-2ex}
    \caption{Per-candidate verification time on the \kw{Real} benchmarks.}
    \label{fig:running_time_all}
    \vspace{-2ex}
\end{figure*}

\begin{figure*}[!t]
\centering
\includegraphics[width=0.9\textwidth]{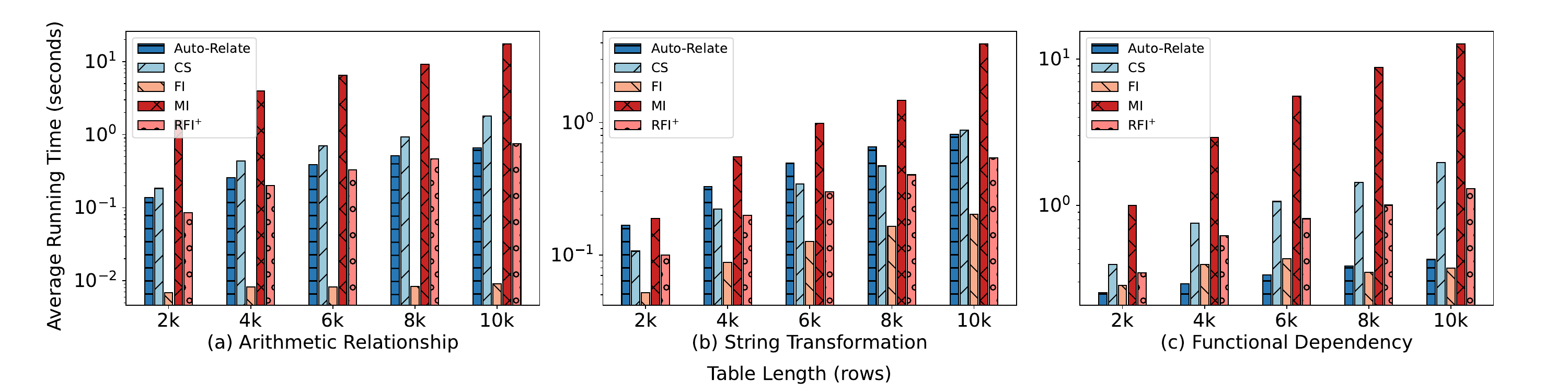}
\vspace{-2ex}
\caption{Scalability with respect to table size on the \kw{Real} benchmarks.}
\label{fig:RowScale}
\vspace{-2ex}
\end{figure*}

\smallskip
\noindent\textbf{Exp-2: Efficiency and Scalability of \autorelate.}
We evaluated the efficiency of \autorelate under the default setting, focusing on its per-candidate running time, scalability with respect to table size, and the practical impact of the proposed optimization strategies.

\noindent\underline{\emph{Comparison vs. baselines.}}
Figure~\ref{fig:running_time_all} reports the per-candidate running time for the reliability-verification stage on each benchmark, excluding candidate generation, which is shared across all methods. 
Specifically, the $x$-axis shows the $50$th, $90$th, and $95$th percentile running times, together with the average running time across all tables. 
On \realAR and \realST, \autorelate processes each table within a fraction of a second across all reported percentiles.
While it is slower than lightweight statistical measures such as \kw{CS} and \kw{MI}, all methods remain in the sub-second regime, keeping \autorelate practical for interactive use.
On \realFD, some baselines (e.g., \kw{RFI$^{\prime+}$}) spend more than 300 seconds per table at the 95th percentile, whereas \autorelate remains substantially faster. 
Although \autorelate is not always the fastest method, its reliability-verification procedure naturally incurs additional cost beyond purely statistical measures. 
Nevertheless, its per-candidate running time remains practical for interactive use while achieving substantially higher accuracy than competing baselines (Table~\ref{tab:Auc}).
\looseness=-1

\noindent\underline{\emph{Varying table size.}}
To evaluate the scalability of \autorelate, we selected tables with more than 1,000 rows and construct datasets of target sizes ranging from 2K to 10K rows via controlled down-sampling or row-level up-sampling. 
As shown in Figure~\ref{fig:RowScale}, the average per-candidate running time generally increases with the number of rows, as expected.
Nevertheless, \autorelate exhibits strong scalability across all \FR types, with consistently low per-candidate running time, while \kw{MI} can exceed 10 seconds per candidate on larger inputs.
Although \autorelate is slower than \kw{FI} and \kw{RFI$^+$} on \realAR and \realST, it achieves substantially better quality (see Table~\ref{tab:Auc}).
These verify that \autorelate scales well with table size while maintaining a practical efficiency-quality trade-off.

\noindent\underline{\emph{Impact of optimization strategies.}}
We evaluated the practical impact of the proposed optimization strategies on \realAR.  
As shown in Table~\ref{tab:figure-runtime-auc-ablation}, the proposed optimizations reduce the per-candidate verification time by 43.5\% in the clean setting (from 35.9 ms to 20.3 ms) and 36.9\% in the dirty setting (from 115.6 ms to 72.9 ms). 
Among them, the closed-form speed-up and binomial bound contribute the most, reducing the per-candidate cost by up to 19.4 ms and 20.6 ms in the dirty setting, respectively.
On \realAR, the closed-form speed-up applies to 44.2\% of sampled candidates, and the binomial bound triggers early termination on every sampled candidate, explaining why they reduce verification time substantially.
The group-by bound has a smaller effect on average per-candidate runtime because 
candidates rejected by it are also quickly eliminated by subsequent tests; nonetheless, it rejects 85.4\% of candidates before the sampling loop and is expected to be more beneficial on wider tables with higher-arity \FRs.

\eat{
As shown in Table~\ref{tab:figure-runtime-auc-ablation}, across all settings, \autorelate is on average \tbf$\times$, \tbf$\times$, \tbf$\times$, and \tbf$\times$ faster than its variants without the group-by bound, closed-form speed-up, binomial bound, and any optimization, respectively, with maximum speedups of \tbf$\times$, \tbf$\times$, \tbf$\times$, and \tbf$\times$.
Moreover, on the clean \realAR benchmark, the group-by bound 
rejects 85.4\% of candidates before sampling, 
the closed-form speed-up applies to 44.2\% of sampled AR candidates, and the binomial bound achieves an overall early-termination effective rate of 100\%.
These verify that the proposed optimizations jointly and substantially reduce the practical cost of the Perturbation Test.
}

\eat{
\begin{table*}[!ht]
\centering
\caption{Ablation test of \autorelate on Real benchmarks}
\vspace{-2ex}
\label{tab:figure-runtime-auc-ablation}
\resizebox{\linewidth}{!}{
\begin{tabular}{|c||cc|cc||cc|cc||cc|cc||}
\hline
 & \multicolumn{4}{c||}{Arithmetic Relationships (\kw{Real-AR})} 
 & \multicolumn{4}{c||}{String Transformations (\kw{Real-ST})} 
 & \multicolumn{4}{c||}{Functional Dependencies (\kw{Real-FD})} \\
\cline{2-5} \cline{6-9} \cline{10-13}
Method 
& \multicolumn{2}{c|}{Clean} & \multicolumn{2}{c||}{Dirty}
& \multicolumn{2}{c|}{Clean} & \multicolumn{2}{c||}{Dirty}
& \multicolumn{2}{c|}{Clean} & \multicolumn{2}{c||}{Dirty} \\
\cline{2-3} \cline{4-5}
\cline{6-7} \cline{8-9}
\cline{10-11} \cline{12-13}
 & Time (ms) & PR-AUC 
 & Time (ms) & PR-AUC
 & Time (ms) & PR-AUC 
 & Time (ms) & PR-AUC
 & Time (ms) & PR-AUC 
 & Time (ms) & PR-AUC \\
\hline
\autorelate      &  20.3   &   0.82    &  72.94   &  0.92     &     &   0.84    &    &   0.70     &     &     0.57  &      &   0.71   \\
No Group-by      & 19.9 &    0.82  &   72.89   &    0.92  &      &      &      &      &      &      &      &      \\
No Closed-form   & 31.3 &   0.83   &   92.26   &   0.92   &      &      &      &      &      &      &      &      \\
No Binomial      & 25.4 &   0.82   &    93.54  &   0.92   &      &      &      &      &      &      &      &      \\
No Optimization  & 35.9 &   0.83   &   115.63   &   0.92   &      &      &      &      &      &      &      &      \\
\hline
\end{tabular}
}
\end{table*}
}

\begin{table}[!t]
\centering
\caption{Ablation of optimizations on \realAR}
\vspace{-2ex}
\label{tab:figure-runtime-auc-ablation}
\resizebox{0.9\linewidth}{!}{
\begin{tabular}{|c|cc|cc|}
\hline
\multirow{2}{*}{Method} 
& \multicolumn{2}{c|}{Clean} 
& \multicolumn{2}{c|}{Dirty} \\
\cline{2-3} \cline{4-5}
 & Time (ms) & \prauc 
 & Time (ms) & \prauc \\
\hline
\autorelate      & 20.3 & 0.82 & 72.9  & 0.92 \\
No Group-by      & 21.6 & 0.82 & 72.9  & 0.92 \\
No Closed-form   & 31.3 & 0.83 & 92.3  & 0.92 \\
No Binomial      & 25.4 & 0.82 & 93.5  & 0.92 \\
No Optimization  & 35.9 & 0.83 & 115.6 & 0.92 \\
\hline
\end{tabular}
}
\vspace{-2ex}
\end{table}

\smallskip
\noindent\textbf{Exp-3: Sensitivity Analysis.}
We evaluated the sensitivity of \autorelate to 
the threshold $\eta$, noise rate, and the size $|C_\Psi|$ in \FRs.

\begin{figure*}[!t]
    \centering
    \includegraphics[scale=0.3]{figures/exp/latest/figure5_6_unified_legend.png}
    \vspace{0.2em}
    
    \begin{subfigure}[b]{0.23\linewidth}
        \centering
        \includegraphics[width=\linewidth]{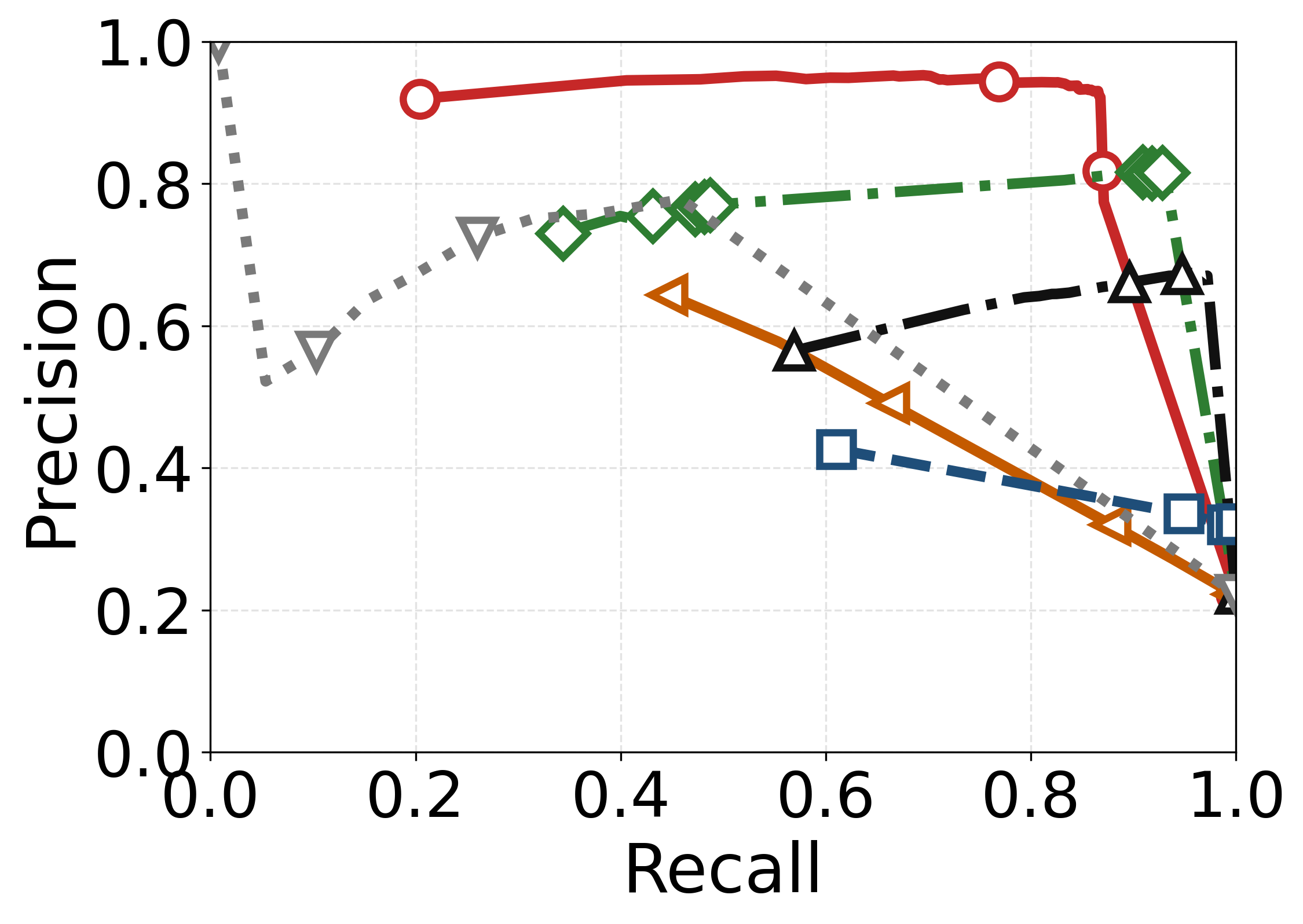}
        \caption{Clean Real-AR: 3-column ARs}
        \label{fig:auc-default-3-column-ar-clean}
    \end{subfigure}
    \hfill
    \begin{subfigure}[b]{0.23\linewidth}
        \centering
        \includegraphics[width=\linewidth]{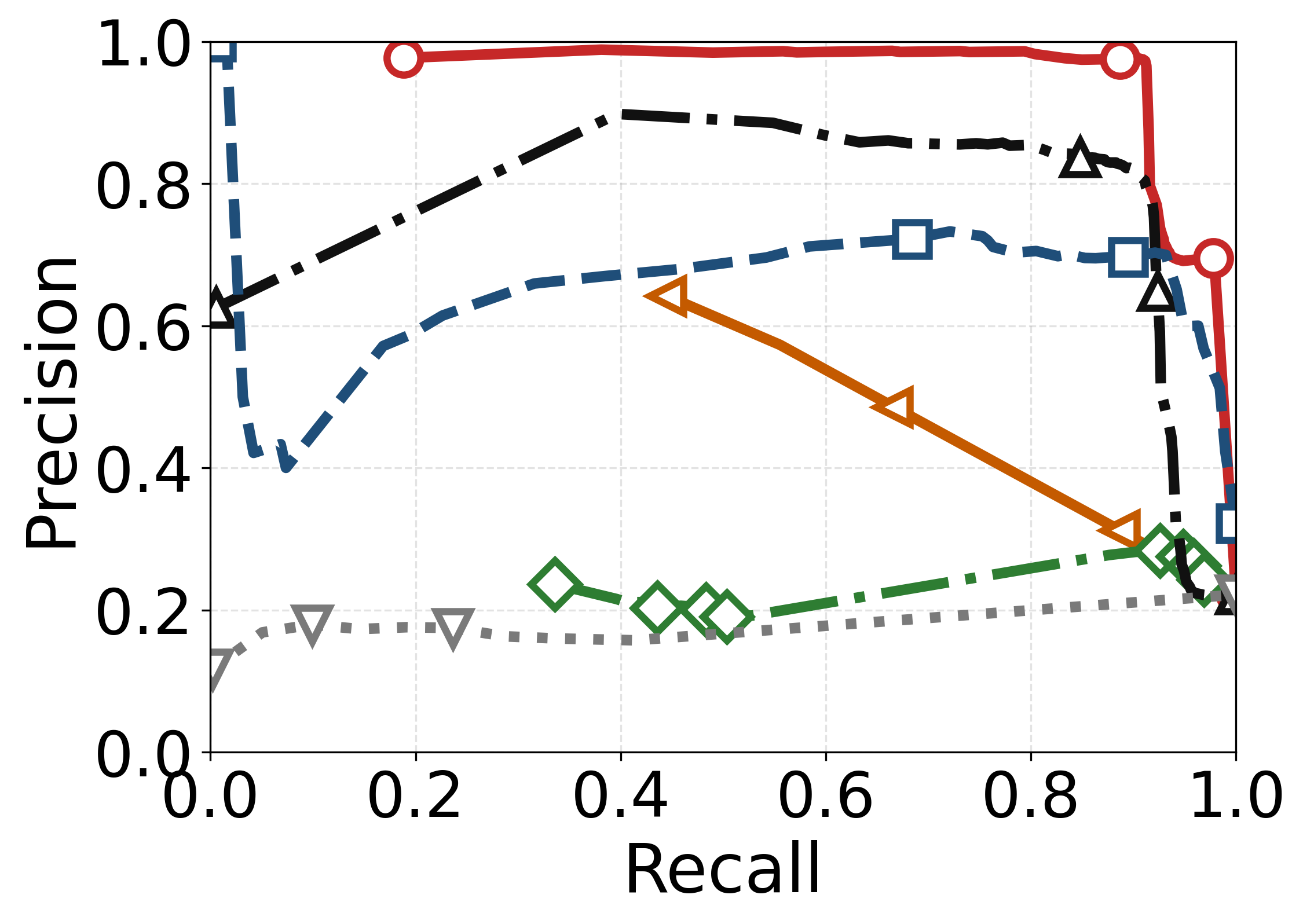}
        \caption{Dirty Real-AR: 3-column ARs}
        \label{fig:auc-default-3-column-ar-dirty}
    \end{subfigure}
    \hfill
    \begin{subfigure}[b]{0.23\linewidth}
        \centering
        \includegraphics[width=\linewidth]{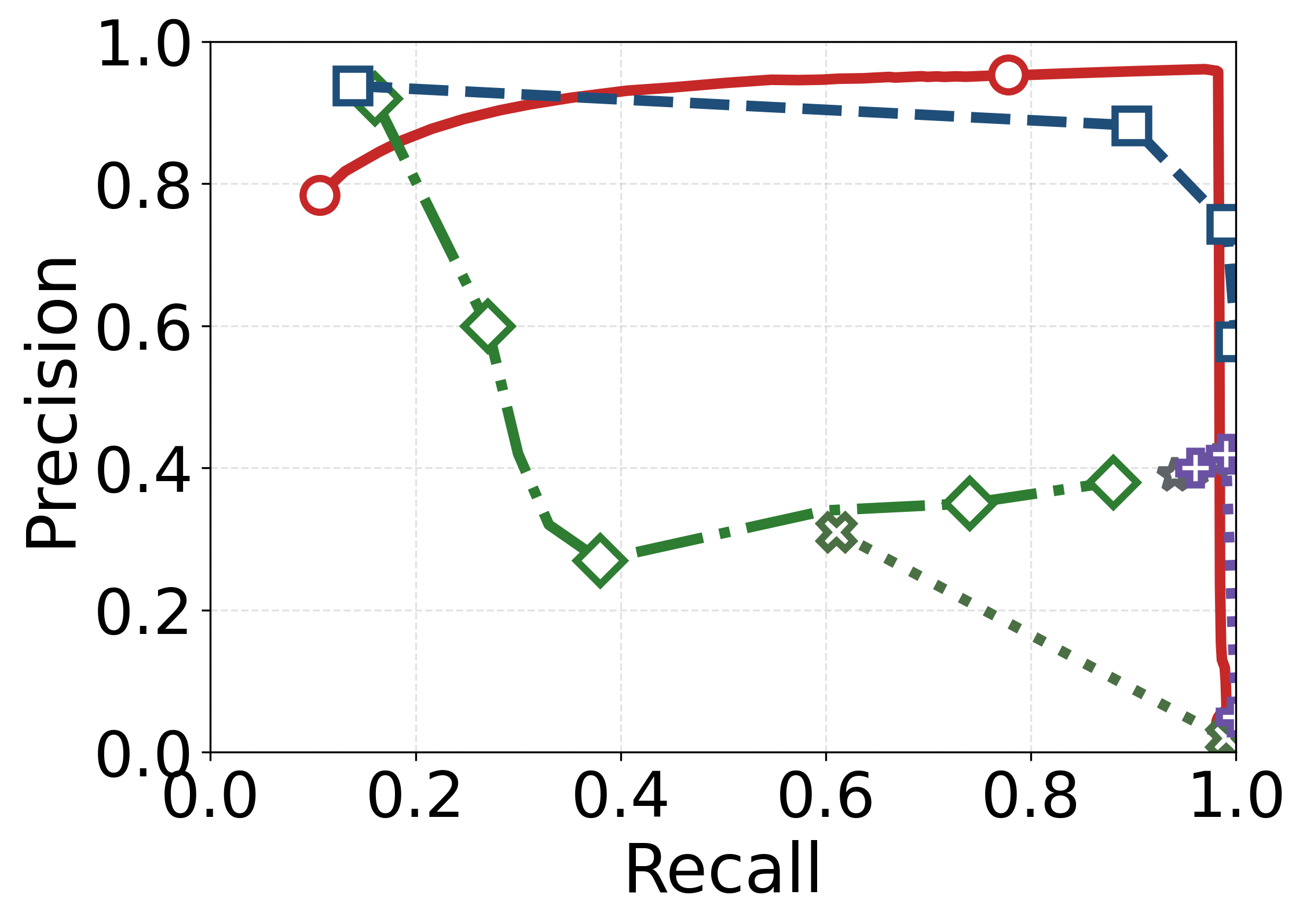}
        \caption{Clean Real-ST: 2-column STs}
        \label{fig:auc-default-2-column-st-clean}
    \end{subfigure}
    \hfill
    \begin{subfigure}[b]{0.23\linewidth}
        \centering
        \includegraphics[width=\linewidth]{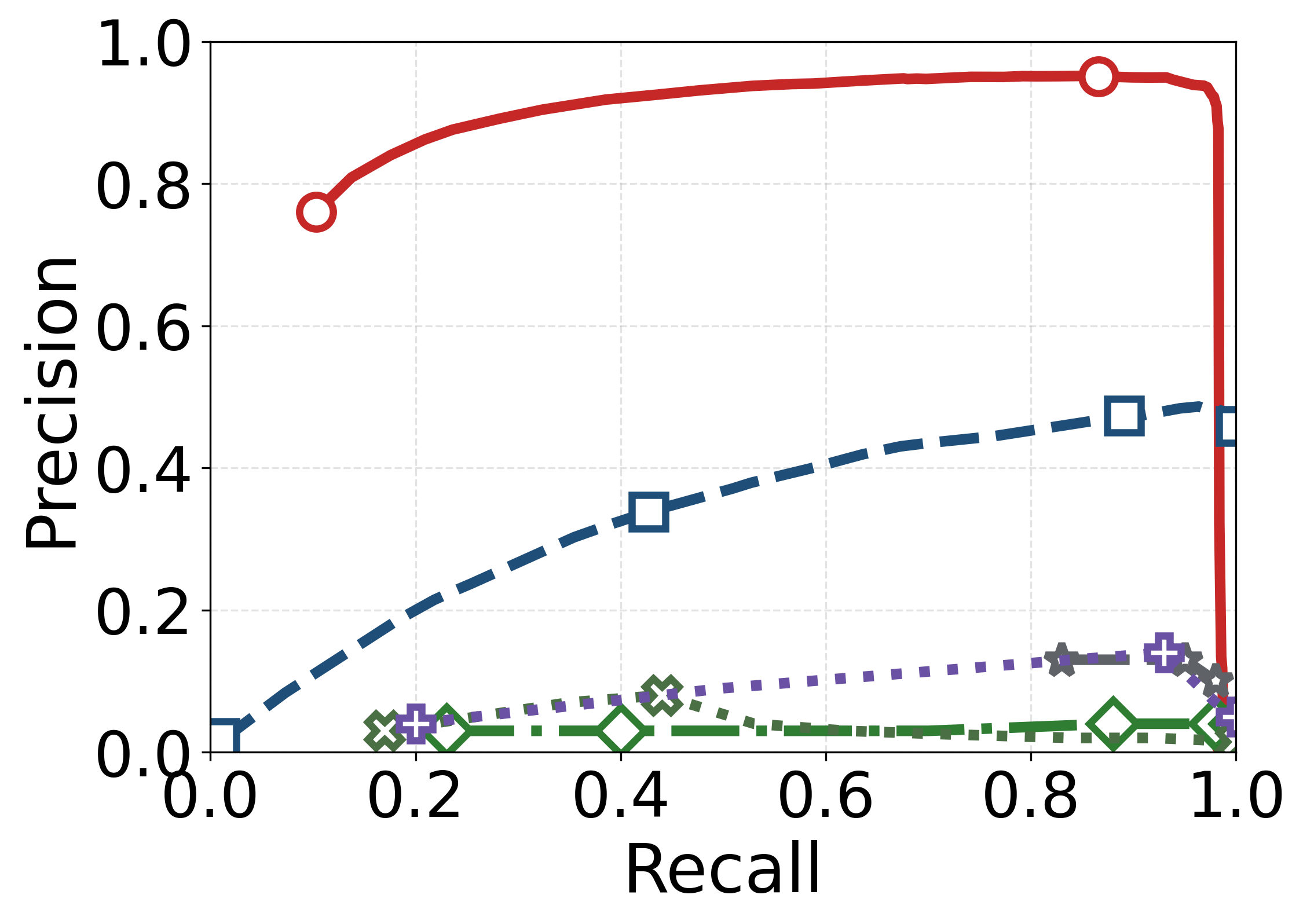}
        \caption{Dirty Real-ST: 2-column STs}
        \label{fig:auc-default-2-column-st-dirty}
    \end{subfigure}
    
    \begin{subfigure}[b]{0.23\linewidth}
        \centering
        \includegraphics[width=\linewidth]{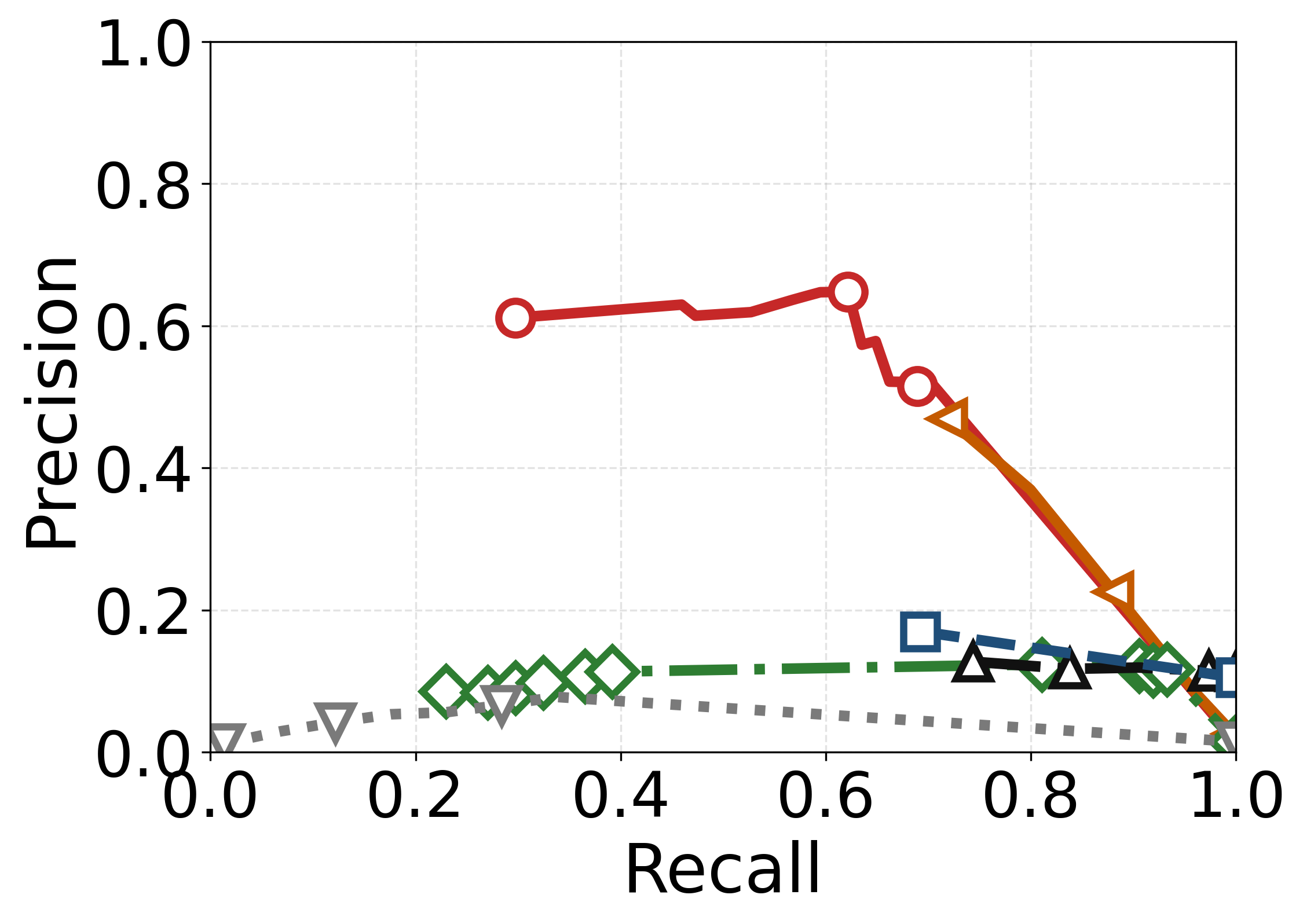}
        \caption{Clean Real-AR: 4-column ARs}
        \label{fig:auc-default-4-column-ar-clean}
    \end{subfigure}
    \hfill
    \begin{subfigure}[b]{0.23\linewidth}
        \centering
        \includegraphics[width=\linewidth]{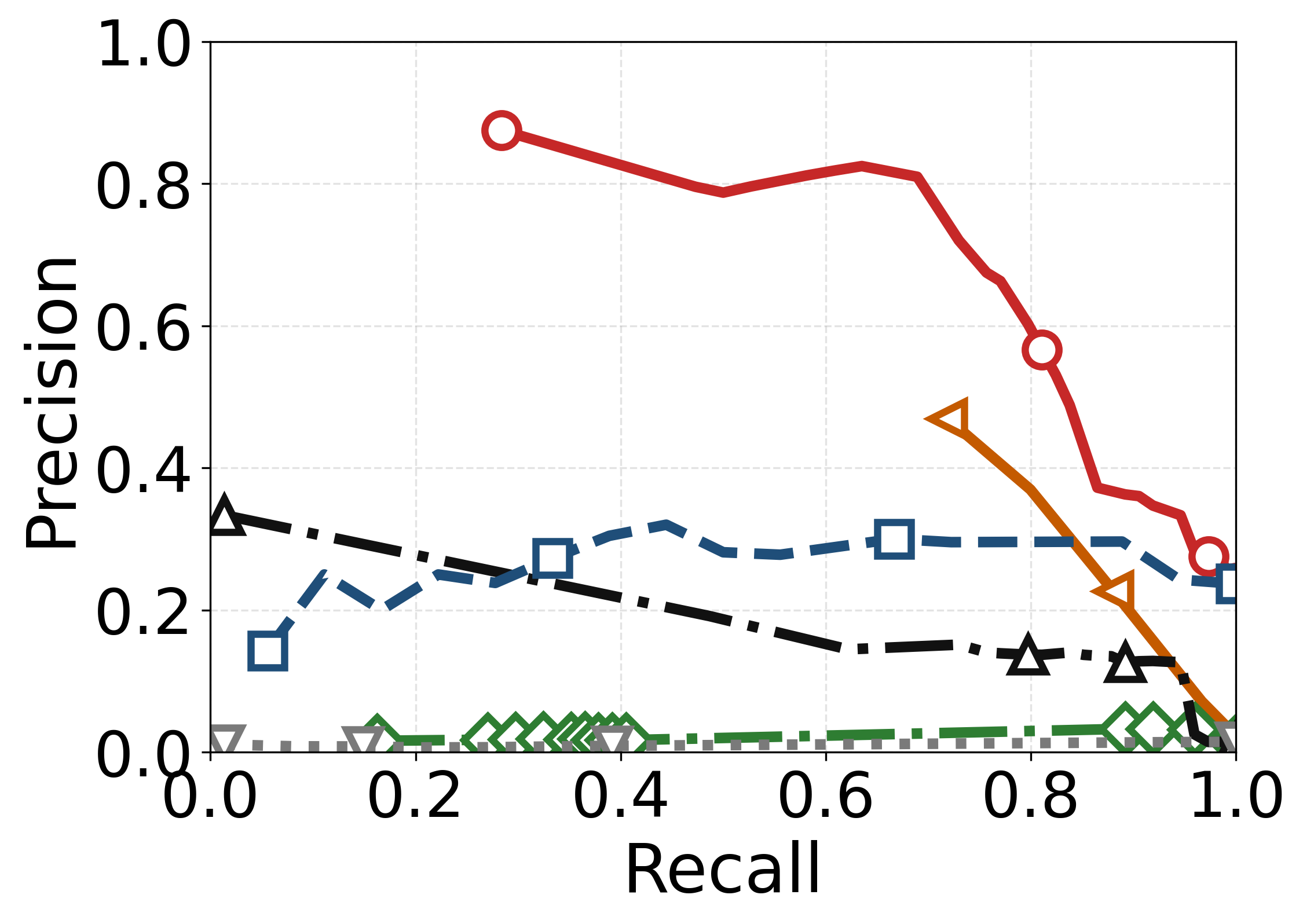}
        \caption{Dirty Real-AR: 4-column ARs}
        \label{fig:auc-default-4-column-ar-dirty}
    \end{subfigure}
    \hfill
    \begin{subfigure}[b]{0.23\linewidth}
        \centering
        \includegraphics[width=\linewidth]{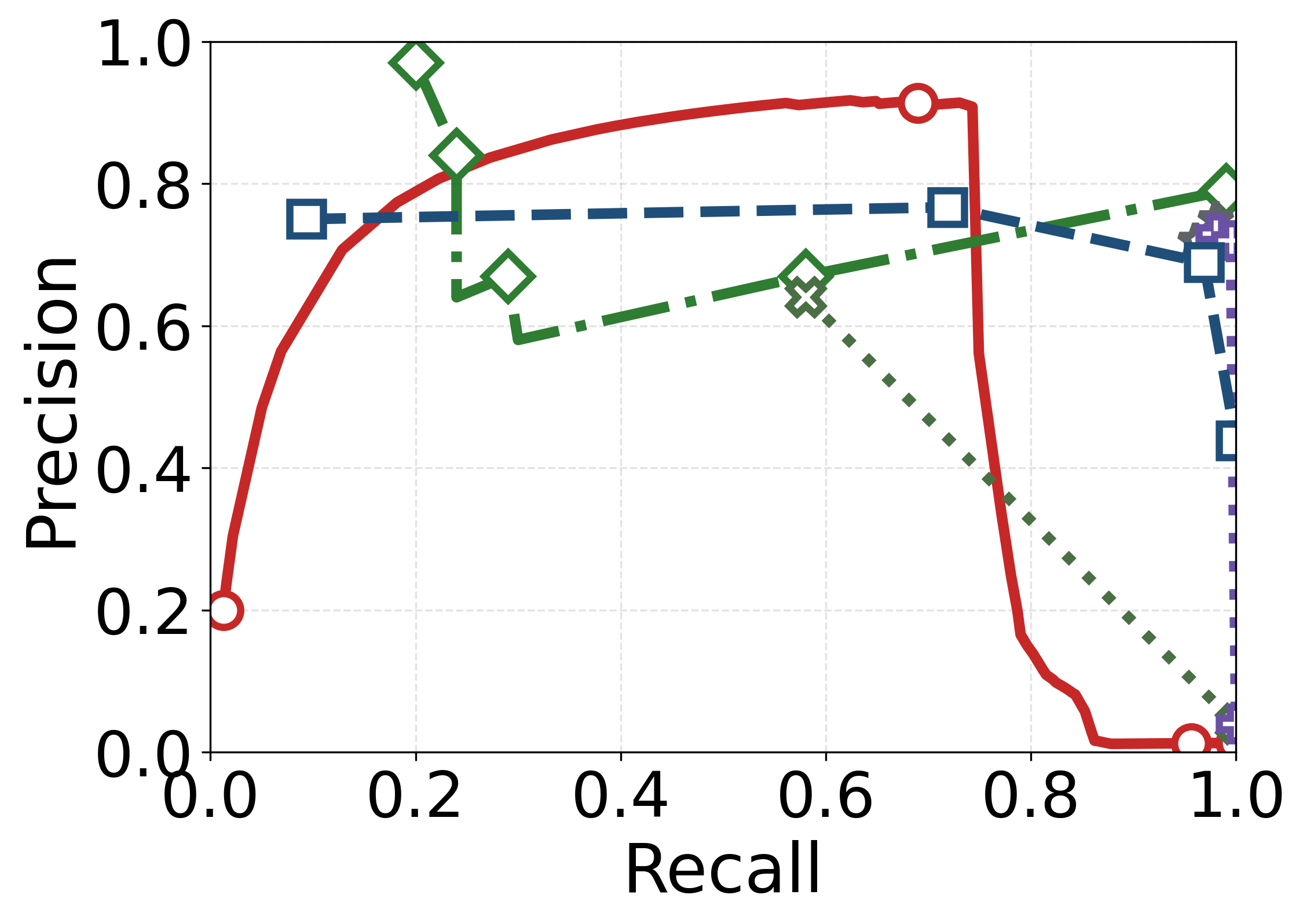}
        \caption{Clean Real-ST: 3-column STs}
        \label{fig:auc-default-3-column-st-clean}
    \end{subfigure}
    \hfill
    \begin{subfigure}[b]{0.23\linewidth}
        \centering
        \includegraphics[width=\linewidth]{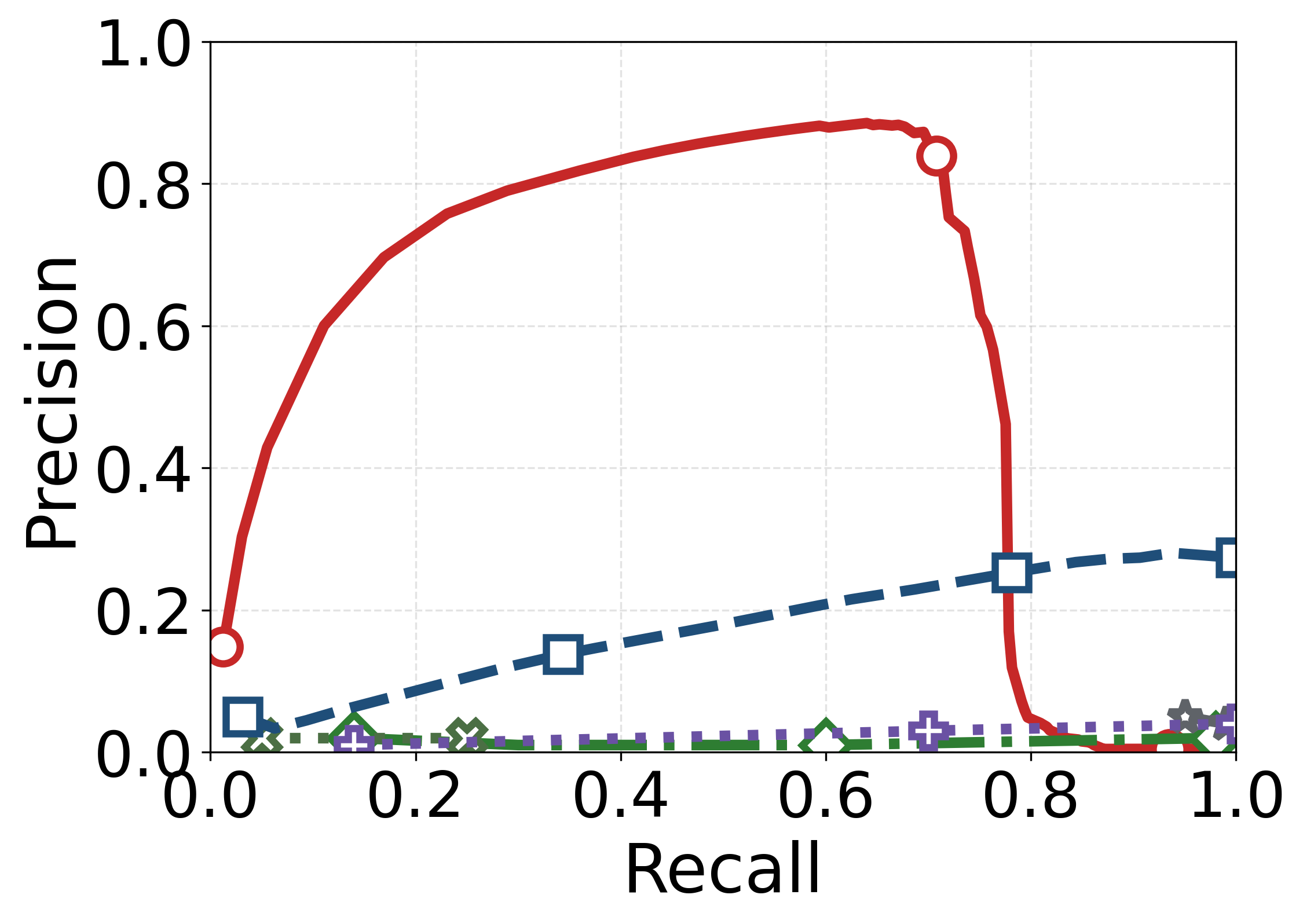}
        \caption{Dirty Real-ST: 3-column STs}
        \label{fig:auc-default-3-column-st-dirty}
    \end{subfigure}

    \vspace{-2ex}
    \caption{Sensitivity to the number of participative columns on \realAR and \realST.}
    \label{fig:sensitivityTest_AR_ST}
    \vspace{-2ex}
\end{figure*}

\eat{
\begin{figure}[t]
\centering
\includegraphics[width=\columnwidth]{figures/exp/sensitivity_columns_AR.pdf}
\caption{Comparison of methods on 3/4 columns functional relationships on the \kw{Real-AR} dataset.}
\label{fig:sensitivityAR}
\end{figure}
}

\eat{
\begin{figure}[t]
\centering
\includegraphics[width=\columnwidth]{figures/exp/sensitivity_columns_ST.pdf}
\caption{Comparison of methods on 2/3 columns functional relationships on the \kw{Real-ST} dataset.}
\label{fig:sensitivityST}
\end{figure}
}

\begin{figure}[!t]
    \centering
    \begin{subfigure}[b]{0.48\linewidth}
        \centering
        \includegraphics[width=\linewidth]{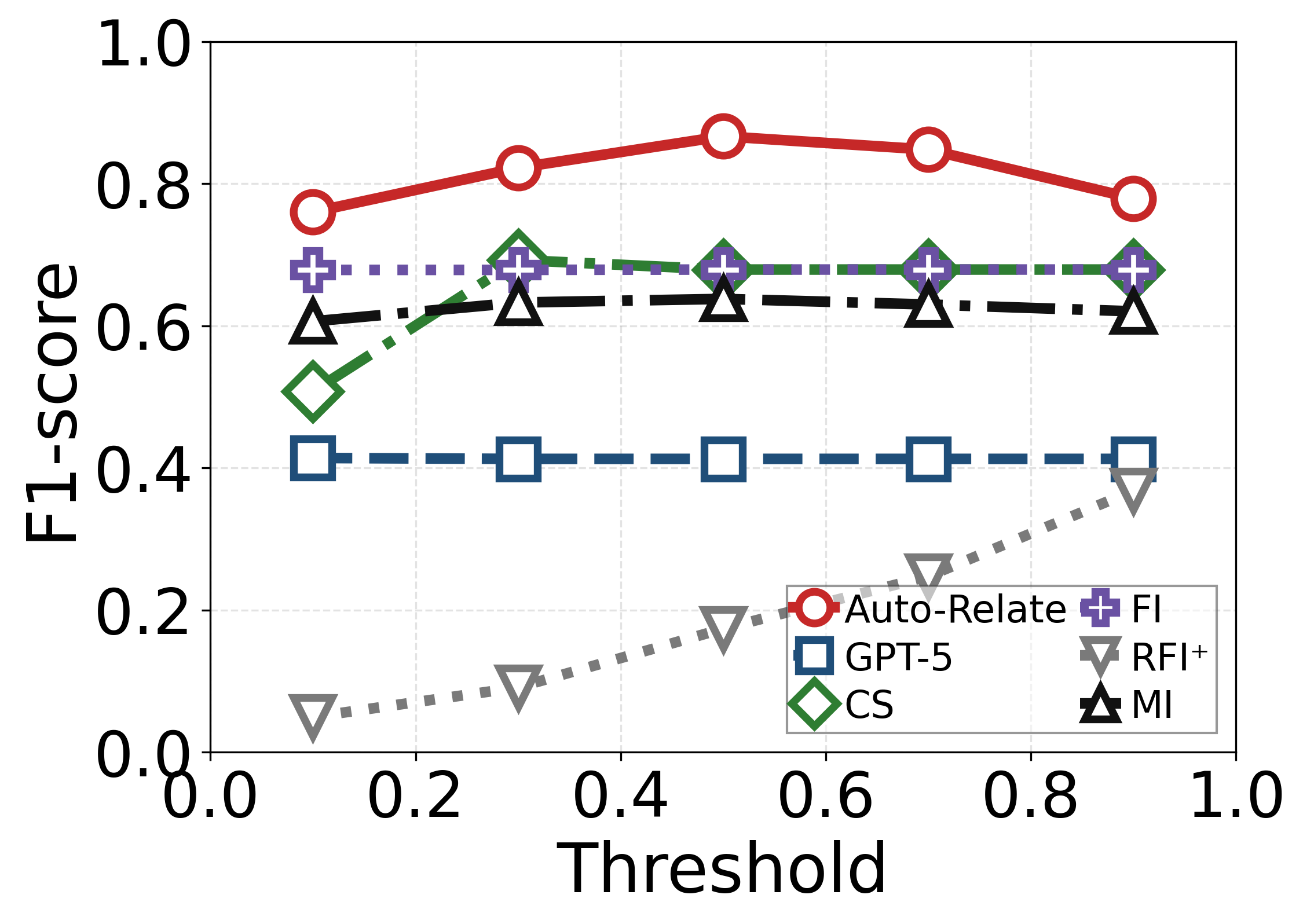}
        \caption{Varying 
        threshold $\eta$}
        \label{fig:sensitivity-real-ar-vary-perturbation-threshold}
    \end{subfigure}
    \hfill
    \begin{subfigure}[b]{0.48\linewidth}
        \centering
        \includegraphics[width=\linewidth]{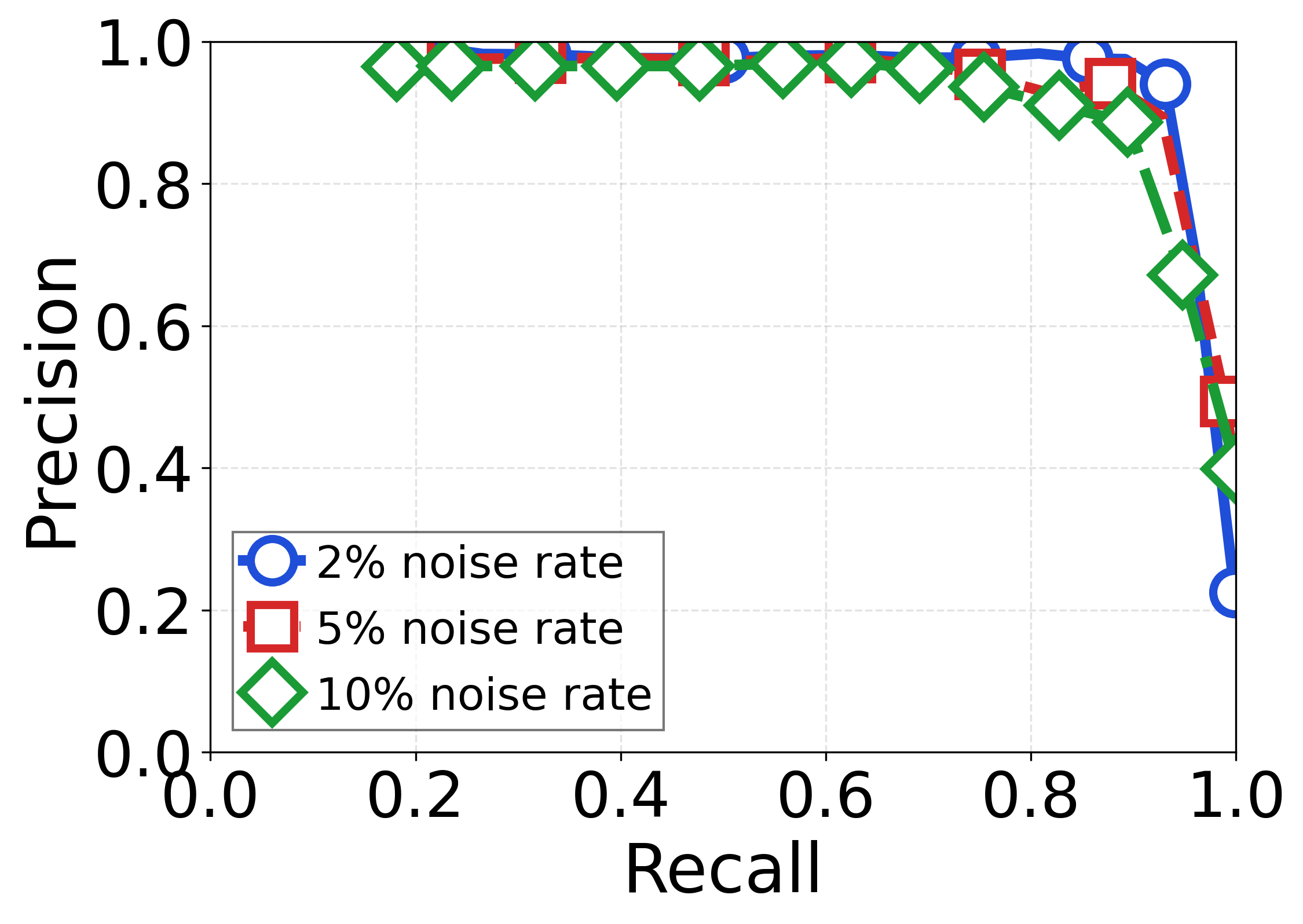}
        \caption{Varying noise rate}
        \label{fig:sensitivity-real-ar-vary-noise-rate}
    \end{subfigure}

    \vspace{-2ex}
    \caption{Sensitivity to threshold $\eta$ and noise rate on \realAR.}
    \label{fig:SensitivityAnalysis}
    \vspace{-2ex}
\end{figure}

\eat{
\begin{figure}[t]
\centering
\includegraphics[width=\columnwidth]{figures/exp/sensitivity_total.pdf}
\caption{Sensitivity analysis on the Real-AR benchmark.}
\label{fig:SensitivityThreshold}
\end{figure}
}

\noindent\underline{\emph{Varying the perturbation threshold.}}
We first varied the perturbation threshold $\eta$ from 0.1 to 0.9 on \realAR and report the resulting F1-scores.
As shown in Figure~\ref{fig:sensitivity-real-ar-vary-perturbation-threshold}, \autorelate consistently achieves the best F1-score across a broad range of threshold settings, 
peaking around the default choice $\eta = 0.5$.
In contrast, most baselines, with the exception of \kw{RFI$^+$}, exhibit relatively stable but consistently lower F1-score curves across different thresholds.
These demonstrate that \autorelate is robust to the perturbation threshold.
\looseness=-1

\noindent\underline{\emph{Varying noise rates.}}
We varied the noise rate in the dirty-data setting from 2\% to 10\% and plot the resulting precision--recall curves of \autorelate on \realAR.
As shown in Figure~\ref{fig:sensitivity-real-ar-vary-noise-rate}, \autorelate exhibits limited sensitivity to the noise rate.
While higher noise rates lead to a slight precision drop in the high-recall region, the overall precision--recall trade-off remains stable, verifying the robustness of \autorelate across varying noise rates.

\eat{
\noindent\underline{\emph{Varying the number of distinct values.}}
Figure~\ref{fig:DistinctValues} reports the PR-AUC and F1-score of \autorelate under different numbers of distinct values in the participative columns, across AR, ST, and FD discovery.
For ARs, both PR-AUC and F1-score remain consistently high across all settings, indicating strong robustness to changes in domain diversity.
For STs, the performance is also stable overall, although some degradation is observed in the largest-cardinality groups, especially in terms of PR-AUC.
For FDs, although the absolute scores are lower than those for ARs and STs, the variation across groups remains limited.
These verify that the effectiveness of \autorelate is not overly sensitive to the diversity of the value domain.
\looseness=-1

\noindent\underline{\emph{Varying the number of rows.}}
As shown in Figure~\ref{fig:DistinctRows}, we further evaluated the performance under varying numbers of distinct rows.
For ARs, \autorelate maintains high PR-AUC and F1-score across all row-count groups, with only a slight drop on the smallest tables, where limited data diversity makes discrimination inherently harder.
For STs, the performance remains stable overall and becomes particularly strong once the table contains a moderate number of distinct rows.
For FDs, both metrics vary only mildly across groups, showing that \autorelate remains effective over a wide range of table sizes.
\looseness=-1

\begin{figure*}[t]
\centering
\includegraphics[width=0.8\textwidth]{figures/exp/distinct_num_values.pdf}
\caption{Comparison of methods on distinct values functional relationships.}
\label{fig:DistinctValues}
\end{figure*}

\begin{figure*}[t]
\centering
\includegraphics[width=0.8\textwidth]{figures/exp/distinct_num_rows.pdf}
\caption{Comparison of methods on the number of distinct rows.}
\label{fig:DistinctRows}
\end{figure*}
}

\noindent\underline{\emph{Varying the number of participative columns.}}
We evaluated how the number of participative columns $|C_\Psi|$ affects the \FR discovery performance.
Figure~\ref{fig:sensitivityTest_AR_ST} reports the precision--recall curves on \realAR and \realST under both clean and dirty settings.
For \AR discovery, we compare the 3-column and 4-column settings; for \ST discovery, we compare the 2-column and 3-column settings.
As expected, with $|C_\Psi|$ increasing, the task becomes more challenging and the performance of all methods generally decreases, especially on dirty data.
Nevertheless, \autorelate consistently achieves the best \prauc and maintains a clear precision--recall advantage across both datasets and settings.
These show that \autorelate remains effective when \FRs involve more participative columns.\looseness=-1

\eat{
For AR discovery, Figures~\ref{fig:auc-default-3-column-ar-clean}--\ref{fig:auc-default-3-column-ar-dirty} and \ref{fig:auc-default-4-column-ar-clean}--\ref{fig:auc-default-4-column-ar-dirty} report the precision--recall curves under 3-column and 4-column settings on both clean and dirty data.
As $|C_\Psi|$ increases from 3 to 4, tasks become more challenging for all methods, and performance generally degrades, as expected.
Nevertheless, \autorelate consistently achieves the best \kw{PR-AUC}, with a particularly clear advantage in the dirty-data setting.
This indicates that \autorelate is robust to the structural complexity of higher-arity arithmetic relationships.
%
For ST discovery, Figure~\ref{fig:auc-default-2-column-st-clean}--\ref{fig:auc-default-2-column-st-dirty} and \ref{fig:auc-default-3-column-st-clean}--\ref{fig:auc-default-3-column-st-dirty} presents a similar comparison under 2-column and 3-column settings.
A consistent trend is observed; as $|C_\Psi|$ increased from 2 to 3, performance degraded across all methods, with the drop being more pronounced on dirty data.
Nevertheless, \autorelate retained the best precision--recall trade-off in all settings, demonstrating that its reliability-driven design generalizes to higher-arity string transformations.
}

\begin{figure}[t]
    \centering
    \begin{subfigure}[b]{0.48\linewidth}
        \centering
        \includegraphics[width=\linewidth]{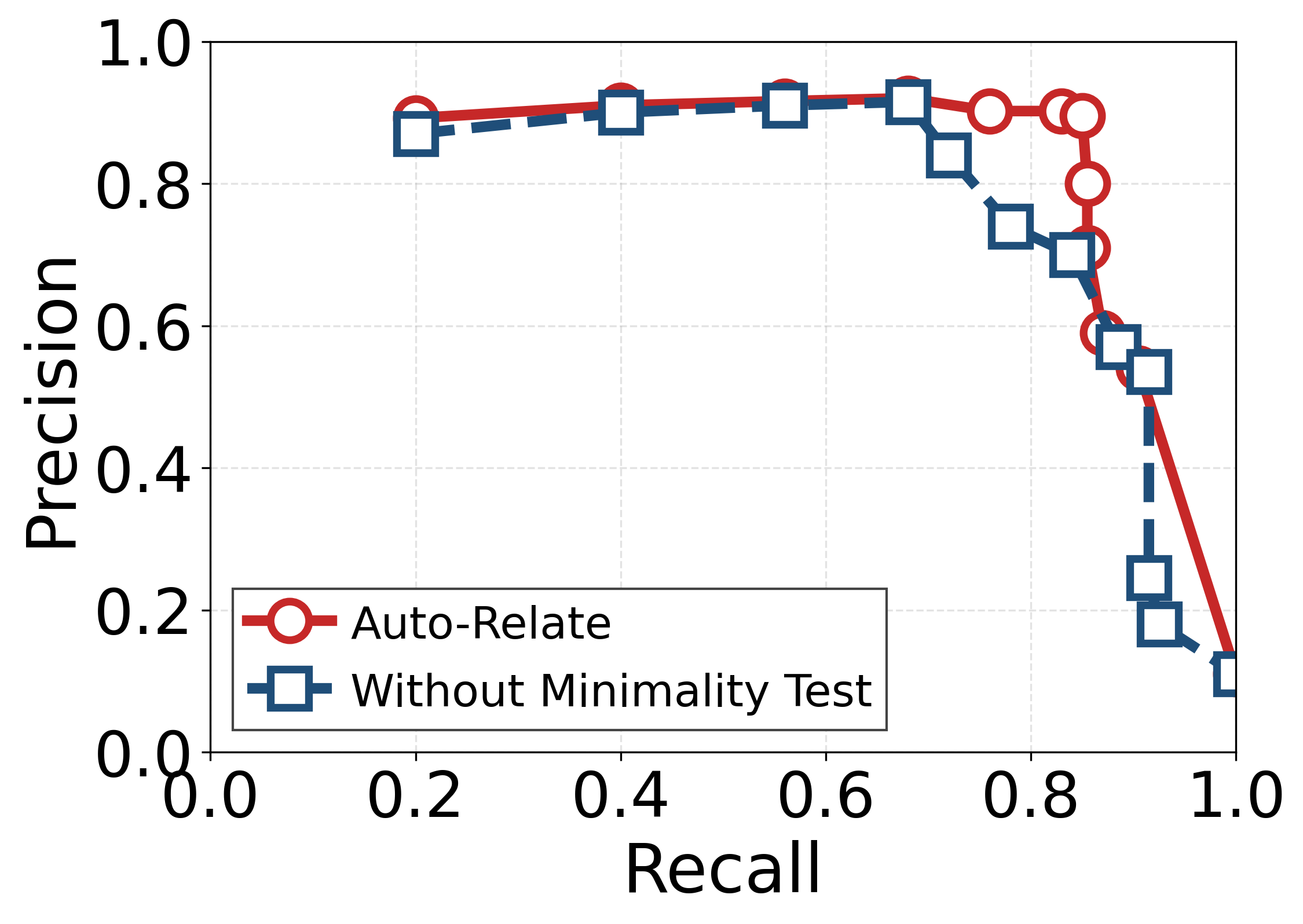}
        \caption{No Minimality Test}
        \label{fig:ablation-no-minimality-test}
    \end{subfigure}
    \hfill
    \begin{subfigure}[b]{0.48\linewidth}
        \centering
        \includegraphics[width=\linewidth]{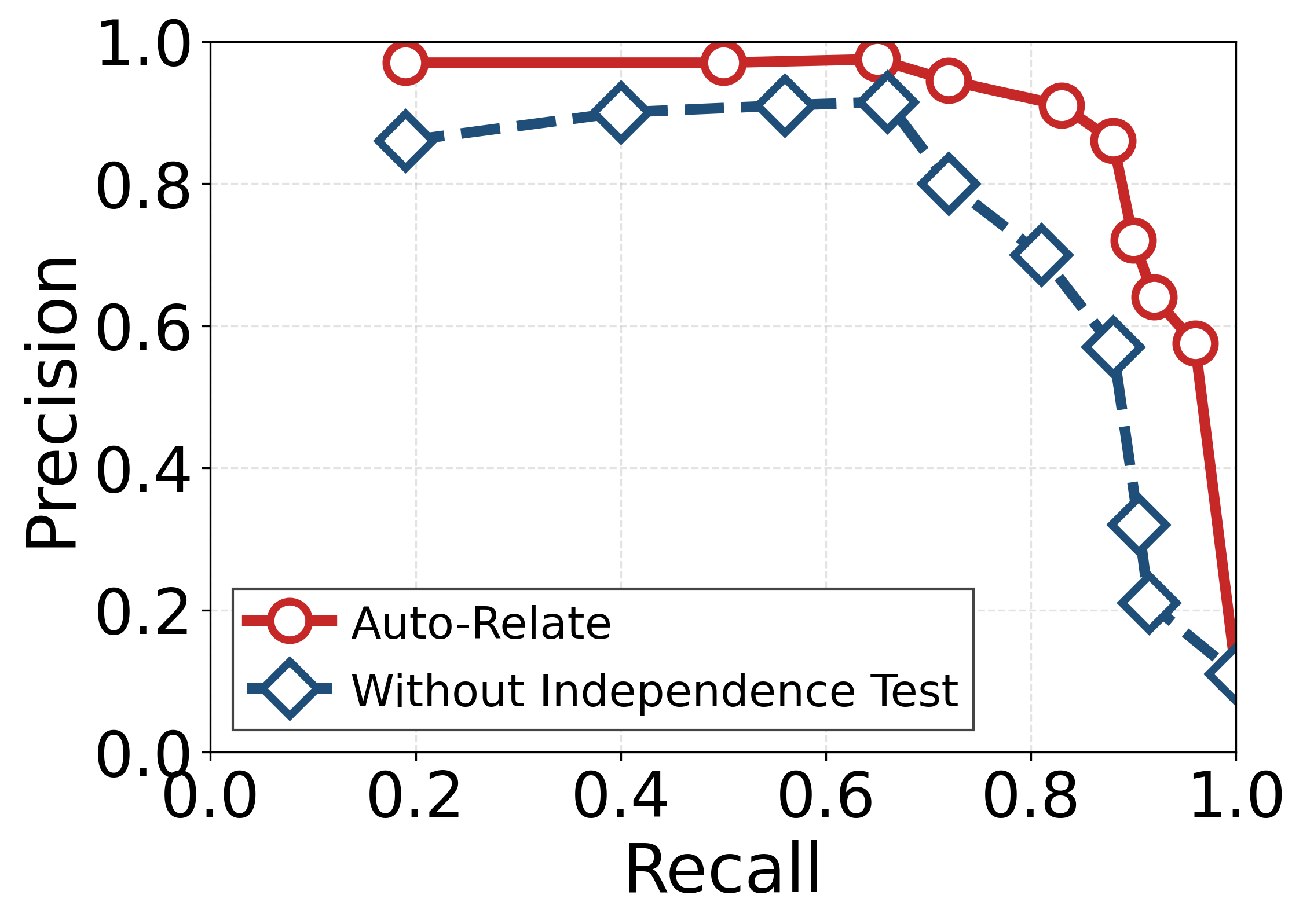}
        \caption{No Independence Test}
        \label{fig:ablation-no-independence-test}
    \end{subfigure}

    \vspace{-2ex}
    \caption{Ablation of reliability tests on \realAR.}
    \label{fig:AblationTest}
    \vspace{-2ex}
\end{figure}

\eat{ 
\begin{figure}[t]
\centering
\includegraphics[width=\columnwidth]{figures/exp/Ablation_total.pdf}
\caption{Ablation study using the Real-AR benchmark.}
\label{fig:AblationTest}
\end{figure}
}

\smallskip
\noindent\textbf{Exp-4: Ablation Study.}
We 
evaluated the impact of reliability tests and optimization strategies in \autorelate.

\noindent\underline{\emph{No Minimality Test.}}
To evaluate the effectiveness of the Minimality Test, we compared \autorelate and its variant without the minimality test, as shown in Figure~\ref{fig:ablation-no-minimality-test}.
The two curves remain relatively close in the low-to-mid recall range, but diverge noticeably in the high-recall region, where the variant without the Minimality Test suffers a clearer precision drop.
This indicates that atomicity-based filtering is effective in suppressing redundant \FRs and improving the overall quality of the returned candidates.

\noindent\underline{\emph{No Independence Test.}}
To evaluate the effectiveness of the Independence Test, we compared \autorelate with a variant that omits the Independence Test, as shown in Figure~\ref{fig:ablation-no-independence-test}.
Removing the Independence Test leads to a clear precision loss, 
with the gap being especially evident between recall 0.6 and 0.9.
This confirms that the Independence Test plays a critical role in filtering incomplete candidates, 
and is therefore particularly important for maintaining the reliability of \autorelate in the dirty-data setting.
\looseness=-1


\eat{
\begin{table*}[t]
\centering
\caption{Ablation test of \autorelate on Real benchmarks}
\label{tab:figure-runtime-real}
\resizebox{\linewidth}{!}{
\begin{tabular}{l|rr|rr|rr}
\toprule
 & \multicolumn{2}{c|}{Real-AR} 
 & \multicolumn{2}{c|}{Real-ST} 
 & \multicolumn{2}{c}{Real-FD} \\
\cmidrule(lr){2-3} \cmidrule(lr){4-5} \cmidrule(lr){6-7}
Method 
& Time (ms) & PR-AUC 
& Time (ms) & PR-AUC 
& Time (ms) & PR-AUC \\
\midrule
\autorelate      &      &      &      &      &      &      \\
No Group-by      & 19.9 &      &      &      &      &      \\
No Closed-form   & 31.3 &      &      &      &      &      \\
No Binomial      & 25.4 &      &      &      &      &      \\
No Optimization  & 35.9 &      &      &      &      &      \\
\bottomrule
\end{tabular}
}
\end{table*}
}

\noindent\underline{\emph{No Optimization Strategies.}} 
We examined whether the proposed optimization strategies affect the quality of \FR discovery.
As shown in Table~\ref{tab:figure-runtime-auc-ablation}, removing individual
or all optimization strategies 
has little effect on \prauc across all settings.
Together with the runtime analysis in Exp-2, this confirms that the proposed optimizations mainly affect efficiency, enabling \autorelate to reduce verification cost while preserving its discovery effectiveness.

\section{Related Work}
\label{sec:related}

\noindent\textbf{Dependency and Relationship Discovery.}
Functional dependencies (\FDs) and approximate functional dependencies (\AFDs) have been extensively studied in data profiling and data management.
Representative discovery algorithms (e.g., \cite{huhtala1999tane,novelli2001fun,wyss2001fastfds,abedjan2014dfd}) 
focus on efficiently enumerating minimal \FDs from relational data, 
with a comprehensive experimental comparison provided by~\cite{papenbrock2015functional}.
More recent work continues to improve \FD discovery, e.g., via hitting-set-based enumeration~\cite{bleifuss2024discovering}, mixed-type extensions~\cite{mandros2020discovering},
and meaningfulness-aware discovery~\cite{wei2023towards},
while relaxed and approximate variants have been surveyed in~\cite{caruccio2015relaxed} and recently advanced in~\cite{caruccio2021discovering,li2025efficient}.
Beyond \FDs, related dependency formalisms have been studied as well, including conditional functional dependencies (\CFDs)~\cite{fan2008conditional,bohannon2007conditional}, denial constraints (\DCs)~\cite{chu2013discovering,dcfinder}, matching dependencies (\MDs)~\cite{schirmer2020efficient,mdedup}, and entity enhancing rules (\REEs)~\cite{REE-discovery-sampling,REE-discovery-topk,REE-discovery-topk-diversified,REE-discovery-topk-diversified-fast}.
These methods are complementary to \autorelate,  
as they mainly focus on enumerating or mining dependency candidates, 
whereas our focus is on reliability verification of discovered relationships.

A related line of work studies \emph{measures} for ranking approximate \kw{FD} candidates, 
including violation-based measures~\cite{giannella2004approximation,kivinen1995approximate,parciak2024measuring}, co-occurrence-based measures~\cite{ilyas2004cords}, information-theoretic measures such as Mutual Information (\kw{MI})~\cite{cheng1997learning} and Fraction of Information (\kw{FI})~\cite{cavallo1987theory,giannella2004approximation}, and probabilistic measures~\cite{piatetsky1993measuring,goodman1979measures}.
A recent comparative study~\cite{parciak2024measuring} systematically evaluates these measures.
While useful for scoring and ranking 
candidates, these methods primarily assess candidates through observed-table signals such as violations, association strength, or uncertainty reduction, rather than explicitly verifying structural reliability.
\eat{
Violation-based measures such as $g_1$, $g_2$, and $g_3$~\cite{giannella2004approximation,kivinen1995approximate,parciak2024measuring} score candidates based on the extent of violations, e.g., the fraction of violating tuple pairs or the amount of row removal needed to restore consistency.
Co-occurrence-based measures, such as $\rho$ in \kw{CORDS}~\cite{ilyas2004cords}, rely on ratios of distinct value combinations.
Information-theoretic measures, including Mutual Information (\kw{MI})~\cite{cheng1997learning} and Fraction of Information (\kw{FI})~\cite{cavallo1987theory,giannella2004approximation}, 
quantify the reduction in uncertainty of the RHS given the LHS.
Probabilistic measures such as $p_{dep}$, $\mu^+$, and $\tau$~\cite{piatetsky1993measuring,goodman1979measures} instead rely on conditional agreement probabilities. 
A recent comparative study~\cite{parciak2024measuring} provides a systematic evaluation of these measures.
}
\eat{Beyond \kw{FD}-specific measures, recent work has studied more general rule discovery under broader evaluation criteria, including objective evaluation metrics~\cite{REE-discovery-sampling}, subjective metrics~\cite{REE-discovery-topk}, diversity-aware ranking~\cite{REE-discovery-topk-diversified}, and user-guided preference modeling~\cite{REE-discovery-topk-diversified-fast}. 
These studies focus on the efficient discovery and ranking of general rules, whereas \autorelate verifies the structural reliability of functional relationships.}

Most closely related to our work are methods for \emph{reliable dependency discovery}. 
\kw{RFI}~\cite{mandros2017discovering,mandros2018discovering,mandros_discovering_2020} and \kw{SMI}~\cite{pennerath2020discovering} improve the robustness of information-theoretic scores under finite or sparse data, while Zhang et al.~\cite{zhang2020statistical} and \kw{DAFDiscover}~\cite{ding2024dafdiscover} target noisy or dirty settings, and the anytime framework~\cite{rana2025anytime} targets resource-constrained settings.
Efficiency-oriented approaches such as~\cite{kruse2018efficient} further improve the scalability of approximate dependency discovery.
\autorelate differs from this line of work in two aspects.
First, prior work ties reliability to score estimation, noise-robust modeling, or efficiency, whereas we characterize reliability through four properties, i.e., accuracy, atomicity, stability, and integrity, and verify them directly on the observed table through dedicated tests.
Second, their scope is primarily \FD/\AFD-style dependency discovery, whereas \autorelate 
unifies \ARs, \STs, and \FDs under a single framework. 
Berti-\'Equille et al.~\cite{berti2018discovery} study genuine \FDs under missing values; by contrast, our Independence Test targets 
incomplete candidates whose violations 
correlate with non-participative attributes.

\smallskip
\noindent\textbf{Table Transformation Discovery.}
Table transformation has been extensively studied in program synthesis, data wrangling, and spreadsheet analysis, largely under the programming-by-example (PBE) paradigm.
Systems such as \kw{FlashFill}~\cite{gulwani2011automating} and its table-level extension~\cite{harris2011spreadsheet} synthesize string transformations from input-output examples,
while \kw{Foofah}~\cite{jin2017foofah} and \kw{CLX}~\cite{jin2018clx} extend this idea to table-level operators and verifiable pattern-based transformations, respectively. \eat{\kw{Foofah}~\cite{jin2017foofah} lifts this idea to table-level operators 
and applies heuristic search over a predefined operator space, while \kw{CLX}~\cite{jin2018clx} emphasizes verifiable PBE by clustering input and output data into patterns and generating regular-expression replacement programs.
Rather than searching only a fixed operator set, \kw{DataXFormer}~\cite{abedjan2016dataxformer} leverages external sources such as web tables and web forms to discover
transformations from user-provided examples, while \kw{TDE}~\cite{he2018transform} builds an extensible search engine that reuses transformation logic from existing sources and ranks candidate transformations against user-provided examples.
\kw{Auto-Pipeline}~\cite{yang2021auto} further generalizes PBE to a by-target setting for synthesizing multi-step pipelines, while \kw{AutoPandas}~\cite{bavishi2019autopandas} synthesizes Pandas programs using neural-backed generators.}
Beyond a fixed operator set, \kw{DataXFormer}~\cite{abedjan2016dataxformer} and \kw{TDE}~\cite{he2018transform} mine transformations from external sources such as web tables and existing transformation logic, \kw{Auto-Pipeline}~\cite{yang2021auto} generalizes PBE to a by-target setting for multi-step pipelines, and \kw{AutoPandas}~\cite{bavishi2019autopandas} synthesizes Pandas programs with neural-backed generators.
From a complementary perspective, \kw{Explain-Da-V}~\cite{shraga2023explaining} explains semantic changes between dataset versions as transformation programs, and we include its conciseness and concentration metrics in our evaluation.
These methods differ from \autorelate in both interaction mode and technical objective. 
PBE-style settings rely on user-provided examples and keep users in the loop to inspect or select synthesized programs, which makes manual validation feasible.
By contrast, \autorelate operates automatically on a single observed table and aims to discover latent functional relationships at scale, 
where distinguishing genuine relationships from coincidental, redundant, or incomplete ones becomes the core challenge.
In particular, while \kw{Explain-Da-V} uses ranking heuristics such as conciseness to prioritize candidate transformations,
\autorelate explicitly verifies candidate reliability 
through dedicated tests of accuracy, atomicity, stability, and
integrity.
\looseness=-1

\smallskip
\noindent\textbf{Spreadsheet Analysis, Table Understanding, and Data Cleaning.} 
Spreadsheets are a common environment in which functional relationships arise, 
and prior work in spreadsheet analysis has studied the prevalence of spreadsheet errors~\cite{panko1998we}, formula error detection~\cite{barowy2018excelint}, cell clustering and smell detection~\cite{cheung2016custodes}, and fault prediction via product metrics~\cite{koch2019metric}.
A parallel line of work on table understanding and data profiling aims to recover the semantics and latent structure of tables, ranging from classical column annotation~\cite{limaye2010annotating} and table profiling~\cite{naumann2014data} to learned column-type inference such as \kw{Sherlock}~\cite{hulsebos2019sherlock} and \kw{Sato}~\cite{zhang2020sato}, and more recent table pre-training and language-model-based approaches~\cite{deng2020turl,suhara2022doduo}.
In data cleaning~\cite{abedjan2016detecting}, dependencies and derived-value relationships have been used as constraints for error detection and repair, e.g., in \kw{HoloClean}~\cite{rekatsinas2017holoclean}, \kw{Raha}~\cite{mahdavi2019raha}, and \kw{Baran}~\cite{mahdavi2020baran}, while related work on data quality~\cite{data-cleaning} and provenance~\cite{provenance-1,provenance-2} relies on similar inter-column relationships for consistency maintenance and lineage tracing.
By contrast, \autorelate treats reliable \FR discovery as the primary problem, providing a unified framework for \ARs, \STs, and \FDs from raw tables, with explicit reliability verification to filter spurious candidates.
\looseness=-1

\eat{In parallel, work on table understanding and data profiling aims to recover the semantics and latent structure of tables, for example by annotating columns with entities, types, and relationships~\cite{limaye2010annotating} or by profiling column- and table-level characteristics for downstream analysis~\cite{naumann2014data}. 
Related motivations also arise in
data quality, consistency maintenance, and provenance research, where dependencies and derived-value relationships can support error detection, repair, and lineage tracing~\cite{data-cleaning,fan2008conditional,provenance-1,provenance-2}.
These studies highlight the importance of inter-column relationships for understanding, validating, and tracing tabular data, but they do not directly study the problem of reliable functional relationship discovery.
Most of them either analyze formulas or table structure that is already present, or use relationships as signals for downstream tasks such as table interpretation, error detection, repair, or provenance analysis.
By contrast, \autorelate treats the reliable discovery of functional relationships itself as the primary problem, and provides a unified framework for discovering arithmetic relationships, string transformations, and functional dependencies directly from raw tables while explicitly filtering spurious candidates through reliability verification.
}
\section{Conclusion}
\label{sec:conclusion}

We proposed \autorelate, a unified framework for discovering reliable functional relationships from tables.
Its main contributions include: 
(1) a unified notion of functional relationships covering arithmetic relationships, string transformations, and functional dependencies;
(2) four reliability criteria, including accuracy, atomicity, stability, and integrity, for distinguishing genuine relationships from spurious ones;
(3) a mine-then-verify framework that first generates candidate relationships and then verifies their reliability through a Minimality Test, a Perturbation Test, and an Independence Test; 
(4) three optimization strategies for efficient reliability verification; and 
(5) a large-scale benchmark suite constructed from 58,679 real-world spreadsheets and relational tables and containing 6,414 ground-truth instances across all three types.
Extensive experiments against 18 baselines 
demonstrate the effectiveness and efficiency of \autorelate, which achieves an average \prauc of 0.87, 
59\% higher than the best competing baseline across all settings.
\looseness=-1



\clearpage

\bibliographystyle{ACM-Reference-Format}
\bibliography{reference}


\end{document}